\documentclass[11pt]{article}
\usepackage[a4paper,margin=1in]{geometry}
\usepackage[latin1]{inputenc}
\usepackage{latexsym}
\usepackage{amsmath}
\usepackage{amsthm}
\usepackage{amssymb}
\usepackage{graphicx}
\usepackage{ifthen}
\usepackage{calc}
\usepackage{stmaryrd}
\usepackage{enumitem}
\usepackage{hyperref}
\usepackage{alphalph}

\usepackage{graphicx}
\usepackage{tikz}
\usetikzlibrary{arrows}
\usetikzlibrary{shapes.geometric}
\usetikzlibrary{calc}
\usetikzlibrary{automata}
\usetikzlibrary{petri}
\usetikzlibrary{decorations.pathmorphing}

\makeatletter
\renewcommand{\section}{\@startsection
  {section}%
  {1}%
  {0mm}%
  {-1\baselineskip}%
  {0.5\baselineskip}%
  {\normalfont\large\bfseries}%
}
\renewcommand{\subsection}{\@startsection
  {subsection}%
  {2}%
  {0mm}%
  {-1\baselineskip}%
  {0.5\baselineskip}%
  {\normalfont\large\itshape}%
}

\renewcommand{\subsubsection}{\@startsection
  {subsubsection}%
  {3}%
  {0mm}%
  {-1\baselineskip}%
  {0.5\baselineskip}%
  {\normalfont\itshape}%
}

\newsavebox{\tempbox}
\renewcommand{\@makecaption}[2]{
  \vspace{10pt}
  \sbox{\tempbox}{\textbf{#1.} #2}
  \ifthenelse{\lengthtest{\wd\tempbox > \linewidth}}{
    \textbf{#1.} #2\par
  }{
    \begin{center}
      \textbf{#1.} #2
    \end{center}
  }
}

\makeatother

\numberwithin{equation}{section}

\newtheoremstyle{mythm}%
  {}%
  {}%
  {\itshape}%
  {}%
  {\bfseries}%
  {.}%
  {.5em}%
  {\thmname{#1}~\thmnumber{#2}\ifthenelse{\equal{\thmnote{#3}}{}}{}{~(\thmnote{#3})}}%

\newtheoremstyle{mydefn}%
  {}%
  {}%
  {\upshape}%
  {}%
  {\bfseries}%
  {.}%
  {.5em}%
  {\thmname{#1}~\thmnumber{#2}\ifthenelse{\equal{\thmnote{#3}}{}}{}{~(\thmnote{#3})}}%

\newtheoremstyle{myremark}%
  {}%
  {}%
  {\upshape}%
  {}%
  {\itshape}%
  {.}%
  {.5em}%
  {\thmname{#1}~\thmnumber{#2}\ifthenelse{\equal{\thmnote{#3}}{}}{}{~(\thmnote{#3})}}%

\theoremstyle{mythm}
\newtheorem{theo}{Theorem}[section]
\newtheorem{lem}[theo]{Lemma}

\newtheorem{cor}[theo]{Corollary}

\theoremstyle{mydefn}

\newtheorem{ass}[theo]{Assumption}
\newtheorem{asss}[theo]{Assumptions}
\theoremstyle{myremark}
\newtheorem{rem}[theo]{Remark}
\theoremstyle{mythm}

\newcommand{\uend}{\hfill$\lrcorner$}
\newcommand{\uende}{\eqno\lrcorner}

\newcounter{claimcounter}
\newenvironment{claim}[1][]{
  \renewcommand{\proof}{\smallskip\par\noindent\textit{Proof. }}
  \medskip\par\noindent%
  \ifthenelse{\equal{#1}{}}{%
    \setcounter{claimcounter}{0}\refstepcounter{claimcounter}\textit{Claim~\arabic{claimcounter}.}
  }{%
    \ifthenelse{\equal{#1}{resume}}{%
      \refstepcounter{claimcounter}\textit{Claim~\arabic{claimcounter}.}
    }{%
      \textit{Claim~#1.}
    }
  }
}{
  \par\medskip
}

\newcommand{\case}[1]{\par\medskip\noindent\textit{Case #1: }}

\newenvironment{cs}{
  \begin{description}
    \renewcommand{\case}[1]{\item[\itshape\mdseries Case ##1:]}
  }{
  \end{description}
}

\newlist{caselist}{description}{10}
\setlist[caselist]{font=\itshape\mdseries}

\setenumerate[1]{label=(\arabic*)}
\newlist{eroman}{enumerate}{2}
\setlist[eroman,1]{label=(\roman*)}
\setlist[eroman,2]{label=(\alph*)}
\newlist{ealph}{enumerate}{1}
\setlist[ealph]{label=(\Alph*)}

\newcounter{nlistcounter}
\newenvironment{nlist}[1]{
  \renewcommand{\thenlistcounter}{\upshape(#1.\arabic{nlistcounter})}
  \begin{list}{\bfseries\thenlistcounter}{%
      \usecounter{nlistcounter}
      \setlength{\labelwidth}{1.5em}%
      \setlength{\leftmargin}{\labelwidth}%
      \addtolength{\leftmargin}{\labelsep}%
      \setlength{\listparindent}{0em}%
      \setlength{\topsep}{5pt}%
      \setlength{\itemsep}{5pt}%
      \setlength{\parsep}{0pt}%
    }
  }{
  \end{list}
}

\renewcommand{\phi}{\varphi}
\newcommand{\bigmid}{\;\big|\;}

\renewcommand{\mathbf}[1]{\textit{\bfseries #1}}

\renewcommand{\tilde}{\widetilde}
\renewcommand{\hat}{\widehat}
\renewcommand{\bar}{\overline}
\renewcommand{\vec}{\overrightarrow}

\newcommand{\angles}[1]{\langle#1\rangle}

\newcommand{\NN}{{\mathbb N}}
\newcommand{\RR}{{\mathbb R}}

\newcommand{\CD}{{\mathcal D}}

\newcommand{\CI}{{\mathcal I}}
\newcommand{\CJ}{{\mathcal J}}

\newcommand{\CL}{{\mathcal L}}
\newcommand{\CM}{{\mathcal M}}

\newcommand{\CP}{{\mathcal P}}

\newcommand{\CT}{{\mathcal T}}

\newcommand{\CX}{{\mathcal X}}
\newcommand{\CY}{{\mathcal Y}}
\newcommand{\CZ}{{\mathcal Z}}

\newcommand{\KT}{{\mathfrak T}}

\usepackage{color}
\definecolor{gruen}{rgb}{0,0.6,0.2}

\newcounter{rbcounter}
\newcommand{\ord}{\operatorname{ord}}

\newcommand{\bw}{\operatorname{bw}}
\newcommand{\rk}{\operatorname{rk}}

\newcommand\contract{\mathord{\downarrow}}
\newcommand\expand{\mathord{\uparrow}}

\newcommand{\width}{\operatorname{wd}}

\newcommand{\dagle}{\trianglelefteq}

\newcommand{\dagri}{\trianglerighteq}

\DeclareMathOperator{\sep}{\textsc{Sep}}
\DeclareMathOperator{\order}{\textsc{Order}}
\DeclareMathOperator{\tangorder}{\textsc{TangOrd}}
\DeclareMathOperator{\size}{\textsc{Size}}

\DeclareMathOperator{\trunc}{\textsc{Trunc}}
\DeclareMathOperator{\find}{\textsc{Find}}

\newcommand{\Mat}{P}
\DeclareMathOperator{\Ext}{Ext}
\newcommand{\Type}{W}
\DeclareMathOperator{\ISO}{Iso}
\DeclareMathOperator{\Iso}{Iso}
\DeclareMathOperator{\Aut}{Aut}

\newcommand{\Sym}{\operatorname{Sym}}
\newcommand{\Act}{A\contract}
\newcommand{\kct}{\kappa\contract}

\begin{document}
\title{Isomorphism Testing for Graphs of Bounded Rank Width}
\author{Martin Grohe and Pascal Schweitzer\\RWTH Aachen
  University\\\normalsize\{grohe,schweitzer\}@informatik.rwth-aachen.de}
\date{}
\maketitle

\begin{abstract}
  We give an algorithm that, for every fixed $k$, decides isomorphism
  of graphs of rank width at most $k$ in polynomial time. As the
  clique width of a graph is bounded in terms of its rank width, we
  also obtain a polynomial time isomorphism test for graph classes of
  bounded clique width.
\end{abstract}

\section{Introduction}
Rank width, introduced by Oum and Seymour~\cite{oumsey06}, is a graph
invariant that measures how well a graph can be recursively decomposed
along ``simple separations''. In this sense, it resembles tree width,
but it fundamentally differs from tree width in how the ``simplicity''
of a separation is measured: for rank width, the idea is to take
the row rank (over the field $\mathbb F_2$) of the matrix that records the
adjacencies between the two parts of a separation, whereas for tree
width one simply counts how many vertices the two parts have in
common. Rank width is bounded in terms of tree
width, but not vice versa. For example, the complete graph $K_n$ has
rank width $1$ and tree width $n-1$. This also shows that graphs of
bounded rank width are not necessarily sparse (as opposed to graphs of
bounded tree width). An interesting aspect of rank width when dealing
with problems like graph isomorphism testing (or various problems related to logical
definability) that make no real distinction between the edge relation
and the ``non-edge relation'' of a graph is that the rank width of a graph and
its complement differ by at most one. Another well-known graph
invariant is clique width \cite{couola00}; it
measures how many labels are needed to generate a graph in a certain
grammar. Rank width is equivalent to
clique width, in the sense that each of the two invariants is bounded in terms of
the other \cite{oumsey06}. As for bounded tree width, many hard algorithmic problems can
be solved in polynomial time (often cubic time) on graph classes of
bounded rank width, or equivalently, bounded clique width (e.g.~\cite{coumakrot01,espegurwan01,fismakrav08,kobrot03}). However,
until now it was open whether the isomorphism problem is among them.

We give an algorithm that, for every fixed $k$, decides isomorphism
of graphs of rank width at most $k$ in polynomial time. Many of the
best known graph classes where the isomorphism problem is known to be
in polynomial time are classes of sparse graphs \cite{hoptar72,filmay80,mil80,luk82,pon88,bod90,gromar15}, among them
planar graphs, graphs of bounded degree, and
graphs of bounded tree width. Less is known for dense
graphs; among the known results are polynomial time isomorphism tests
for classes with bounded eigenvalue multiplicities \cite{babgrimou82} and various
hereditary graph classes, specifically classes intersection graphs \cite{CurLinMcC13,KoblerKV13}, among them interval graphs~\cite{LuekerB79},
and classes defined by excluding specific induced subgraphs
\cite{boothcolbourn,KratschS12,Schweitzer15}. Our result substantially extends the realm of hereditary
graph classes with a tractable isomorphism problem. While it subsumes
several known results \cite{bod90,Corneil1981163,LimouzyMR07,Schweitzer15}, for the classes of clique width at most~$k$ a polynomial time isomorphism algorithm was only known for the case~$k \leq  2$ (\cite{LimouzyMR07}).

Technically, we found the isomorphism problem for bounded rank width
graphs much harder than anticipated. The overall proof strategy
is generic: first compute a canonical decomposition of a graph, or if
that is impossible, a canonical family of decompositions with a compact representation, and
then use dynamic programming to solve the isomorphism problem. Indeed,
this is the strategy taken for bounded tree width graphs in
\cite{bod90,lokpilpil+14}. However, for graphs of bounded rank width,
both steps of this general strategy turned out to be difficult to
implement. To compute canonical decompositions, we heavily rely on the
general theory of connectivity functions, branch decompositions, and
tangles \cite{gm10,geegerwhi09}, and in particular on computational aspects of the theory
recently developed in \cite{grosch15}. Our starting point is an
algorithm for canonically decomposing a connectivity function into highly
connected regions described by maximal tangles \cite{grosch15}. The
technical core of the first part of this paper is a decomposition of these highly
connected regions into pieces of bounded width
(Lemma~\ref{lem:nodedec}). It has been slightly disturbing to find
that even with a canonical decomposition given, the isomorphism
problem is still nontrivial and requires a complicated (though
elementary) group
theoretic machinery. The intuitive reason for this can be explained by
a comparison with bounded tree width. In a bounded-width tree
decomposition of a graph, we have low order vertex separations of the graph, and after
removing the separating vertices (a bounded number) we can deal with
the two
parts of a separation independently. In
a bounded-rank-width decomposition, we have partitions of the graph into
two parts such that the adjacency matrix between these parts has low
rank. For such a partition, removing a bounded number of vertices
shows no effect. Instead, we need to fix a bounded number of rows and
columns in the matrix, but even then there is a nontrivial interaction
between the two parts, which fortunately we can capture group theoretically.

The paper is organised as follows: after reviewing the necessary
background in Section~\ref{sec:tangles}, in the short
Section~\ref{sec:cover}, we show that all tangles of a connectivity
function have ``triple covers'' of bounded size, providing another
technical tool for dealing with tangles (which may be of independent
interest).
In
Section~\ref{sec:treelike}, we introduce treelike decompositions of
connectivity functions, which may be viewed as compact
representations of families of tree decompositions. Sections~\ref{sec:1tan}--\ref{sec:ctl} are devoted to a
proof of the canonical decomposition theorem
(Theorem~\ref{theo:candec}). In Section~\ref{sec:part:rank}, we describe the
situation at a single node of our decomposition and its children in
matrix form and introduce the notion of partition rank of the matrix
to capture the width of the decomposition at this node.
Finally, in Sections~\ref{sec:iso} we
develop the group theoretic machinery and give the actual isomorphism
algorithm.

Throughout this paper, we often speak of ``canonical'' constructions. The
precise technical meaning depends on the context, but in general a
construction (or algorithm) is \emph{canonical} if every isomorphism between its input objects commutes with an isomorphism between the output objects.

\section{Connectivity Functions, Tangles, and Branch Decompositions}
\label{sec:tangles}

A \emph{connectivity function} on a finite set $A$ is a symmetric and
submodular function $\kappa\colon 2^A\to\NN$ with
$\kappa(\emptyset)=0$. \emph{Symmetric} means that
$\kappa(X)=\kappa(\bar X)$ for all $X\subseteq A$; here and whenever
the ground set $A$ is clear from the context we write $\bar X$ to
denote $A\setminus X$, the complement of~$X$. \emph{Submodular} means
that $\kappa(X)+\kappa(Y)\ge\kappa(X\cap Y)+\kappa(X\cup Y)$ for all
$X,Y\subseteq A$. Observe that a symmetric and submodular set function
is also \emph{posimodular}, that is, it satisfies
$\kappa(X)+\kappa(Y)\ge\kappa(X\setminus Y)+\kappa(Y\setminus X)$
(apply submodularity to $X$ and $\bar Y$).

The only connectivity function that we consider in this paper is the
\emph{cut rank} function $\rho_G$ of a graph $G$. For all subsets
$X,Y\subseteq V(G)$, we let $M_{X,Y}$ be the $X\times Y$-matrix over
the 2-element field $\mathbb F_2$ with entries
$m_{xy}=1\iff xy\in E(G)$. We define $\rho_G:2^{V(G)}\to\NN$ by
letting $ \rho_G(X) $ be the row rank of the matrix $M_{X,Y}$ over
$\mathbb F_2$. It is not hard to prove that $\rho_G$ is indeed a
connectivity function.

For the rest of this section, let $\kappa$ be a connectivity function
on a finite set $A$.
We often think of a subset $Z\subseteq A$ as a \emph{separation} of $A$
into $Z$ and $\bar Z$ and of $\kappa(Z)$ as the \emph{order} of this
separation; consequently, we also refer to $\kappa(Z)$ as the
\emph{order of $Z$}. For disjoint sets $X,Y\subseteq A$, an
\emph{$(X,Y)$-separation} is a set $Z\subseteq A$ such that
$X\subseteq Z\subseteq \bar Y$. Such a separation $Z$ is minimum if
its order is minimum, that is, if $\kappa(Z)\le\kappa(Z')$ for all
$(X,Y)$-separations $Z'$. It is an easy consequence of the
submodularity of $\kappa$ that there is a unique minimum
$(X,Y)$-separation $Z$ such that $Z\subseteq Z'$ for all other minimum
$(X,Y)$-separations $Z'$. We call $Z$ the \emph{leftmost minimum
  $(X,Y)$-separation}. There is also a unique \emph{rightmost minimum
  $(X,Y)$-separation}, which is easily seen to be the complement of
the leftmost minimum $(Y,X)$-separation.

A \emph{$\kappa$-tangle} of order $k\ge0$ is a set $\CT\subseteq 2^A$
satisfying the following conditions.
  \begin{nlist}{T}
  \setcounter{nlistcounter}{-1}
  \item\label{li:t0}
    $\kappa(X)<k$ for all $X\in\CT$, 
  \item\label{li:t1}
    For all $X\subseteq A$ with $\kappa(X)<k$, either $X\in\CT$ or
    $\bar X\in\CT$.
  \item\label{li:t2}
    $X_1\cap X_2\cap X_3\neq\emptyset$ for all $X_1,X_2,X_3\in\CT$.
  \item\label{li:t3}
    $\CT$ does not contain any singletons, that is, $\{a\}\not\in\CT$ for all $a\in A$.
\end{nlist}
We denote the order of a $\kappa$-tangle $\CT$ by $\ord(\CT)$.\footnote{There is a small technical issue that one needs to be aware of, but
that never causes any real problems: if we view tangles as families of
sets, then their order is not always well-defined. Indeed, if there is
no set $X$ of order $\kappa(X)=k-1$, then every tangle of order $k$ is
equal to its truncation to order $k-1$. In such a situation, we
have to explicitly annotate a tangle with its order, formally viewing a
tangle as a pair $(\CT,k)$ where $\CT\subseteq 2^A$ and $k\ge 0$.}

Let
$\CT,\CT'$ be $\kappa$-tangles. If $\CT'\subseteq\CT$, we say that
$\CT$ is an \emph{extension} of $\CT'$. The tangles $\CT$ and $\CT'$
are \emph{incomparable} (we write $\CT\bot\CT'$) if neither is an
extension of the other. 
The \emph{truncation} 
of $\CT$ to order $k\le\ord(\CT)$ is the set
$
\{X\in\CT\mid\kappa(X)<k\},
$
which is obviously a tangle of order $k$. Observe that if $\CT$ is
an extension of $\CT'$, then $\ord(\CT')\le\ord(\CT)$, and $\CT'$ is
the truncation of $\CT$ to order $\ord(\CT')$. 

A $\kappa$-tangle $\CT$ is \emph{maximal} if there is no
$\kappa$-tangle $\CT'\subseteq\CT$ with $\ord(\CT')>\ord(\CT)$.
A $\kappa$-tangle $\CT$ is \emph{$\ell$-maximal}, for some $\ell\ge 0$,
if either $\ord(\CT)=\ell$ or $\CT$ is maximal.

A \emph{$(\CT,\CT')$-separation} is a set $Z\subseteq A$ such that
$Z\in\CT$ and $\bar Z\in\CT'$. Obviously, if $Z$ is a
$(\CT,\CT')$-separation then $\bar Z$ is a
$(\CT',\CT)$-separation. Observe that there is a
$(\CT,\CT')$-separation if and only if $\CT$ and $\CT'$ are
incomparable. The \emph{order} of a $(\CT,\CT')$-separation $Z$ is
$\kappa(Z)$. A $(\CT,\CT')$-separation $Z$ is \emph{minimum} if its
order is minimum. It can be shown \cite{grosch15} that if
$\CT\bot\CT'$ then there is a unique minimum
  $(\CT,\CT')$-separation $Z$ such that
  $Z\subseteq Z'$ for all minimum $(\CT,\CT')$-separations
  $Z'$. We call $Z$ the \emph{leftmost minimum $(\CT,\CT')$-separation}.
Of course there is also a  \emph{rightmost minimum
  $(\CT,\CT')$-separation}, which is the complement of the leftmost
minimum $(\CT',\CT)$-separation.

Now that we have defined
$(X,Y)$-separations for sets $X,Y$ and $(\CT,\CT')$-separations for
tangles $\CT,\CT'$, we also need to define combinations of both. For a
$\kappa$-tangle $\CT$ and a set $X\subseteq A$ such that
$X\not\in\CT$, a \emph{$(\CT,X)$-separation} is a set $Z\in\CT$ such
that $Z\subseteq\bar X$. A $(\CT,X)$-separation is \emph{minimum} if
its order is minimum, and again it can be proved that if there is a
$(\CT,X)$-separation, then there is a
unique \emph{leftmost minimum $(\CT,X)$-separation} and a
\emph{rightmost minimum $(\CT,X)$-separation}. Analogously, we define (leftmost, rightmost minimum) $(X,\CT)$-separations.

\begin{lem}\label{lem:tanglesep3}
  Let $\CT,\CT'$ be $\kappa$-tangles and $X,Y\subseteq A$ such that neither
  $Y\subseteq X$ nor $\bar Y\subseteq X$.
  \begin{enumerate}
  \item If $X$ is a minimum $(\CT,\CT')$-separation, then 
    $\kappa(X\cap Y)\le\kappa(Y)$ or
    $\kappa(X\cap \bar Y)\le\kappa(Y)$.
  \item If $X$ a rightmost minimum
  $(\CT,\CT')$-separation, then 
  $\kappa(X\cap Y)<\kappa(Y)$ or $\kappa(X\cap \bar Y)<\kappa(Y)$.

  \end{enumerate}
\end{lem}

\begin{proof}
  Part (1) is Lemma~4.9 of \cite{grosch15}.
We only prove (2). (The proof of (1) is similar.)

  Suppose that $\kappa(X\cap Y)\ge\kappa(Y)$ and
  $\kappa(X\cap\bar Y)\ge\kappa(Y)$. By submodularity,
  $\kappa(X\cup Y)\le\kappa(X)$ and $\kappa(X\cup\bar
  Y)\le\kappa(X)$. Then $X\cup Y,X\cup\bar Y\in\CT$, because both sets
  are supersets of $X$. Since $\bar X\cap (X\cup Y)\cap(X\cup\bar
  Y)=\emptyset$, either $\bar{X\cup Y}\in\CT'$ or $\bar{X\cup \bar
    Y}\in\CT'$. If $\bar{X\cup Y}\in\CT'$, then $X\cup Y$ is a
  $(\CT,\CT')$-separation whose order is at most the order of $X$, and as
  $X$ is a rightmost minimum $(\CT,\CT')$-separation, it follows that
  $X\cup Y\subseteq X$ and thus $Y\subseteq X$. Similarly, if
  $\bar{X\cup Y}\in\CT'$, then $\bar Y\subseteq X$.
\end{proof}

The last concept we need to define is that of branch decompositions
and branch width of a connectivity function. A \emph{cubic tree} is a
tree where every node that is not a leaf has degree~$3$. An
\emph{oriented edge} of a tree $T$ is a pair $(s,t)$, where $st\in
E(T)$. We denote the set of all oriented edges of $T$ by $\vec E(T)$
and the set of leaves of $T$ by $L(T)$.  A \emph{branch
  decomposition} of $\kappa$ is a pair $(T,\xi)$, where $T$ is a cubic
tree and $\xi\colon L(T)\to A$ is a bijective mapping. For every
oriented edge
$(s,t)\in\vec E(T)$, we let $\tilde\xi(s,t)\subseteq A$ be the set of
all $\xi(u)$ where $u$ is a leaf in the component of $T-\{st\}$ that
contains $t$ (so the oriented edge $(s,t)$ points towards
$u$). Observe that $\tilde\xi(s,t)=\bar{\tilde\xi(t,s)}$. We
define the \emph{width} of the branch decomposition $(T,\xi)$ to be 
\[
\width(T,\xi):=\max\{\kappa(\tilde\xi(s,t))\mid(s,t)\in\vec E(s,t)\}.
\]
The \emph{branch width} $\bw(\kappa)$ of $\kappa$ is the minimum of
the width of all branch decompositions of $\kappa$. The \emph{rank
  width} of a graph $G$ is defined to be the branch width of the cut
rank function $\rho_G$.

\begin{theo}[Duality Theorem~\cite{gm10}]\label{theo:duality}
  The branch width of $\kappa$ is exactly the maximum order of a $\kappa$-tangle.
\end{theo}

For disjoint sets~$X,Y\subseteq A$ we define~$
\kappa_{\min}(X,Y):=\min\{\kappa(Z)\mid X\subseteq Z\subseteq
\overline{Y}\}$. Note that for all $X,Y$ the two functions
$X'\mapsto\kappa_{\min}(X',Y)$ and $Y'\mapsto\kappa_{\min}(X,Y')$ are
monotone and submodular.

For sets $Y\subseteq X$, we say that a set $Y$ is \emph{free} in $X$ if
$\kappa_{\min}(Y,\bar X) = \kappa(X)$ and~$|Y| \leq  \kappa(X)$. It
can be shown that for every~$X\subseteq A$ there is a set~$Y$ that is
free in $X$ \cite{oumsey07,grosch15}.

\subsection{Computing with Tangles}
Algorithms expecting a set function $\kappa:2^A\to\NN$ as input are
given the ground set $A$ as actual input (say, as a list of objects),
and they are given an oracle that returns for~$X\subseteq A$ the value
of~$\kappa(X)$. The running time of such algorithms is measured in
terms of the size $|A|$ of the ground set. We assume this computation
model for all algorithms dealing with abstract connectivity functions $\kappa$. Of course, if $\kappa=\rho_G$ is the cut rank function
of a graph $G$, then we assume a standard computation model (without
oracles), where the graph $G$ is given as input; we can use $G$ to
simulate oracle access to $\rho_G$.
 
An important fact underlying most of our algorithms is that, under this model of
computation, submodular functions can be efficiently minimised
\cite{iwaflefuj01,schrijver00}.

 In \cite{grosch15}, we introduced a data structure for representing
all tangles of a graph up to a certain order. A \emph{comprehensive tangle data structure} of order~$k$ for a connectivity function~$\kappa$ over a set~$A$ is a data structure~$\mathcal{D}$ with functions~$\order_\mathcal{D}$, $\size_\mathcal{D}$, $\mathcal{T}_\mathcal{D}$, $\tangorder_\mathcal{D}$, $\trunc_\mathcal{D}$, $\sep_\mathcal{D}$, and~$\find_\mathcal{D}$ that provide the following functionalities.
\begin{enumerate}
\item The function~$\order_\mathcal{D}()$ returns the fixed integer~$k$.
\item For~$\ell\in[k]$ the function~$\size_\mathcal{D}(\ell)$ returns the number of~$\kappa$-tangles of
  order at most~$\ell$. We denote the number of~$\kappa$-tangles of
    order at most~$k$ by~$|\mathcal{D}|$.
\item For each~$i\in \big[|\CD|\big]$ the
  function~$\mathcal{T}_\mathcal{D}(i,\cdot)\colon 2^A \rightarrow
  \{0,1\}$ is a tangle~$\mathcal{T}_i$ of order
  at most~$k$, (i.e., the function call~$\mathcal{T}_\mathcal{D}(i,X)$
  determines whether~$X\in  \mathcal{T}_i$).

  We call $i$ the \emph{index} of the tangle $\CT_i$ within the data
  structure.
\item For~$i\in \big[|\CD|\big]$ the call~$\tangorder_\mathcal{D}(i)$ returns~$\ord(\mathcal{T}_i)$. %
\item For~$i\in \big[|\CD|\big]$ and~$\ell \leq \ord(\CT_i)$ the call~$\trunc_\mathcal{D}(i,\ell)$ returns an integer~$j$ such that~$\mathcal{T}_j$ is the truncation of~$\mathcal{T}_i$ to order~$\ell$. If~$\ell >\ord(\CT_i)$ the function returns~$i$.
\item For distinct~$i, j \in \big[|\CD|\big]$ the call~$\sep_\mathcal{D}(i,j)$ outputs a set~$X\subseteq A$
  such that~$X$ is the leftmost
  minimum~$(\mathcal{T}_i,\mathcal{T}_j)$-separation
or states that no such set exists (in which case one of the tangles is a truncation of the other).
\item Given~$\ell \in\{0,\ldots,k\}$ and a tangle~$\mathcal{T}'$ of order~$\ell$ (via a
  membership oracle) the function~$\find_\mathcal{D}(\ell,\mathcal{T'})$,
  returns the index of $\CT'$, that is, the unique integer~$i \in \big[|\mathcal{D}|\big]$ such that~$\ord(\mathcal{T}_i) = \ell$ and~$\mathcal{T}' = \mathcal{T}_i$.
\end{enumerate}

\begin{theo}[\cite{grosch15}]\label{theo:ds}
For every constant~$k$ there is a polynomial time algorithm that,
given oracle access to a connectivity function $\kappa$, computes an efficient
comprehensive tangle data  structure of order~$k$.
\end{theo}

Using a comprehensive tangle data structure, we can design polynomial
time algorithms
for other computational problems related to tangles. 

\begin{lem}
  Let $k\ge 0$.
  \begin{enumerate}
  \item There is a polynomial time algorithm that, given a set
    $X\subseteq A$ and a tangle $\CT$ of order $k$ (via its index in a comprehensive
    tangle data structure), computes the leftmost minimum
    $(\CT,X)$-separation if it exists or reports that there is no $(\CT,X)$-separation.
  \item There is a polynomial time algorithm that, given a tangle $\CT$
    of order $k$
    (via its index in a comprehensive tangle data structure), computes
    a list of all inclusionwise minimal elements of $\CT$.
  \end{enumerate}
\end{lem}

\begin{proof}
  Assertion (1) follows from Lemma~2.20 of \cite{grosch15}.

  To prove (2), we claim that a set $X\in A$ is an inclusionwise
  minimal element of a tangle $\CT$ of order $k$ if and only if the
  following two conditions are satisfied.
  \begin{eroman}
  \item There is a set $Y\subseteq \bar X$ of size $|Y|\le k$ such that
    $X$ is the leftmost minimum $(\CT,Y)$-separation.
  \item There is a no set $Z\subseteq A$ of size $|Z|\le k$ such that
    the leftmost minimum $(\CT,Z)$-separation is a proper subset of
    $X$.
  \end{eroman}
  To see this, we simply observe that if $X$ is an inclusionwise
  minimal element of $\CT$, then it trivially satisfies (ii), and it
  satisfies (i), because we can let $Y$ be a set that is free in $\bar
  X$. Conversely, if $X$ satisfies (i) then it is an element of $\CT$,
  and (ii) makes sure that it is inclusionwise minimal.

  There are at most $\binom{|A|}{k}$ sets $X$ satisfying (i), and
  using (1) we can list these in polynomial time. Then, using (1)
  again, for each of these sets we can check whether they satisfy (ii).
\end{proof}

\subsection{Contractions}
\label{sec:contraction}

Contractions give a  way to construct new connectivity functions from
given ones. To define a
\emph{contraction}, we take one or several disjoint subsets of the
ground set and ``contract'' these sets to single points. In the new
decomposition, these new points represent the sets of the original
decomposition

For the formal treatment, let $\kappa$ be a connectivity
function on a set $A$.

 Let $C_1,\ldots,C_m\subseteq A$ be mutually disjoint subsets of
 $A$. Let $B:=A\setminus(C_1\cup\ldots\cup C_m)$, and let
 $c_1,\ldots,c_m$ be fresh elements (mutually distinct, and distinct
 from all elements of $B$). We define
\[
A\contract_{C_1,\ldots,C_m}:=B\cup\{c_1,\ldots,c_m\}.
\]
To simplify the notation, here and in the following we omit the index
${}_{C_1,\ldots,C_m}$ if the sets $C_i$ are clear from the context.
For every subset $X\subseteq A\contract$, we define its
\emph{expansion} to be the set 
\[
X\expand:=X\expand_{C_1,\ldots,C_m}:=(X\cap B)\cup\bigcup_{\substack{i\in[m]\\c_i\in
    X}}C_i.
\]
The \emph{$C_1,\ldots,C_m$-contraction} of $\kappa$ is the
function $\kappa\contract$, or $\kappa\contract _{C_1,\ldots,C_m}$, on
$2^{A\contract}$ defined by
\[
\kappa\contract(X):=\kappa(X\expand).
\]
It is easy to verify that $\kappa\contract$ is indeed a connectivity
function.

\begin{rem}
  A different view on contractions is to maintain the 
  ground set, but define the connectivity function on a
  sublattice of the power set lattice. That is, not all separations of the ground set get an
  order, but only some of them. 

  Formally, we let $\CL:=\CL(A\contract C_1,\ldots,C_m)$ be the sublattice
  of $\CP(A):=(2^A,\cap,\cup)$ consisting of all sets $X\subseteq A$ such that
  for all $i\in[m]$ either $C_i\subseteq X$ or $C_i\subseteq\bar
  X$. Obviously, $\CL$ is closed under
  intersection and union and thus indeed a sublattice. Observe that
  every $X\in \CL$ has a natural
  contraction
  \[
  X\contract:=(X\cap B)\cup\{c_i\mid i\in[m]\text{ with }C_i\subseteq
  X\},
  \]
  and we have $X\contract\expand=X$. As we also have $X'\expand\in
  \CL$ for all $X'\subseteq A\contract$,
  the contraction mapping is a bijection between $\CL(A\contract
  C_1,\ldots,C_m)$ and $A\contract$. It follows immediately from the
  definition of $\kappa\contract$ that for all $X\in \CL(A\contract
  C_1,\ldots,C_m)$ we have
  \[
  \kappa(X)=\kappa\contract(X\contract).
  \]
  Thus the contraction mapping is an isomorphism from the
  \emph{connectivity system} $\big(\CL,\kappa|_{\CL})$, where
  $\kappa|_{\CL}$ denotes the restriction of $\kappa$ to $\CL$, and
  the connectivity system $(2^{A\contract},\kappa\contract)$.

  The view of a contraction of $\kappa$ as a restriction to a
  sublattice will be useful when dealing with contractions of the
  cut-rank function of a graph in Section~\ref{sec:nwl2}.
  \uend
\end{rem}

Let $\CT$ be a $\kappa$-tangle of order $k$. We define 
\[
\CT\contract:=\CT\contract_{C_1,\ldots,C_m}:=\{X\subseteq
A\contract\mid X\expand\in\CT\}.
\]
Note that $\CT\contract$ is not necessarily a
$\kappa\contract$-tangle: if $C_i\in\CT$ for some $i\in[m]$, then
$\{c_i\}\in\CT\contract$, and thus $\CT\contract$ violates \ref{li:t3}.
However, it is straightforward to verify that $\CT\contract$ is a
$\kappa\contract$-tangle (of the same order $k$) if and only if
$C_1,\ldots,C_m\not\in\CT$. 

\section{Triple Covers}
\label{sec:cover}

A \emph{cover} of a $\kappa$-tangle $\CT$ is a set $C\subseteq A$ such
that $C\cap Y\neq\emptyset$ for all $Y\in\CT$. It is not hard to prove
that every $\kappa$-tangle of order $k$ has a cover of size at most
$k$. 

A \emph{triple cover} of a  $\CT$ is a set $Q\subseteq A$ such
that $Q\cap Y_1\cap Y_2\cap Y_3\neq\emptyset$ for all $Y_1,Y_2,Y_3\in\CT$. We
shall prove that every tangle of order $k$ has a triple cover of size
bounded in terms of $k$.

Observe that we can test in polynomial time whether a given set $Q$ is
a triple cover for a $\kappa$-tangle $\CT$, given by its index in a
comprehensive tangle data structure: using the data structure, we
produce a list of all inclusionwise minimal elements of $\CT$, and
then we check if any three of them have a nonempty intersection with
$Q$.

Let $\theta:\NN\to\NN$ be defined by $\theta(0):=0$ and 
\[
\theta(i+1):=\theta(i)+3^{\theta(i)}.
\]

\begin{lem}\label{lem:triple-cover}
  Let $\CT$ be a $\kappa$-tangle of order $k$. Then $\CT$ has a triple
  cover of size at most $\theta(3k-2)$.
\end{lem}

\begin{proof}
  By induction on $i\ge 0$ we construct sets $Q_i$ such that for all
  $Y_1,Y_2,Y_3\in\CT$, if $Q_i\cap Y_1\cap Y_2\cap Y_3=\emptyset$ then
  $\kappa(Y_1)+\kappa(Y_2)+\kappa(Y_3)\ge i$. Then $Q:=Q_{3k-2}$ is a triple cover
  of $\CT$.

  We let $Q_0:=\emptyset$.

  For the inductive step $i\to i+1$, suppose that $Q_i$ is
  defined. Let $\CX_i$ be the set of all partitions $(X_1,X_2,X_3)$
  oif $Q_i$ into three possibly empty sets. Then $|\CX_i|\le 3^{\theta(i)}$.

  For $X=(X_1,X_2,X_3)\in\CX_i$ we shall define an element
  $y_{X}\in A$ such
  that for all $Y_1,Y_2,Y_3\in\CT$, if $Q_i\cap Y_j\subseteq X_j$ for
  $j=1,2,3$ then either $y_X\in Y_1\cap Y_2\cap Y_3$
  or $\kappa(Y_1)+\kappa(Y_2)+\kappa(Y_3)\ge i+1$. Then we let
  \[
  Q_{i+1}:=Q_i\cup\{y_X\mid X\in\CX_i\}.
  \]
  Clearly, $|Q_{i+1}|\le |Q_i|+|\CX_i|\le \theta(i+1)$, and if $Y_1,Y_2,Y_3\in\CT$ with
  $Q_{i+1}\cap Y_1\cap Y_2\cap Y_3=\emptyset$, then
  $\kappa(Y_1)+\kappa(Y_2)+\kappa(Y_3)\ge i+1$, because otherwise
  \[
  y_{((Q_i\cap Y_1),(Q_i\cap Y_2),(Q_i\cap Y_3))}\in Y_1\cap Y_2\cap
  Y_3.
  \]

  Observe that for every $X\subseteq Q_i$, either there is no $Y\in\CT$
  such that $Q_i\cap Y\subseteq X$ or there is a (unique) $Y\in\CT$
  such that
  \begin{eroman}
  \item $Q_i\cap Y\subseteq X$;
  \item subject to (i), $\kappa(Y)$ is minimum;
  \item $Y\subseteq Y'$ for all $Y'\in\CT$ satisfying (i) and (ii).
  \end{eroman}
  This can be proved by a standard submodularity argument.

  Now let $X=(X_1,X_2,X_3)\in\CX_i$. If for some $j\in\{1,2,3\}$,
  there is no $Y_j\in\CT$ such that $Q_i\cap Y_j\subseteq X_j$, then
  there is nothing to do, and we can choose $y_{X}$
  arbitrarily. Otherwise, for $j=1,2,3$ we let $Y_j$ be the unique set
  in $\CT$ satisfying (i)--(iii) with respect to $X_j$. Then $Q_i\cap
  Y_1\cap Y_2\cap Y_3\subseteq X_1\cap X_2\cap X_3=\emptyset$ and
  thus, by the induction hypothesis,
  $\kappa(Y_1)+\kappa(Y_2)+\kappa(Y_3)\ge i$.

  Let $y=y_{X}\in
  Y_1\cap Y_2\cap Y_3$; such a $y$ exists by \ref{li:t2}.
  Let $Z_1,Z_2,Z_3\in\CT$ such that $Q_i\cap Z_j\subseteq X_j$ and
  $y\not\in Z_1\cap Z_2\cap Z_3$; if no such $Z_j$ exist there is nothing to prove. We claim that 
  \begin{equation}\label{eq:triple-cover}
  \kappa(Z_1)+\kappa(Z_2)+\kappa(Z_3)>\kappa(Y_1)+\kappa(Y_2)+\kappa(Y_3)\ge i.
  \end{equation}
  We observe first that for $j=1,2,3$ we have
  $\kappa(Z_j)\ge\kappa(Y_j)$ by (ii), and if
  $\kappa(Z_j)=\kappa(Y_j)$ then $Y_j\subseteq Z_j$ by (iii).
  
  Without loss of generality we may assume that $y\not\in Z_1$. Then
  $Y_1\not\subseteq Z_1$ and thus $\kappa(Z_1)>\kappa(Y_1)$. This proves
  \eqref{eq:triple-cover}.  
\end{proof}

\section{Treelike Decompositions}
\label{sec:treelike}

In a directed graph $D$, by $N^D_+(t)$ or just $N_+(t)$ if $D$ is clear from the
context, we denote the
set out-neighbours of a node $t$. By $\dagle^D$ or just $\dagle$ we
denote the reflexive transitive closure of $E(D)$, which is a partial
order if $D$ is acyclic. A \emph{directed tree} is a directed graph
$T$ where for all nodes $t$ the set $\{s\mid s\dagle t\}$ is linearly
ordered by $\dagle$.

Let $A$ be a set. A
\emph{directed decomposition} of $A$ is a pair $(D,\gamma)$,
where $D$ is a directed graph and $\gamma:V(D)\to 2^{A}$. If $\kappa$
is a connectivity function on $A$, we also say that $(D,\gamma)$ is a
directed decomposition of $\kappa$.
For every
node $t\in V(D)$, we let 
\begin{equation}
  \label{eq:1}
  \beta(t):=\gamma(t)\setminus\bigcup_{u\in N_+^D(t)}\gamma(u).
\end{equation}
We call $\beta(t)$ the \emph{bag} and $\gamma(t)$ the \emph{cone} at
$t$. 
We always denote the bag function of a directed decomposition
$(D,\gamma)$ by $\beta$, and we use implicit naming conventions by
which, for example, we denote the bag function of $(D',\gamma')$ by $\beta'$. 

A directed decomposition $(D,\gamma)$ of $A$ is
\emph{treelike}, or a \emph{treelike decomposition}, if it satisfies the following axioms.
\begin{nlist}{TL}
  \item\label{li:tl1} $D$ is a acyclic.
  \item\label{li:tl2} For all $(t,u)\in E(D)$,
    \[
    \gamma(t)\supseteq\gamma(u).
    \]
  \item\label{li:tl3} For all $t\in V(T)$ and $u_1,u_2\in N_+^D(t)$,
    \[
    \gamma(u_1)=\gamma(u_2)\quad\text{or}\quad
    \gamma(u_1)\cap\gamma(u_2)=\emptyset.
    \]
  \item\label{li:tl4} There is a $t\in V(D)$ such that $\gamma(t)=A$.
\end{nlist}
If $(D,\gamma)$ only satisfies (TL.1)--(TL.3), we call it a
\emph{partial treelike decomposition}. The treelike decompositions of
connectivity functions 
introduced here are adaptations of treelike decompositions of
graphs introduced in \cite{gro12+a,gro08a}.

In the following, let $(D,\gamma)$ be a partial treelike decomposition of $A$.
Observe
that for all $t\in V(D)$,
\begin{equation}
  \label{eq:2}
  \gamma(t)=\bigcup_{u\dagri t}\beta(u).
\end{equation}
$(D,\gamma)$ is a \emph{(partial) directed tree decomposition}\footnote{Deviating from previous work~\cite{geegerwhi09,grosch15}, we view the trees in tree
decompositions as being directed.} if $D$ is a
directed tree
and for all $t\in V(T)$ and all distinct $u_1,u_2\in N_+^D(t)$,
\[
\gamma(u_1)\cap\gamma(u_2)=\emptyset.
\]
Observe that $(D,\gamma)$ is a directed tree decomposition if and only if $D$
is a directed tree and the bags
$\beta(t)$ for $t\in V(D)$ are mutually disjoint and have union $A$
(that is, they form a partition of $A$ with possibly empty parts).

Now assume that $\kappa$ is a connectivity function on $A$ and
$(D,\gamma)$ a (partial) treelike decomposition of $\kappa$. The \emph{width}
of a node $t\in V(D)$ in $(D,\gamma)$ is 
\[
\width(D,\gamma,t):=\max_{\substack{X\subseteq\beta(t)\\U\subseteq N_+^D(t)}}\kappa\left(
  X\cup\bigcup_{u\in U}\gamma(u)\right).
\]
The \emph{width} $\width(D,\gamma)$ of the decomposition is the
maximum of the widths of its nodes.

It is sometimes convenient to normalise treelike decompositions. The
\emph{roots} of a directed acyclic graph $D$ are the nodes $r$ of in-degree
$0$. The
\emph{leaves} of a directed acyclic graph $D$ are the nodes $t$ of out-degree
$0$; we denote the set of all leaves of $D$ by $L(D)$. All non-leaf
nodes are called \emph{inner nodes}. We
say that a treelike decomposition $(D,\gamma)$ of $\kappa$ is
\emph{normal} if (in addition to the axioms
\ref{li:tl1}--\ref{li:tl4}) it satisfies the following conditions.
\begin{nlist}{NTL}
\item\label{li:ntl1} For all inner nodes $t\in V(D)\setminus L(D)$ it holds that
  $\beta(t)=\emptyset$.
\item\label{li:ntl2} For all leaves $t\in L(D)$ it holds that $|\beta(t)|=1$.
\item\label{li:ntl3} For all nodes $t\in V(D)$, either $\gamma(u_1)=\gamma(u_2)$ for
  all $u_1,u_2\in N_+(t)$ or $\gamma(u_1)\cap\gamma(u_2)=\emptyset$ for
  all distinct $u_1,u_2\in N_+(t)$.
\item\label{li:ntl4} $D$ has a unique root.
\end{nlist}

\begin{lem}\label{lem:normal}
  Let $(D,\gamma)$ be a treelike decomposition of $\kappa$. Then there
  is a normal treelike decomposition $(D',\gamma')$ of $\kappa$ such that
  $\width(D,\gamma)=\width(D',\gamma')$.
  If $D$ is a tree, then $D'$ is a tree as well, and if $(D,\gamma)$
  is a tree decomposition then $(D',\gamma')$
  is a tree decomposition as well.

  Furthermore, the construction of $(D',\gamma')$ from $(D,\gamma)$
  is canonical and can be carried out in polynomial time.
\end{lem}

\begin{proof}
  To satisfy \ref{li:ntl1} and \ref{li:ntl2}, we extend our
  decomposition as follows.
  Let $t$ be an inner node with
  $\beta(t)\neq\emptyset$ or a leaf $t\in L(D)$ with $|\beta(t)|>1$. For each
  $x\in \beta(t)$, we add a fresh node $t_x$ and an edge from $t$ to
  $t_x$. We leave $\gamma(v)$ unchanged for all old
  nodes $v$ and set $\gamma(t_x):=\{x\}$. We obtain a new treelike decomposition that
  satisfies \ref{li:ntl1} and \ref{li:ntl2}. This new decomposition has the same
  width as the old one.
 
  To satisfy \ref{li:ntl3}, we modify our decomposition as follows.
  For every node $t\in V(D)$ that has distinct children
  $u_1,u_2,u_3\in N_+(t)$
  such that $\gamma(u_1)=\gamma(u_2)\neq\gamma(u_3)$, we partition
  $N_+(t)$ into sets $U_1,\ldots,U_m$ such that all nodes in $U_i$ have the
  same cone and the nodes in distinct $U_i,U_j$ have disjoint cones. We
  delete all edges from $t$ to its children. Then we add $m$
  fresh nodes $t_1,\ldots,t_m$ and edges from $t$ to $t_i$ and from $t_i$
  to all nodes in $U_i$. We leave $\gamma(v)$ unchanged for all old
  nodes $v$ and set $\gamma(t_i):=\gamma(u_i)$ for some (and hence
  all) $u_i\in U_i$. We obtain a new treelike decomposition that
  satisfies \ref{li:ntl3}. This new decomposition has the same
  width as the old one, and it also leaves \ref{li:ntl1} and
  \ref{li:ntl2} intact.

  Suppose now that $(D,\gamma)$
  satisfies \ref{li:ntl1}--\ref{li:ntl3}; it remains to satisfy
  \ref{li:ntl4}. We repeatedly remove roots $t$
  with $\gamma(t)\neq A$ until we are left with a graph $D'$
  where all roots we have $\gamma(t)=A$. Then we add a new node $r$ an
  edges from $r$ to all roots of $D'$; the resulting graph
  $D''$ has $r$ as its only root. We define $\gamma'': V(D'')\to 2^A$
  by $\gamma''(r):=A$ and $\gamma''(t):=\gamma(t)$ for all $t\in
  V(D')\subseteq V(D)$. It is easy to see that $(D'',\gamma'')$
  satisfies \ref{li:tl1}--\ref{li:tl4} and thus is a treelike
  decomposition. It satisfies \ref{li:ntl4} by construction.
  We have $\beta''(r)=\emptyset$ and
  $\beta''(t)=\beta(t)$ for all $t\in V(D')$. Thus the decomposition
  $(D'',\gamma'')$ satisfies \ref{li:ntl1} and \ref{li:ntl2}. It
  also satisfies \ref{li:ntl3}, because all children of the root $r$
  have the same cone $A$ and all nodes $t\in V(D')$ have the same
  children as in $D$. Finally, $(D'',\gamma'')$ has the same width as
  $(D,\gamma)$, because the width at the root $r$ is $0$ and the width
  at all other nodes $t$ is the same as in $(D,\gamma)$.
\end{proof}

The following lemma shows that our definition of width is fairly robust.

\begin{lem}\label{lem:width}
  For every $k\in\NN$, the following statements are equivalent.
  \begin{eroman}
    \item $\bw(\kappa)\le k$.
   \item $\kappa$ has a directed tree decomposition of width at most $k$.
   \item $\kappa$ has a treelike decomposition of width at most $k$.
  \end{eroman}
\end{lem}

\begin{proof}
  To prove (i)$\implies$(ii), let $(T,\xi)$ be a branch decomposition
  of $\kappa$. If $E(T)=\emptyset$, then $|A|=1$, and the claim is trivial. So we
  assume $E(T)\neq\emptyset$. Let $e_0=s_0t_0\in E(T)$. We define a
  directed tree $T'$ by subdividing the edge $e_0$, making the newly
  inserted node, say, $r$, the root of $T'$ and directing all edges
  away from $r$. We define $\gamma:V(T')\to 2^A$ by:
  \begin{itemize}
  \item
    $\gamma(r):=A$;
  \item $\gamma(s_0):=\tilde\xi(t_0,s_0)$ and
    $\gamma(t_0):=\tilde\xi(s_0,t_0)$;
  \item $\gamma(t):=\tilde\xi(s,t)$ for every node
    $t\in V(T')\setminus\{r,s_0,t_0\}$ with parent $s$.
  \end{itemize}
  Observe every internal node $t\in V(T')\setminus L(T')$ has
  precisely two children, because the tree $T$ is cubic. 

  Clearly, $(T',\gamma)$ is a directed tree decomposition of $\kappa$.

  \begin{claim}
    $\width(T',\gamma)\le\width(T,\xi)$.

    \proof
    Let $t\in V(T')$; we shall prove that $\width(T',\gamma,t)\le
    \width(T,\xi)$.
    
    Suppose first that $t\in L(T)$. Then
    $\gamma(t)=\beta(t)=\{\xi(t)\}$, and thus 
    \[
    \width(T,\gamma,t)=\max\{\kappa(\emptyset),\kappa(\{\xi(t)\})\}=\kappa(\{\xi(t)\})\le \width(T,\xi).
    \]
    Suppose next that $t$ has parent $s\neq r$ and children
    $u_1,u_2$. Then
  $\bar{\gamma(t)}=\tilde\xi(t,s)$, $\gamma(u_1)=\tilde\xi(t,u_1)$, and
  $\gamma(u_2)=\tilde\xi(t,u_2)$. As the union of these three sets is
  $A$, we have $\gamma(t)=\gamma(u_1)\cup\gamma(u_2)$ and thus
  $\beta(t)=\emptyset$. Thus
  \begin{align*}
    \width(T',\gamma,t)&=\max\{\kappa(\emptyset),\kappa(\gamma(u_1)),\kappa(\gamma(u_2)),\kappa(\gamma(t))\}\\
    &=\max\{0,\kappa(\tilde\xi(t,u_1)),\kappa(\tilde\xi(t,u_2)),\kappa(\tilde\xi(s,t))\}\le
    \width(T,\xi).
  \end{align*}
  If $t=s_0$ or $t_0$, we can argue completely analogously, using
  $t_0$ or $s_0$, respectively, in place of the parent $s$.

  Finally, suppose that $t=r$ is the root. Then $N_+(t)=\{s_0,t_0\}$,
  and thus    
    \[
    \beta(t)=\gamma(t)\setminus(\gamma(s_0)\cup\gamma(t_0))=A\setminus(\tilde\xi(t_0,s_0)\cup\tilde\xi(s_0,t_0))=\emptyset,
    \]
    because $\xi(s_0,t_0)=\bar{\xi(t_0,s_0)}$. It follows that
    \[
    \width(T,\gamma,t)=\max\{\kappa(\emptyset),\kappa(\tilde\xi(t_0,s_0)),\kappa(\tilde\xi(s_0,t_0)),\kappa(A)\}\le
    \width(T,\xi).\uende
    \]
  \end{claim}

  To prove (ii)$\implies$(i), let $(T,\gamma)$ be a directed tree decomposition
  of $\kappa$. Without loss of generality we may assume that it is normal.
  If $T$ is a binary tree (where each internal node has exactly
  two children), we can turn the decomposition into a branch
  decomposition simply by inverting the construction in the proof of 
  (i)$\implies$(ii).

  It remains to turn the tree $T$ into a binary tree. Suppose that
  $t\in V(T)$ is a node with at least three children. Let $U_1,U_2$ be
  a partition of $U:=N_+(t)$. We modify the tree by deleting all edges
  from $t$ to its children, inserting two new
  children $u_1,u_2$ for $t$, and making all nodes in $U_1$ children of
  $u_1$ and and all nodes in $U_2$ children of $u_2$. Let $T'$ be the
  resulting tree. We define $\gamma':V(T')\to 2^A$ by letting
  $\gamma'(u_i):=\bigcup_{u\in U_i}\gamma(u)$ for $i=1,2$ and
  $\gamma'(x):=\gamma(x)$ for all $x\in
  V(T')\setminus\{u_1,u_2\}=V(T)$. Observe that
  $\width(T',\gamma',u_i)\le \width(T,\gamma,t)$ for $i=1,2$ and 
  \[
  \width(T',\gamma',t)=\max\{\kappa(\emptyset),\kappa(\gamma(u_1))
  ,\kappa(\gamma(u_1)), \kappa(\gamma(t))\}\le \width(T,\gamma,t)
  \]
  and $\width(T',\gamma',x)\le \width(T,\gamma,x)$ for all $x\in
  V(T')\setminus\{u_1,u_2,t\}$. Thus
  $\width(T',\gamma')\le\width(T,\gamma)$. We repeat this construction
  until the tree is binary.

  \medskip
  The implication (ii)$\implies$(iii) is trivial, because every
  directed tree
  decomposition is a treelike decomposition. Thus it remains to prove
  (iii)$\implies$(ii). Let $(D,\gamma)$ be a treelike decomposition of
  $\kappa$. By repeatedly duplicating subtrees, starting from the
  leafs, we can turn $D$ into a forest $F$ (which may be exponentially
  larger than $D$). We can define a function $\gamma':V(F)\to
  2^A$ accordingly and obtain a treelike decomposition $(F,\gamma')$
  of the same width that is based on the forest $F$ instead of the
  directed acyclic graph $D$. We pick a root node $r$ of some tree $T$
  of $F$ with $\gamma'(r)=A$. Then we prune the tree as follows:
  whenever we have a node $t$ with a family $U\subseteq N_+(T)$ of
  children that all have the same $\gamma(u)$, we delete all but one
  of these children and the whole subtrees attached to them. The
  result is a tree decomposition.  
\end{proof}

\begin{rem}
  Note that the constructions in the proof of Lemma~\ref{lem:width},
  turning treelike decomposition into a tree decomposition and a tree
  decomposition into a branch decomposition, are not canonical. In
  fact, it is not difficult to see that there are no width-preserving canonical
  constructions for these tasks; this is
  why we introduced treelike decompositions and tree decompositions in
  the first place.
\end{rem}

Now let $\KT$ be a family of mutually incomparable $\kappa$-tangles. A
\emph{directed tree decomposition for $\KT$} is a triple
$(T,\gamma,\tau)$, where $(T,\gamma)$ is a directed tree decomposition
of $\kappa$ and $\tau:\KT\to V(T)$ a bijective mapping such that the
following two conditions are satisfied.
  \begin{nlist}{DTD}
  \item\label{li:dtd1} For all nodes $t,u\in V(T)$ with $u\not\dagle t$ there is a minimum $(\tau^{-1}(u),\tau^{-1}(t))$-separation
    $Y$ such that $\gamma(u)\subseteq Y$.
\item\label{li:dtd2} For all nodes $t\in V(T)$ except the root, there is a node
    $u\in V(T)$ such that $t\not\dagle u$ and $\gamma(t)$ is a
    leftmost minimum $(\tau^{-1}(t),\tau^{-1}(u))$-separation.
  \end{nlist}
Observe that \ref{li:dtd1} implies that for all nodes $t\in V(T)$ and children $u\in N_+(t)$ we have
    $\gamma(u)\not\in\tau^{-1}(t)$. Furthermore, \ref{li:dtd2} implies
    that $\gamma(t)\in\tau^{-1}(t)$.

Recall that a $\kappa$-tangle $\CT$ is \emph{$k$-maximal}, for some $k\ge 0$, if
either $\ord(\CT)=k$ or $\ord(\CT)<k$ and $\CT$ is an (inclusionwise)
maximal tangle. We denote the family of all $k$-maximal
$\kappa$-tangles by $\KT^{\le k}_{\max}$. Observe that for $k=\bw(\kappa)$ the $k$-maximal
$\kappa$-tangles are precisely the maximal
$\kappa$-tangles.

\begin{theo}[\cite{grosch15}]\label{theo:dcandec}
  Let $\ell\ge0$. Then there is a polynomial time algorithm that,
  given oracle access to a connectivity function $\kappa$ and a $\kappa$-tangle
  $\CT_{\textup{root}}\in \KT^{\le\ell}_{\max}$ (via a membership oracle or its index in a
  comprehensive tangle data structure for $\kappa$), computes a
  canonical directed tree decomposition $(T,\gamma,\tau)$ for the set
  $\KT^{\le\ell}_{\max}$ such
  that $\tau^{-1}(r)=\CT_{\textup{root}}$ for the root $r$ of $T$.
\end{theo}

Here canonical means that if $\kappa':2^{A'}\to\NN$ is another
connectivity function and $\CT'_{\textup{root}}$ an $\ell$-maximal
$\kappa'$-tangle, and $(T',\gamma',\tau')$ is the decomposition
computed by our algorithm on input $(\kappa',\CT'_{\textup{root}})$,
then for every isomorphism $f$ from
$(\kappa,\CT_{\textup{root}})$ to $(\kappa',\CT'_{\textup{root}})$,
that is, bijective mapping $f:A\to A'$ with $\kappa(X)=\kappa'(f(X))$
and $X\in \CT_{\textup{root}}\iff f(X)\in \CT'_{\textup{root}}$ for
all $X\subseteq A$, there is an isomorphism $g$ from $T$ to $T'$ such
that that $f(\gamma(t))=\gamma'(g(t))$ for all $t\in V(T)$ and
$X\in\tau^{-1}(t)\iff f(X)\in(\tau')^{-1}(g(t))$ for all $X\subseteq
A,t\in V(T)$.

\section{Partitioning with Respect to a Maximal Tangle}
\label{sec:1tan}

Let $G$ be a graph of rank width at most $k$.
In this and the following two sections, we describe our construction of
a canonical treelike decomposition of $\rho_G$ of width at most $a(k)$ (for some
function $a$). Since large parts of the construction go through for
arbitrary connectivity functions, we find it convenient to let
$\kappa:=\rho_G$ and $A:=V(G)$.

We start from a directed tree decomposition $(T,\gamma,\tau)$ for
$\KT^{\le k}_{\max}$. The idea is to decompose the ``pieces'' of this
decomposition, corresponding to the nodes of $T$, further into
decompositions of bounded width and then merge all these
bounded-width decompositions into one big decomposition. The largest
part of the construction, resulting in Lemma~\ref{lem:nodedec}, deals
with a single node of $T$.

So we fix a node $t\in V(T)$. We let $\CT_0:=\tau^{-1}(t)$ be the
maximal tangle associated with $t$ and $k_0:=\ord(\CT_0)$. Let
$B:=\beta(t)$ and
$C_0:=\bar{\gamma(t)}$. Assuming that the children of $t$ in $T$ are
$u_1,\ldots,u_m$, we let $C_i:=\gamma(u_i)$ for $i\in[m]$ Observe that
the sets $B,C_0,\ldots,C_m$ form a partition of $A$ (the set $C_0$ may
be empty). Now we contract the sets $C_0,\ldots,C_m$. We shall
construct a bounded width decomposition of the resulting connectivity
function $\kct$ on the contracted set $\Act$.

We construct the decomposition recursively. At any time, we have a set
$X\subseteq\Act$ that still needs to be decomposed, and we will show how to
partition $X$ in a canonical way, at any time keeping control of the
width of the resulting decomposition. 

We initialise the construction by taking a triple cover $Q$ of the tangle $\CT_0$ of
    size $|Q|\le\theta(3k_0-2)$. We let $Q^{\vee}$ be the
    ``projection'' of $Q$ into $\Act$ (precise definitions follow). The set $Q^\vee\cup\{c_0\}$ will be the
bag at the    root of our decomposition, and the first set $X$ to be
    decomposed further is $\Act\setminus (Q^\vee\cup\{c_0\})$.

Now suppose we are in some decomposition step where we need to
decompose $X\subseteq \Act\setminus (Q^\vee\cup\{c_0\})$. 
Depending on $\kct(X)$, we do this in two completely different
ways. In this section (Section~\ref{sec:1tan}), we consider the case $\kct(X)<(3k+2)\cdot k$,
and in Section~\ref{sec:nwl} we shall consider the case $\kct(X)\ge (3k+2)\cdot k$.

\subsection{Assumptions}
Before we start the technical construction, we step back and collect
the assumptions we make in a slightly more abstract setting, which we
fix for the rest of the section.

\begin{asss}\label{ass:1}
  \begin{enumerate}
  \item $\kappa:2^A\to\NN$ is a connectivity function on a set $A$.
  \item $k:=\bw(\kappa)\ge 1$.
  \item $C_0,\ldots,C_m\subseteq A$ are mutually disjoint sets with
    $\kappa(C_i)<k$, and $C_1,\ldots,C_m$ are nonempty.
  \item For all $i\in[m]$ there are tangles $\CT_i,\CT_i'$ such that
    $C_i$ is a leftmost minimum
    $(\CT_i,\CT_i')$-separation. Furthermore, if $C_0\neq\emptyset$
    then there are tangles $\CT_0,\CT_0'$ such that such that
    $\bar C_0$ is a leftmost minimum
    $(\CT_0,\CT_0')$-separation.
   \item
    $B:=A\setminus(C_0\cup\ldots\cup C_m)$.
  \item $\Act:=A\contract_{C_0,\ldots,C_m}$, and $c_i$ is 
    the element of $\Act$ corresponding to the contracted set $C_i$,
    for $i=0,\ldots,m$.
 \item $\kct:=\kappa\contract_{C_0,\ldots,C_m}$.
  \end{enumerate}
\end{asss}

The assumption $\bw(\kappa)\ge 1$ is without loss of generality,
because if $\bw(\kappa)=0$ then $\kappa(\{x\})=0$ for all $x\in A$ and
thus $\kappa(X)=0$ for all $X\subseteq A$.

\begin{asss}\label{ass:2}
  \begin{enumerate}
\item
    There is a maximal $\kappa$-tangle $\CT_0$ such that $\bar
    C_0\in\CT_0$ and $C_i\not\in\CT_0$ for $i=1,\ldots,m$.\footnote{If
      $C_0\neq\emptyset$, then the tangles $\CT_0$ in
Assumption~\ref{ass:1}(4) and Assumption~\ref{ass:2}(1) are the same,
but this is irrelevant.}
  \item For every $\kappa$-tangle $\CT\bot\CT_0$, there is an
    $i\in\{0,\ldots,m\}$ and a set $Y\subseteq C_i$ such that $Y\in\CT$.
  \item $k_0:=\ord(\CT_0)$. (Note that $k_0\le k$.) 
  \end{enumerate}
\end{asss}

Observe that \[
\CT_0\contract:=\CT_0\contract_{C_0,\ldots,C_m}.
\]
is a $\kct$-tangle, because
$C_i\not\in\CT_0$ for $i=0,\ldots,m$ by Assumption~\ref{ass:2}(1).

\begin{asss}\label{ass:3}
  \begin{enumerate}
  \item $Q\subseteq A$ is a triple cover of the tangle $\CT_0$ of
    size $|Q|\le\theta(3k_0-2)$.
  \item $Q^{\vee}:=(B\cap Q)\cup\{c_i\mid 0\le i\le m, C_i\cap Q\neq\emptyset\}$.
  \end{enumerate}
\end{asss}
Observe that $Q^{\vee}$ is a triple cover for the 
$\kct$-tangle $\CT_0\contract$.

All algorithms we devise in this section will get $\kappa$ and
$C_0,\ldots,C_m$ and $Q$ as input, and possibly other objects. We
assume that we have constructed a comprehensive tangle data structure
for $\kappa$ and have determined the index of $\CT_0$ in this data
structure. Thus our algorithms also have access to $\CT_0$.

Whenever we refer to a construction in
this section as being \emph{canonical}, what we mean is that it is
canonical given $\kappa$ and $C_0,\ldots,C_m$ and $Q$. Note that $\CT_0$ is canonical
given $\kappa$ and $C_0,\ldots,C_m$, because $\CT_0$ is the unique maximal
$\kappa$-tangle with $\bar
    C_0\in\CT_0$ and $C_i\not\in\CT_0$ for $i=1,\ldots,m$. Thus we may depend on $\CT_0$ in canonical
constructions.

Our goal is to prove the following lemma.

\begin{lem}\label{lem:1tan}
  For every $k_1\in\NN$ there are $a_1=a_1({k},k_1)$, $b_1=b_1({k},k_1)$,
  and $f_1=f_1({k},k_1)>0$ such that for every
  $X\subseteq \Act\setminus (Q^{\vee}\cup\{c_0\})$ of order $\kct(X)=k_1$ and size
  $|X|\ge 2$, one of the
  following two conditions is satisfied.
  \begin{eroman}
  \item
    There is a canonical partition of $X$ into $b\le b_1$
    sets $X_1,\ldots,X_b$ such that $\kct(X_i)\le a_1$ and $|X_i|\le(1-1/f_1)|X|$ for
    all $i\in[b]$.
  \item
    There is a canonical partition of $X$ into sets $X_0,X_1,\ldots,X_n$ such
    that
    \begin{enumerate}[label=\alph*.]
    \item $\kct(X_0)\le k_1$,
    \item $\kct\left(\bigcup_{i\in I}X_i\right)\le 2k_1$ for every set $I\subseteq[n]$,
    \item $|X_i|\le(1-1/f_1)|X|$ for every $i\in[n]$.
    \end{enumerate}
  \end{eroman}
  Furthermore, given $X$ (in addition to $\kct$ and $C_0,\ldots,C_m$ and $Q$), the partition
  in (i) or (ii) can be computed in polynomial time (for fixed ${k},k_1$).
\end{lem}

The lemma will be proved in Section~\ref{sec:1tan3}.

For the rest of Section~\ref{sec:1tan}, we fix a set $X\subseteq\Act\setminus (Q^{\vee}\cup\{c_0\})$.
Let 
\begin{align*}
k_1&:=\kct(X),\\
k_2&:=k_0+k_1.
\end{align*}
We assume that
\begin{equation}\label{eq:xsize}
  |X|\ge 6k_2.
\end{equation}
Note that this implies $|X|\ge 6$ by our assumption that $k_0\ge 1$.

\subsection{Existence of a Balanced Separations}
\label{sec:1tan1}

We call a set $Z\subseteq X$ a \emph{balanced $X$-separation} if
$\kct(Z)\le k_1= \kct(X)$ and 
\[
\frac{1}{3}|X|-k_2+\kct(Z)\le |Z|\le \frac{2}{3}|X|+k_2-\kct(Z)
\]
Note that this notion does not only depend on $X$, but through $k_2$
also on $k_0$, the order of the tangle $\CT_0$.

\begin{lem}\label{lem:1tan1}
  Suppose that $X\subseteq\bar Y$ for some $Y$ in
  $\CT_0\contract$. Then there is a balanced $X$-separation.
\end{lem}

\begin{proof}
  By the assumption of the lemma, there exists an $(X,\CT_0\contract)$
  separation $Y$.  Let $X'$ be the leftmost minimum
  $(X,\CT_0\contract)$-separation. Then $X\subseteq X'\subseteq\Act$
  and $\kct(X')\le\kct(X)=k_1$ and $\bar X'\in\CT_0\contract$ and thus
  $\kct(X')<\ord(\CT_0\contract)=k_0$.

  \begin{claim}
    There is a $Y\in\CT_0\contract$ such that
    $\kct(Y)\le\kct(X')$ and 
    \[
    \frac{1}{3}|X|-k_0+\kct(Y)\le|X\cap Y|\le
    \frac{2}{3}|X|+k_0-\kct(Y).
    \]
    
    \proof
    We define a weight function $\phi:A\to \RR$ as follows:
    \[
    \phi(x):=
    \begin{cases}
      \frac{1}{|C_i|}&\text{if $x\in C_i$ for some $i\in[m]$ such
        that $c_i\in X$},\\
      1&\text{if }x\in X\cap B,\\
      0&\text{otherwise}.
    \end{cases}
    \]
    For a set $Y\subseteq A$ we let $\phi(Y):=\sum_{y\in
      Y}\phi(Y)$. Note that for every $Y\subseteq A$ we have
    $\phi(Y)+\phi(\bar Y)=|X|$. Furthermore, for all $Y^\vee\subseteq\Act$ we have
    $\phi(Y^\vee\expand)=|X\cap Y^\vee|$.
    
    Suppose first that there is no $Y\subseteq A$ such that
    $\kappa(Y)\le\kct(X')$ and
    $\frac{1}{3}|X|\le\phi(Y)\le\frac{2}{3}|X|$. Let
    \[
    \CT:=\{ Y\subseteq A\mid \kappa(Y)\le\kct(X')\text{ and
    }\phi(Y)>\textstyle\frac{2}{3}|X|\}.
    \]
    Then $\CT$ is a $\kappa$-tangle of order $\kct(X')+1\le k_0$. To
    verify \ref{li:t3}, let $x\in A$. Then $\phi(\{x\})\le
    1\le(2/3)|X|$, because $|X|\ge 2$. 
    
    Observe that $X'\expand\in\CT$, because $\phi(X'\expand)=|X\cap X'|=|X|$. Thus $\CT\bot\CT_0$.
    
    For every $i\in\{0,\ldots,m\}$ and every $Y\subseteq C_i$ we
    have $\phi(Y)\le\phi(C_i)\le 1$ and thus $Y\not\in\CT$. This
    contradicts Assumption~\ref{ass:2}(2).
    
    We choose $\ell$ minimum such that there exists a $Y\subseteq A$ with
    $\kappa(Y)=\ell$ and
    \begin{equation}
      \label{eq:3}
      \frac{1}{3}|X|-k_0+\ell\le\phi(Y)\le
    \frac{2}{3}|X|+k_0-\ell.
    \end{equation}
    Then $\ell\le \kct(X')\le k_0$, because we have just proved that
    there is a set $Y$ with $\kappa(Y)\le\kct(X')$ and 
    \[
      \frac{1}{3}|X|-k_0+\kappa(Y)\le \frac{1}{3}|X|\le\phi(Y) \le
    \frac{2}{3}|X|\le
    \frac{2}{3}|X|+k_0-\kappa(Y).
    \]
    Without loss of generality we may assume that either
    $Y\cap C_0=\emptyset$ or $C_0\subseteq Y$. To see this, suppose
    that neither $Y\cap C_0=\emptyset$ nor $C_0\subseteq Y$, or
    equivalently, neither $Y\subseteq\bar C_0$ nor
    $\bar Y\subseteq \bar C_0$. By Assumption~\ref{ass:1}(4), $\bar C_0$
    is a minimum $(\CT_0,\CT_0')$-separation. Thus by
    Lemma~\ref{lem:tanglesep3}(1) (applied to $X=\bar C_0$), either
    $\kappa(Y\cap\bar C_0)\le\kappa(Y)$ or
    $\kappa(\bar Y\cap\bar C_0)\le\kappa(Y)$. As $c_0\not\in X$, we
    have $\phi(Y\cap\bar C_0)=\phi(Y\cup C_0)=\phi(Y)$. Thus
    if $\ell':=\kappa(Y\cap\bar C_0)\le\kappa(Y)$, then
    $Y':=Y\cap\bar C_0$ satisfies \eqref{eq:3} with $Y',\ell'$
    instead of $Y,\ell$, and if
    $\ell':=\kappa(\bar Y\cap\bar C_0)\le\kappa(Y)$, then
    $Y':=\bar Y\cap\bar C_0$ satisfies \eqref{eq:3} with $Y',\ell'$
    instead of $Y,\ell$. In both cases, we have
    $Y'\cap C_0=\emptyset$. This justifies the assumption that either
    $Y\cap C_0=\emptyset$ or $C_0\subseteq Y$.

    Suppose for contradiction that there is an $i\in[m]$ such that
    neither $Y\cap C_i=\emptyset$ nor $C_i\subseteq Y$, or
    equivalently, neither $Y\subseteq\bar C_i$ nor
    $\bar Y\subseteq \bar C_i$. We argue
    similarly as for $C_0$, but with Lemma~\ref{lem:tanglesep3}(2). By Assumption~\ref{ass:2}, $\bar C_i$
    is a rightmost minimum $(\CT_i',\CT_i)$-separation. Thus by
    Lemma~\ref{lem:tanglesep3}(2) (applied to $X=\bar C_i$), either
    $\kappa(Y\cap\bar C_i)<\kappa(Y)$ or
    $\kappa(\bar Y\cap\bar C_i)<\kappa(Y)$. Without loss of
    generality, we assume that $\kappa(Y\cap \bar C_i)<\kappa(Y)$. The
    case $\kappa(\bar Y\cap\bar C_i)<\kappa(Y)$ is symmetric, because \eqref{eq:3}
    is symmetric in $Y,\bar Y$.
 Let $Y':=Y\cap \bar C_i=Y\setminus C_i$ and
    $\ell':=\kappa(Y')<\ell$. We have
    \[
    \phi(Y)\ge\phi(Y')=\phi(Y)-\phi(C_i)\ge\phi(Y)-1.
    \]
    Thus by \eqref{eq:3},
   \begin{align*}
   \frac{1}{3}|X|-k_0+\ell'&\le
   \frac{1}{3}|X|-k_0+\ell-1\le
   \phi(Y)-1\\
                                &\le\phi(Y')
\le\phi(Y)\le
   \frac{2}{3}|X|+\kct(X')-\ell\\
     &\le
   \frac{2}{3}|X|+\kct(X')-\ell'.
   \end{align*}
   Thus $Y',\ell'<\ell$ satisfy \eqref{eq:3}.
   This contradicts the minimality of $\ell$.
    
   We have proved that for all $i\in[m]$, either $Y\cap C_i=\emptyset$ or
    $C_i\subseteq Y$. This implies that there is a $Y^\vee\subseteq\Act$
      such that $Y=Y^\vee\expand$. As $\phi(Y)=|X\cap Y^\vee|$ and
      $\kct(Y^\vee)=\kappa(Y)=\ell$, by \eqref{eq:3} we have
      \begin{equation}
        \label{eq:4}
      \frac{1}{3}|X|-k_0+\kct(Y^\vee)\le|X\cap Y^\vee|\le
    \frac{2}{3}|X|+k_0-\kct(Y^\vee).
      \end{equation}
      As $\kct(Y^\vee)=\ell\le\kct(X')<\ord(\CT_0\contract)$, either
      $Y^{\vee}\in\CT_0\contract$ or $\bar Y^{\vee}\in\CT_0\contract$.
      Since \eqref{eq:4} is symmetric in $Y^\vee$ and $\bar Y^\vee$,
      we may assume that $Y^\vee\in\CT_0\contract$.       \uend
    \end{claim}

    We choose $Y\in\CT_0\contract$ according to the claim.
    Let
    \[
    Z:=X\cap\bar Y.
    \]
    Suppose for contradiction that $\kct(Z)>\kct(X)$. Then by submodularity,
    $\kct(X\cup\bar Y)<\kct(Y)$. As 
    \[
    \bar X'\cap Y\cap(X\cup\bar
    Y)\subseteq \bar X'\cap Y\cap(X'\cup\bar
    Y)=\emptyset,
    \]
    we have $X\cup\bar
    Y\not\in\CT_0\contract$ and thus $\bar{X\cup \bar
      Y}\in\CT_0\contract$, which means that $X\cup \bar Y$ is an $(X,\CT_0\contract)$-separation. However, as
    $\kct(X\cup \bar Y)<\kct(Y)\le\kct(X')$, this contradicts  $X'$
    being a minimum $(X,\CT_0\contract)$-separation.

    Hence $\kct(Z)\le\kct(X)=k_1$, and we have
    \[
    \frac{1}{3}|X|-k_2+\kct(Z)\le \frac{1}{3}|X|-k_0\le
    \frac{1}{3}|X|-k_0+\kct(Y)\le|X\cap Y|=|X\setminus Z|=|X|-|Z|,
    \]
    which implies $|Z|\le \frac{2}{3}|X|+k_2-\kct(Z)$. Similarly,
    $|Z|\ge \frac{1}{3}|X|-k_2+\kct(Z)$.
\end{proof}

In the following lemma, we drop the assumption $X\subseteq\bar Y$ for
some $Y\in\CT_0\contract$.

\begin{lem}\label{lem:1tan2}
  There is a balanced $X$-separation.
\end{lem}

\begin{proof}
  If $X\subseteq\bar Y$ for some $Y\in\CT_0\contract$, then the assertion follows
  from Lemma~\ref{lem:1tan1}. So suppose that $X\cap Y\neq\emptyset$ for
  all $Y\in\CT_0\contract$.

    \begin{claim}
      There is a $Y\subseteq\Act$ such that
      $\kct(Y)\le k_0$ and 
      \[
      \frac{1}{3}|X|-k_0+\kct(Y)\le|X\cap
      Y|\le \frac{1}{3}|X|-k_0+\kct(Y).
      \]

      \proof As in the proof of the previous lemma, we define a weight
      function $\phi:A\to \RR$ by
      \[
      \phi(x):=
      \begin{cases}
        \frac{1}{|C_i|}&\text{if $x\in C_i$ for some $i\in[m]$ such
          that $c_i\in X$},\\
        1&\text{if }x\in X\cap B,\\
        0&\text{otherwise}.
      \end{cases}
      \]
      For a set $Y\subseteq A$ we let $\phi(Y):=\sum_{y\in
        Y}\phi(Y)$. 
      
      Suppose first that there is no $Y\subseteq A$ such that
      $\kappa(Y)\le k_0$ and
      $\frac{1}{3}|X|\le\phi(Y)\le\frac{2}{3}|X|$. Let
      \[
      \CT:=\{ Y\subseteq A\mid \kappa(Y)\le k_0\text{ and
      }\phi(Y)>\textstyle\frac{2}{3}|X|\}.
      \]
      Then $\CT$ is a $\kappa$-tangle of order $k_0+1$. 
      We have $\CT\bot\CT_0$, because $\CT_0$ is maximal and
      $\ord(\CT)>\ord(\CT_0)$.
      For every $i\in\{0,\ldots,m\}$ and every $Y\subseteq C_i$ we
      have $\phi(Y)\le\phi(C_i)\le 1$ and thus $Y\not\in\CT$. This
      contradicts Assumption~\ref{ass:2}(2).

      We choose $\ell$ minimum such that there exists a $Y\subseteq A$ with
      $\kappa(Y)=\ell$ and
      \begin{equation}
        \label{eq:5}
        \frac{1}{3}|X|-k_0+\ell\le\phi(Y)\le
        \frac{2}{3}|X|+k_0-\ell.
       \end{equation}
       Arguing as in the proof of Claim~1 in the proof of
       Lemma~\ref{lem:1tan1}, we prove that $\ell\le k_0$ and that we
       may assume $C_i\cap Y=\emptyset$ or $C_i\subseteq Y$ for all
       $i\in\{0,\ldots,m\}$.

      Then there is a $Y^\vee\subseteq\Act$
      such that $Y=Y^\vee\expand$. As $\phi(Y)=|X\cap Y^\vee|$, the claim
      follows.
      \uend
    \end{claim}

   We choose $Y\subseteq\Act$ according to the claim.
   Suppose for contradiction that $\kct(X\cap Y)>\kct(X)$ and
  $\kct(X\cap \bar Y)>\kct(X)$. 
  Then by submodularity, $\kct(X\cup Y)<\kct(Y)\le k_0$ and
  $\kct(X\cup \bar Y)< k_0$. As $X$ has a nonempty
  intersection with every element of $\CT_0\contract$, we have $X\cup
  Y,X\cup\bar Y\in\CT_0\contract$. Thus $Q^{\vee}\cap (X\cup Y)\cap(X\cup\bar
  Y)\neq\emptyset$, because $Q^{\vee}$ is a triple cover of
  $\CT_0\contract$ (a double cover would be sufficient here). This is
  a contradiction.

  Thus either $\kct(X\cap Y)\le\kct(X)$ or
  $\kct(X\cap \bar Y)\le\kct(X)$. Without loss of generality we assume
  the former. Then it is easy to see that $Z:=X\cap Y$ is $X$-balanced.
\end{proof}

\subsection{A Canonical Family of Separations}
\label{sec:1tan2}

For $0\le\ell\le k_1$, let 
\[
p(\ell):=2^{-3^{k_2-\ell}}.
\]
Note that 
\begin{equation}
  \label{eq:6}
  \frac{1}{8}\ge p(k_1)\ge p(k_1-1)\ge \ldots\ge p(0)
\end{equation}
and
\begin{equation}
  \label{eq:7}
  p(\ell-1)=p(\ell)^3
\end{equation}
for all $\ell\ge 1$.

Let us call a set $Z\subseteq X$ of order $\ell:=\kct(Z)$ 
\emph{good} (or a \emph{good separation}) if
\begin{equation}\label{eq:8}
p(\ell)\cdot|X|\le|Z|<|X|.
\end{equation}
Recall that $\frac{1}{6}|X|\ge k_2$ by \eqref{eq:xsize}. Thus for $0\le\ell\le k_1$ we have 
\[
\frac{1}{3}|X|-k_2+\ell\ge 
\frac{1}{3}|X|-k_2\ge
\frac{1}{6}|X|\ge 
\max\big\{1,
p(\ell)\cdot|X|\big\}
\]
It follows that every balanced $X$-separation is good. 
Hence by Lemma~\ref{lem:1tan2}, there is a good separation $Z$ of
order $\kct(Z)\le k_1$.

Let $\ell$ be minimum such that there is a good separation  $Z$ of order
$\kct(Z)=\ell$. 

Let $\CZ$ be the set of all $Z\subseteq X$ such that
\begin{eroman}
\item $Z$ is good;
\item $\kct(Z)=\ell$;
\item $|Z|$ is maximum subject to (i)  and (ii).
\end{eroman}
Observe that $|Z|=|Z'|\ge p(\ell)\cdot|X|$ for
all $Z,Z'\in\CZ$.
Let
\[
\CY:=\{Y\subseteq\Act\mid \bar Y\in\CZ\}.
\]
Note that $\bar X\subset Y$ and $\kct(Y)=\ell$ and
$|Y|=|Y'|$ for
all $Y,Y'\in\CY$. Let us call two sets $Y,Y'$ \emph{$X$-disjoint} if $Y\cap
Y'\subseteq\bar X$. Observe that for $X$-disjoint sets $Y,Y'\in\CY$ we
have $Y\cap Y'=\bar X$.

Our next goal is to prove the following lemma.

\begin{lem}\label{lem:z0}
  There is a $b_2=b_2(k_0,k_1)$ such that if $|\CY|>b_2$ then the elements
  of $\CY$ are mutually $X$-disjoint.
\end{lem}

The idea of the proof is as follows. Assume that there are $Y,Y'\in\CY$ that
are not $X$-disjoint. Then there are $Z,Z'\in\CZ$ whose union is a proper
subset of $X$. The choice of the function $p$ and a submodularity
argument guarantee that these sets $Z,Z'$ have a small
intersection. Thus $|Y\setminus Y'|=|Z'\setminus Z|$ is relatively large (close to
$p(\ell)|X|$, i.e., a constant fraction of $|X|$) and thus
$|Y\setminus\bar X|\ge|Y\setminus Y'|$ is relatively large. As all elements of $\CY$ have
the same size, this holds for all $Y\in\CY$. Now we apply Ramsey's
Theorem and find that if $\CY$ is very large either (i) there is a large family
$Z_1,\ldots,Z_n\in\CZ$ such that all pairwise unions  $Z_i\cup Z_j$
are proper
subsets of $X$, or (ii) there is a large family
$Z_1,\ldots,Z_n\in\CZ$ such that all pairwise unions $Z_i\cup Z_j$ are
equal to $X$. In
case (i), we argue that the $Z_i$ are
relatively large, but have a small intersection, and thus for large $n$ their union
becomes larger than $|X|$, which is impossible. In case (ii) we argue that the $Y_i:=\bar
Z_i$ are mutually $X$-disjoint, and as the sets $Y_i\setminus\bar X$
are relatively large, for large $n$ their union
becomes larger than $|X|$. Again, this is impossible. Thus the size of $\CY$ must be bounded.

The actual proof requires some preparation.

\begin{lem}\label{lem:z1}
  Let $Z,Z'\in\CZ$ with $Z\neq Z'$ and $Z\cup Z'\neq X$. Then
  \[
  |Z\cap Z'|<p(\ell)^3\cdot|X|
  \]
\end{lem}

\begin{proof}
  Let $\ell_{\cup}:=\kct(Z\cup Z')$. Suppose for contradiction that
  $\ell_{\cup}\le \ell$. Then 
   \[
   p(\ell_\cup)\cdot|X|\le p(\ell)\cdot|X|\le|Z|\le|Z\cup Z'|<|X|,
  \]
  where the last inequality holds, because $Z\cup Z'\neq X$. By
  the choice of $\ell$ we thus have $\ell_{\cup}=\ell$. Hence $Z\cup Z'$
  satisfies (i) and (ii). However, $|Z_\cup|>|Z|$, because $Z\neq
  Z'$ and $|Z|=|Z'|$. This is a contradiction, which proves
  $\kct(Z\cup Z')=\ell_{\cup}>
  \ell$. 

  By submodularity,
  $\ell_\cap:=\kct(Z\cap Z')<\ell$. Thus by the choice of $\ell$ we have
  \[
  |Z\cap Z'|<p(\ell_\cap)\cdot|X|\le p(\ell)^3\cdot|X|.
  \qedhere
  \]
\end{proof}

\begin{lem}\label{lem:z2}
  Let $Z_1,\ldots,Z_m\in\CZ$ such that $Z_i\cup Z_j\neq X$ for all distinct
  $i,j\in[m]$. Then 
  \[
  m<\frac{2}{p(\ell)}
  \]
\end{lem}

\begin{proof}
  As $|Z_i|\ge p(\ell)\cdot |X|$ and $|Z_i\cap Z_j|\le
  p(\ell)^3\cdot|X|$,
  we have
  \[
  \left|Z_{i}\setminus\bigcup_{j=1}^{i-1}Z_j\right|\ge
  \big(p(\ell)-(i-1)\cdot p(\ell)^3\big)\cdot|X|.
  \]
  This implies, for all $i\le m$
  \[
  |X|\ge \left|\bigcup_{j=1}^i Z_j\right|\ge \left(i\cdot
    p(\ell)-\sum_{j=1}^{i-1}j\cdot p(\ell)^3\right)\cdot|X|.
  \]
  Thus
  \begin{align*}
    \notag
    &\left(i\cdot
    p(\ell)-\frac{i\cdot(i-1)}{2}\cdot p(\ell)^3\right)=\left(i\cdot
    p(\ell)-\sum_{j=0}^{i-1}j\cdot p(\ell)^3\right)\le 1\\
    \iff\;&\frac{p(\ell)^3}{2}\cdot
            i^2-\left(p(\ell)+\frac{p(\ell)^3}{2}\right)\cdot i+1\ge0
  \end{align*}
  It is easy to see that this last inequality is violated for
  $i=\frac{2}{p(\ell)}$, which is an integer because
    $\frac{1}{p(\ell)}$ is a power of $2$. 
    Thus $m<\frac{2}{p(\ell)}$.
\end{proof}

\begin{lem}\label{lem:z3}
  For all $n\ge 1$ there is an $m=m(\ell,n)$ such that if $|\CZ|> m$ then
  there are $Z_1,\ldots,Z_n\in\CZ$ such that $Z_i\cup Z_j=X$ for all distinct $i,j\in[n]$.
\end{lem}

\begin{proof}
  If $\CZ$ is sufficiently large, then by Ramsey's Theorem one of the
  following two assertions holds.
  \begin{eroman}
    \item
    There are $Z_1,\ldots,Z_{2/p(\ell)}\in\CZ$ such that $Z_i\cup Z_j\neq X$ for all distinct
  $i,j\in\big[2/p(\ell)\big]$.
  \item
    There are $Z_1,\ldots,Z_n\in\CZ$ such that $Z_i\cup Z_j=X$ for all distinct $i,j\in[n]$.
  \end{eroman}
  Lemma~\ref{lem:z2} rules out (i). Thus (ii) holds.
\end{proof}

Let us now turn to the sets in $\CY$. Recall that $Y,Y'\in\CY$ are
$X$-disjoint if $Y\cap Y'=\bar X$.

\begin{lem}\label{lem:z4}
  If there are distinct sets in $\CY$ that are not $X$-disjoint, then
  for all $Y\in\CY$,
  \[
  |Y\setminus\bar X|\ge \left(p(\ell)-p(\ell)^3\right)\cdot|X|.
  \]
\end{lem}

\begin{proof}
  Let $Y_1,Y_2\in\CY$ such that $Y_1\neq Y_2$ and $Y_1\cap Y_2\neq\bar X$, and let
  $Z_1:=\bar Y_1$ and $Z_2:=\bar Y_2$. Then $Z_1\cup Z_2\neq X$.
  Thus 
  \[
  |Y_1\setminus\bar X|\ge|Y_1\setminus Y_2|
  =|Z_2\setminus Z_1|=|Z_2|-|Z_1\cap Z_2|\ge
  \big(p(\ell)-p(\ell)^3\big)\cdot|X|,
  \]
  where the last inequality follows from Lemma~\ref{lem:z1}. For all
  $Y\in\CY$ we have $|Y|=|Y_1|$ and thus $|Y\setminus\bar
  X|=|Y_1\setminus\bar X|$.
\end{proof}

\begin{proof}[Proof of Lemma~\ref{lem:z0}]
  We let 
  \[
  n(\ell):=\left\lfloor\frac{1}{\big(p(\ell)-p(\ell)^3\big)}\right\rfloor+1,
  \]
  and choose $m=m(\ell,n(\ell))$ according to Lemma~\ref{lem:z3}. Suppose
  that $|\CZ|=|\CY|>m$. Then there are $Z_1,\ldots,Z_{n(\ell)}$ such that
  $Z_i\cup Z_j=X$ for all distinct $i,j\in[n(\ell)]$. For all $i\in[n(\ell)]$, let
  $Y_i=\bar Z_i$. Then $Y_1,\ldots,Y_{n(\ell)}$ are mutually
  $X$-disjoint. Thus
  \[
  \sum_{i=1}^{n(\ell)}|Y_i\setminus\bar X|\le|X|.
  \]
  By the choice of $n(\ell)$ and Lemma~\ref{lem:z4}, it follows that all sets in $\CY$ are
  mutually $X$-disjoint.

  To complete the proof, we let
  \[
  b_2(k_0,k_1):=\max_{0\le\ell\le
    k_1}m\left(\ell,n(\ell)\right).
  \]
\end{proof}

\begin{lem}\label{lem:1tan5}
  Suppose that the elements of $\CY$ are mutually $X$-disjoint.
  Then for all $\CY_0\subseteq\CY$,
  \[
  \kct\Big(\bigcup_{Y\in\CY_0}Y\Big)\le \ell,
  \]
  with equality for all $\CY_0\subset\CY$.
\end{lem}

\begin{proof}
  We prove by induction on $i\le|\CY|$ that for all
  $Y_1,\ldots,Y_i\in\CY$,
 \begin{equation}\label{eq:9}
  \kct(Y_1\cup\ldots\cup Y_i)\le \ell,
  \end{equation}
  with equality if $i<|\CY|$. 

  The base step $i=1$ is trivial. For $i=2$, let $Y_1,Y_2\in\CY$. Then
  $Y_1\cap Y_2=\bar X$ and thus $\kct(Y_1\cap Y_2)=\kct(X)=k_1\ge \ell$. By
  submodularity, $\kct(Y_1\cup Y_2)\le \ell$. 

  Now let $2\le i<|\CY|$, and suppose that
  \begin{equation}\label{eq:10}
    \kct(Y_1\cup\ldots\cup Y_{i-1})=\ell
  \end{equation}
  for all $Y_1,\ldots,Y_{i-1}\in\CY$ and \eqref{eq:9} for all
  $Y_1,\ldots,Y_i\in\CY$.

  Let $Y_1,\ldots,Y_{i+1}\in\CY$ be mutually distinct. By posimodularity
  \[
  \kct(Y_1\cup\ldots\cup Y_i)+\kct(Y_i\cup Y_{i+1})\ge
  \kct(Y_1\cup\ldots\cup Y_{i-1})+\kct(Y_{i+1}).
  \]
  As $\kct(Y_1\cup\ldots\cup Y_{i-1})=\kct(Y_{i+1})=\ell$ by
  \eqref{eq:10} and $\kct(Y_i\cup Y_{i+1})\le \ell$ by \eqref{eq:9},
  it follows that $\kct(Y_1\cup\ldots\cup Y_i)\ge \ell$, which combined
  with  \eqref{eq:9} implies equality.

  Furthermore, by submodularity,
  \[
  \kct(Y_1\cup\ldots\cup Y_i)+\kct(Y_i\cup Y_{i+1})\ge
  \kct(Y_i)+\kct(Y_1\cup\ldots\cup Y_{i+1}).
  \]
  As $\kct(Y_1\cup\ldots\cup Y_{i})=\kct(Y_i\cup
  Y_{i+1})=\kct(Y_1)=\ell$, this implies $\kct(Y_1\cup\ldots\cup
  Y_{i+1})\le \ell$.
\end{proof}

\subsection{Proof of Lemma~\ref{lem:1tan}}
\label{sec:1tan3}

We continue to use the notation of Section~\ref{sec:1tan2}.
Essentially, the lemmas proved there show how to use the family $\CY$
to obtain the desired partition of $X$. 
The main question that remains to be solved is how to compute $\CY$.

\begin{lem}\label{lem:1tan6}
  There is a polynomial time algorithm that, given $X$ and oracle
  access to $\kct$, computes $\CY$.
\end{lem}

\begin{proof}
  Let $\CZ^*$ be the family of all $Z\subseteq A$ satisfying the following
  three conditions:
  \begin{eroman}
  \item $Z$ satisfies \eqref{eq:8}, that is, $p(\ell)\cdot|X|\le|Z|<|X|$.
  \item[(ii')] $\kct_{\min}(Z,\bar X)=\ell$;
  \item[(iii')] there are a set $Z_0\subseteq X$ of size
    $|Z_0|\le\ell$ and an element $x\in X$ such
    that $Z$ is a rightmost minimum $(Z_0,\bar X\cup\{x\})$-separation.
  \end{eroman}
  Let $m:=\max\big\{|Z|\bigmid Z\in\CZ^*\big\}$.
  
  \begin{claim}
    \[
    \CZ=\big\{Z\in\CZ^*\bigmid |Z|=m\big\}.
    \]
    \proof
    We first prove that $\CZ\subseteq\CZ^*$. Let $Z\in\CZ$. Clause (i)
    in the definition of $\CZ$ is the same as clause (i)
    above.

    If there was some $Z'$ such that $Z\subseteq
    Z'\subseteq X$ and $\kct(Z')<\ell$, then $Z'\subset X$,
    because $\kct(X)=k_1\ge\ell$, and $Z'$ would also satisfy
    \eqref{eq:8}, because $|Z'|\ge |Z|$. Thus (ii) would be
    violated. This proves (ii').

    To see that $Z$ satisfies (iii'), let $Z_0\subseteq Z$ be inclusionwise
    minimal such that $\kct_{\min}(Z_0,\bar X)=\ell$. Suppose for
    contradiction that $|Z_0|=n>\ell$, and let $z_1,\ldots,z_n$ be
    an enumeration of $Z_0$. For every $i\in[n]$, let
    $Z^i=\{z_1,\ldots,z_i\}$. Then
    $\kct_{\min}(Z^i,\bar X)\leq \kct_{\min}(Z^{i+1},\bar X)$ for all
    $i<\ell$, because $\kct_{\min}$ is monotone in the first
    argument. As $\kct_{\min}(Z_0,\bar X)=\ell$, there is an
    $i<\ell$ such that
    $\kct_{\min}(Z^i,\bar X)= \kct_{\min}(Z^{i+1},\bar X)$. By the submodularity of $\kct_{\min}$ in the first argument,
    \[
    \kct_{\min}(Z_0\setminus \{z_{i+1}\},\bar X) +
    \kct_{\min}(Z^{i+1},\bar X) \geq \kct_{\min}(Z^{i},\bar X)+
    \kct_{\min}(Z_0,\bar X).
    \]
    It follows that $\kct_{\min}(Z_0\setminus
    \{z_{i+1}\},\bar X)= \kct_{\min}(Z_0,\bar X)$, contradicting the minimality of~$Z_0$.
    This proves that $|Z_0|\le\ell$.

   Let $x\in X\setminus Z$. Then $Z_0\subseteq Z\subseteq X\setminus\{x\}\subseteq X$,
    and $\kct_{\min}(Z_0,\bar X)=\ell$ and $\kappa(Z)=\ell$ imply $\kct_{\min}(Z_0,\bar
    X\cup\{x\})=\ell$. Thus $Z$ is a minimum
    $(Z_0,\bar X\cup\{x\})$-separation, and now clause (iii) in the definition of
    $\CZ$ (the maximality of $|Z|$) implies that $Z$ is
    rightmost. This completes the proof of (iii') and thus of the
    inclusion $\CZ\subseteq\CZ^*$. The maximality of the elements of
    $\CZ$ (clause (iii) in the definition) then implies that
   \[
    \CZ\subseteq\big\{Z\in\CZ^*\bigmid |Z|=m\big\}.
    \]
    
    To prove the converse inclusion, let $Z\in\CZ^*$ with
    $|Z|=m$. Then by (iii'), $Z\subset\bar X$. Clauses (i) and (ii') above
    imply clauses (i) and (ii) in the definition of $\CZ$. 

    Suppose that there is some $Z'\subset X$ satisfying (i) and (ii)
    such that $|Z'|>|Z|$. Choose such a $Z'$ of maximum size. Then
    $Z'\in\CZ$, and thus $|Z'|=m=|Z|$. This is a contradiction.
    \uend
  \end{claim}

  It is easy to see that $\CZ^*$ can be computed in polynomial time,
  and this implies that $\CZ$ and thus $\CY$ can be computed in
  polynomial time.
\end{proof}

\begin{proof}[Proof of Lemma~\ref{lem:1tan}]
  Recall that $k_0\le k$. We let 
  \begin{align}
  f_1:=f_1({k},k_1)&:=\frac{1}{p(0)-p(0)^3}\\
  \intertext{and}
  a_1:=a_1({k},k_1)&:=\max\{k,2k_1\cdot b_2({k},k_1)\},\\
  b_1:=b_1({k}.k_1)&:=\max\{6(k+k_1),2^{b_2({k},k_1)}\},
  \end{align}
  where
  $b_2:=b_2({k},k_1)$ is chosen according to Lemma~\ref{lem:z0}.

  If $|X|< 6(k+k_1)$, we simply partition $X$ into 1-element
  sets. Note that $\kct(\{b\})=\kappa(\{b\})\le k$ for all
  $b\in B$, because $\bw(\kappa)\le k\le a_1$, and
  $\kappa(\{c_i\})=\kappa(C_i)<k\le a_1$ for $0\le i\le m$ by
  Assumption~\ref{ass:1}(3). Thus (i) is satisfied. 

  In the following, we assume that $|X|\ge 6(k+k_1)\ge 6k_2$. This is the assumption needed for the
  previous results.

  \begin{cs}
    \case1
    There are distinct $Y,Y'\in\CY$ that are not $X$-disjoint.\\
     Then $|\CY|\le b_2(k,k_2)$ by Lemma~\ref{lem:z0} and
     \begin{equation}
       \label{eq:11}
     |Y\cap X|>(p(\ell)-p(\ell)^3)|X|\ge\frac{1}{f_1}\cdot|X|
     \end{equation}
     by Lemma~\ref{lem:z4}. Moreover, $\bar Y\in\CZ$ for all
     $Y\in\CY$, which implies
     \begin{equation}
       \label{eq:12}
        |\bar Y|\ge p(\ell)\cdot |X|\ge
     \frac{1}{f_1}\cdot|X|.
     \end{equation}
     Let $Y^+_1,\ldots,Y^+_n$ be an enumeration of all sets
     $|Y\cap X|$ for $Y\in\CY$. Note that $n\le b_2({k},k_1)$. For
     every $i\in[n]$, let $Y^{-}_i:=X\setminus Y_i^+$. For every
     $i\in[n]$, we have $\kct(Y^{-}_i)=\ell$ and, by submodularity,
     $\kct(Y^{+}_i)\le k_1+\ell\le 2k_1$.

Let
     $X_1,\ldots,X_b$ be a list of all nonempty sets of the form
     \[
     \bigcap_{i=1}^nY_i^{\sigma(i)},
     \]
     for some function $\sigma:[n]\to\{+,-\}$. Then $b\le 2^n\le
     b_1$. Submodularity implies that
     \[
     \kct(X_i)\le 2k_1n\le a_1.
     \]
     It follows from \eqref{eq:11} and \eqref{eq:12} that
     $|Y_i^{\sigma}|\le(1-1/f_1)|X|$ for all $i\in[n]$ and
     $\sigma\in\{+,-\}$.
     Thus the partition $X_1,\ldots,X_b$ satisfies assertion (i) of
     the lemma.

    \case 2
    The elements of $\CY$ are 
    mutually $X$-disjoint.\\
    We let 
    \[
    X_0:=\bigcap_{Y\in\CY}\bar Y,
    \]
    and we let $X_1,\ldots,X_n$ be an enumeration of the sets $Y\cap
    X$ for $Y\in\CY$. Note that the sets $X_0,\ldots,X_n$ form a
    partition of $X$. It follows from Lemma~\ref{lem:1tan5} that
    $\kct(X_0)\le \ell\le k$ and  $\kct\left(\bigcup_{i\in
        I}X_i\right)\le k+\ell\le 2k$ for every set $I\subseteq[n]$. It follows
    from \eqref{eq:8} that $|X_i|\le(1-1/f_1)|X|$ for every
    $i\in[n]$. 

    Thus the partition $X_0,\ldots,X_n$ satisfies assertion (i) of
     the lemma.
  \end{cs}
  It follows from Lemma~\ref{lem:1tan6} that in both cases the
  partition can be computed in polynomial time.
\end{proof}

\section{The Non-Well-Linked Case}
\label{sec:nwl}

\subsection{Partitioning with Respect to an Independent Set}
\label{sec:nwl1}

In this section, we make Assumptions~\ref{ass:1} again (but not
Assumptions~\ref{ass:2} and \ref{ass:3}). 

Let
$X\subseteq \Act$ such that 
\begin{equation}
  \label{eq:13}
  k_1:=\kct(X)\ge(3k+2)\cdot k. 
\end{equation}
We define a function $\lambda:2^{\bar X}\to\NN$ by letting
\[
\lambda(Y):=\kct_{\min}(Y,X)
\]
for all $Y\subseteq\bar X$. Then $\lambda$ is submodular and monotone,
and we have $\lambda(\emptyset)=0$. Such a function is known as an
\emph{integer polymatroid}. It induces a
matroid $\CM(\lambda)$ on $\bar X$ whose independent sets are all
$Y\subseteq\bar X$ satisfying
\begin{equation}
  \label{eq:14}
  |Z|\le\lambda(Z)\text{ for all }Z\subseteq Y
\end{equation}
(see \cite{oxl11}, Proposition 12.1.2).
The rank function $r_\lambda$ of $\CM(\lambda)$ is defined by
\[
r_{\lambda}(Y):=\min\big\{\lambda(Z)+|Y\setminus Z|\bigmid Z\subseteq
Y\big\}.
\]
(see \cite{oxl11}, Proposition 12.1.7). Observe that for all $y\in\bar
X$ we have $\lambda(\{y\})\le\kct(\{y\})\le k$. If $y\in B$ then this
holds because $\kct(\{y\})=\kappa(\{y\})\le\bw(\kappa)$, and if
$y=c_i$ it follows from Assumption~\ref{ass:1}(3). A
straightforward induction based on the submodularity of $\lambda$ then
implies that $\lambda(Y)\le k|Y|$ for all $Y\subseteq\bar X$. Thus for all $Z\subseteq\bar X$,
\[
\lambda(Z)\ge\lambda(\bar X)-\lambda(\bar X\setminus Z)\ge k_1-k\cdot
|\bar X\setminus Z|,
\]
which implies $\lambda(Z)+k|\bar X\setminus Z|\ge k_1$ and hence
$\lambda(Z)+|\bar X\setminus Z|\ge k_1/k$.
By the definition of $r_\lambda$, we get
\[
r_\lambda(\bar X)\ge\frac{k_1}{k}\ge (3k+2).
\]
Thus there is a set $Y'\subseteq \bar X$ of size $|Y'|=3k+2$ that is an
independent set of $\CM(\lambda)$. As all subsets of an independent set
are independent as well, there is an independent set $Y\subseteq\bar
X\setminus\{c_0\}$ of size $|Y|=3k+1$. We keep such a set $Y$
fixed in the following.

\begin{lem}\label{lem:nwl1}
  Let $Z\subseteq \Act$ such that
  $\kct(Z)< |Y\setminus Z|$.
  
  Then
  \begin{equation*}
    \label{eq:15}
  \kct\left(X\cap Z\right)<\kct(X).
  \end{equation*}
\end{lem}

\begin{proof}
  As $Y$ is independent, we have $|Y\setminus Z|\le\lambda(Y\setminus
  Z)$. As $Y\setminus Z\subseteq\bar X\setminus Z\subseteq\bar X$, we have
  $\lambda(Y\setminus Z)\le\kct(\bar X\setminus Z)$. Thus
  $\kct(\bar X\setminus Z)-|Y\setminus Z|\ge 0$ and therefore
  \begin{align*}
    \kct(X\cap Z)&\le \kct(X\cap Z)+\kct(\bar X\setminus
                     Z)-|Y\setminus Z|\\
    &=\kct(X\cap Z)+\kct(X\cup Z)-|Y\setminus
      Z|&\text{(symmetry)}\\
    &\le\kct(X)+\kct(Z)-|Y\setminus Z|&\text{(submodularity)}\\
    &<\kct(X)+|Y\setminus Z|-|Y\setminus Z|&\text{(assumption
                                                    of the lemma)}\\
    &=\kct(X).
  \end{align*}
\end{proof}

\begin{lem}
  There is a set $Z\subseteq\Act$ such that $\kct(Z)\le k$ and
  \begin{equation}
    \label{eq:16}
    \kct(Z)<\min\{|Y\cap Z|,|Y\setminus Z|\}.
  \end{equation}
  Furthermore, we can compute such a set $Z$ in polynomial time (for
  fixed $k$).
\end{lem}

\begin{proof}
        We define a weight function $\phi:A\to \RR$ as follows:
      \[
      \phi(x):=
      \begin{cases}
        \frac{1}{|C_i|}&\text{if $x\in C_i$ for some $i\in\{0,\ldots,m\}$ such
          that $c_i\in Y$},\\
        1&\text{if }x\in Y\cap B,\\
        0&\text{otherwise}.
      \end{cases}
      \]

      \begin{claim}
        There is a $Z\subseteq A$ such that  $\kappa(Z)\le k$ and 
        \[
        \min\{\phi(Y\expand\cap Z),\; \phi(Y\expand\setminus Z)\}>k.
        \]
        \proof
        Suppose for contradiction that for all $Z\subseteq A$ such
        that  $\kappa(Z)\le k$, either $\phi(Y\expand\cap Z)\le
        k$ or
        $\phi(Y\expand\setminus Z)\le k$. Let
        \[
        \CT:=\{Z\subseteq A\mid \kappa(Z)\le k,\phi(Y\expand\cap Z)\ge 2k+1\}.
        \]
        Then $\CT$ is a $\kappa$-tangle of order
        $k+1$. Indeed,  it obviously satisfies \ref{li:t0}. It
        satisfies \ref{li:t1}, because 
        \[
        \phi(Y\expand\cap Z)+\phi(Y\expand\setminus
        Z)=\phi(Y\expand)=|Y|=3k+1.
        \]
        It satisfies \ref{li:t2}, because $2k+1>(2/3)|Y|$, and it
        satisfies \ref{li:t3} because $\phi(\{x\})\le 1<2k+1$ for all
        $x\in A$. 

        However, as $\bw(\kappa)=k$ no tangle of order $k+1$
        exists.
        \uend
      \end{claim}

      Let $\ell$ be minimum such that there is a $Z\subseteq A$ such
      that $\kappa(Z)\le \ell$ and
      \begin{equation}
        \label{eq:17}
        \min\{\phi(Y\expand\cap Z),\; \phi(Y\expand\setminus Z)\}>\ell.
      \end{equation}
      Let $Z\subseteq A$ such that $\kappa(Z)\le \ell$ and
      \eqref{eq:17}.

      Without loss of generality we may assume that either $Z\cap
      C_0=\emptyset$ or $C_0\subseteq Z$. To see this, suppose
    that neither $Z\cap C_0=\emptyset$ nor $C_0\subseteq Z$, or
    equivalently, neither $Z\subseteq\bar C_0$ nor
    $\bar Z\subseteq \bar C_0$. By Assumption~\ref{ass:1}(4), $\bar C_0$
    is a minimum $(\CT_0,\CT_0')$-separation. Thus by
    Lemma~\ref{lem:tanglesep3}(1) (applied to $X=\bar C_0$), either
    $\kappa(Z\cap\bar C_0)\le\kappa(Z)$ or
    $\kappa(\bar Z\cap\bar C_0)\le\kappa(Z)$. As $c_0\not\in Y$, we
    have $\phi(Z\cap\bar C_0)=\phi(Z\cup C_0)=\phi(Z)$. Thus
    if $\ell':=\kappa(Z\cap\bar C_0)\le\kappa(Z)$, then
    $Z':=Z\cap\bar C_0$ satisfies \eqref{eq:17} with $Z',\ell'$
    instead of $Z,\ell$, and if
    $\ell':=\kappa(\bar Z\cap\bar C_0)\le\kappa(Z)$, then
    $Z':=\bar Z\cap\bar C_0$ satisfies \eqref{eq:17} with $Z',\ell'$
    instead of $Z,\ell$. In both cases, we have
    $Z'\cap C_0=\emptyset$. This justifies the assumption that either
    $Z\cap C_0=\emptyset$ or $C_0\subseteq Z$.

      \begin{claim}[resume]
        For all $i\in[m]$, either $C_i\cap Z=\emptyset$ or
        $C_i\subseteq Z$.

        \proof
        Suppose for contradiction that neither $C_i\cap Z=\emptyset$ nor
        $C_i\subseteq Z$. Then neither $Z\subseteq \bar C_i$ nor $\bar
        Z\subseteq \bar C_i$. By
        Assumption~\ref{ass:1}(4) and Lemma~\ref{lem:tanglesep3},
        either $\kappa(Z\cap \bar C_i)<\kappa(Y)$ or $\kappa(\bar
        Z\cap \bar C_i)<\kappa(Y)$. 

        Without loss of generality we assume that
        $\ell':=\kappa(Z\cap \bar C_i)<\kappa(Y)=\ell$. Let $Z':=Z\cap
        \bar C_i$.
        Then
        \[
        \phi(Y\expand\cap Z')\ge \phi(Y\expand\cap
        Z)-1\ge\ell-1\ge\ell'
        \]
        and
          \[
        \phi(Y\expand\setminus Z')\ge \phi(Y\expand\setminus
        Z)\ge\ell\ge\ell'.
        \]
        This contradicts the minimality of $\ell$.
        \uend
      \end{claim}
        
      It follows that there is a $Z^\vee\subseteq\Act$ such
      that $Z=Z^\vee\expand$. Then
      \[
      |Y\cap Z^\vee|=\phi(Y\expand\cap Z)>\ell=\kappa(Z)=\kct(Z^\vee)
      \]
      and, similarly, $|Y\setminus Z^\vee|>\kct(Z^\vee)$.

      We can compute a set $Z$ satisfying \eqref{eq:16} in polynomial
      time as follows: for every $Z_0\subseteq Y$ we compute a leftmost
minimum $(Z_0,Y\setminus Z_0)$ separation $Z$ until we find one with
$\kct(Z)<|Z_0|=|Y\cap Z|$ and $\kct(Z)<|Y\setminus Z_0|=|Y\setminus Z|$.
\end{proof}

\begin{lem}\label{lem:nwl2}
  Let $Z\subseteq \Act$ such that $\kct(Z)<\min\{|Y\cap Z|,|Y\setminus
  X|\}$. Then $X\cap Z,X\setminus Z$ is a partition of $X$ into two
  nonempty sets with $\kct(X\cap Z),\kct(X\setminus Z)<\kct(X)$.
\end{lem}

\begin{proof}
  As $Y$ is independent in $\CM(\lambda)$, for each $Z$ with $\kct(Z)<\min\{|Y\cap Z|,|Y\setminus
  X|\}$ we have $X\cap Z\neq\emptyset$, because for
$Z'\subseteq\bar X$ we have
$\kct(Z')\ge\lambda(Y\cap Z')\ge|Y\cap Z'|$ by \eqref{eq:14}. By
Lemma~\ref{lem:nwl1}, we have $\kct(X\cap Z)<\kct(X)$. By
symmetry, we also have $X\setminus Z=X\cap\bar Z\neq\emptyset$ and
$\kct(X\setminus Z)<\kct(X)$.
\end{proof}

\subsection{A Canonical Family of Partitions}
\label{sec:nwl2}

While so far, all our constructions work for general connectivity
functions, in this section we need to restrict our attention to the
cut rank function of a graph. In addition to Assumptions~\ref{ass:1},
which we still maintain, we make the following assumption.

\begin{ass}\label{ass:4}
  There is a graph $G$ such that $\kappa=\rho_G$.
\end{ass}

As in the previous subsection, let $X\subseteq \Act$ such that
$k_1:=\kct(X)\ge(3k+2)k$. Then 
\[
\rk(M_{X\expand,\bar X\expand})=k_1.
\]
For every $W\subseteq\Act=V(G)\contract$, we let $G[W\expand]$ be the induced subgraph of $G$
with vertex set $W\expand$, and we let
$\kct_W:=\rho_{G[W\expand]}\contract$, where of course we contract only
those $C_i$ that are contained in $W\expand$. Observe that for every
$Z\subseteq W$ we have
\[
\kct_W(Z)=\rk\big(M_{Z\expand,(W\setminus Z)\expand}\big).
\]

By $\bar X^{\underline\ell}$ we denote the set of all $\ell$-tuples
of elements of $\bar X$ with mutually distinct entries.
For every $\ell\ge 1$, we shall define an
equivalence relation $\equiv_X^\ell$ on $\bar X^{\underline\ell}$ with index (that
is, number of equivalence classes) bounded in terms of $k_1$ and
$\ell$ such that the following holds.

\begin{lem}\label{lem:nwl3}
  Let $\mathbf w=(w_1,\ldots,w_\ell),\mathbf w'=(w'_1,\ldots,w'_\ell)\in\bar
  X^{\underline\ell}$ such that $\mathbf w\equiv^\ell_X\mathbf w'$, and let
  $W:=\{w_1,\ldots,w_\ell\}$, $W':=\{w_1,\ldots,w_\ell\}$.

  Let $Z\subseteq X\cup W$ and $Z':=(X\cap Z)\cup\{w'_i\mid
  i\in[\ell]\text{ such that }w_i\in Z\}$. Then
  \[
  \kct_{X\cup W}(Z)=\kct_{X\cup W'}(Z').
  \]
\end{lem}

Let us first consider the special case that no sets are contracted,
that is, $m=-1$ (this is not a case that we actually need to consider,
but it is helpful to explain the ideas). Then $\Act=A=V(G)$ and $\kct=\kappa=\rho_G$.
For $\mathbf
w=(w_1,\ldots,w_\ell),\mathbf w'=(w'_1,\ldots,w'_\ell)\in\bar X^{\underline\ell}$,
we let $\mathbf w\equiv_X^\ell\mathbf w'$ if for all $i\in[\ell]$ the
columns of the matrix $M_{X,\bar X}$ indexed by $w_i$ and $w_i'$ are
equal and for all $i,j\in[\ell]$ we have $w_iw_j\in E(G)\iff
w_i'w_j'\in E(G)$. 
We can rephrase these two conditions as follows:
\begin{eroman}
  \item
    For all $i\in[\ell]$ the matrices $M_{X,\{w_i\}}$ and
    $M_{X,\{w_i'\}}$ are equal.
  \item
    For all $i\in[\ell]$ the matrices $M_{\{w_1,\ldots,w_{i-1},w_{i+1},\ldots,w_\ell\},\{w_i\}}$ and
    $M_{\{w'_1,\ldots,w'_{i-1},w'_{i+1},\ldots,w'_\ell\},\{w'_i\}}$ are equal.
\end{eroman}
Let $W:=\{w_1,\ldots,w_\ell\}$ and $W':=\{w'_1,\ldots,w'_\ell\}$.
For all $Z\subseteq X\cup W$ and $Z':=(X\cap Z)\cup \{w_i'\mid w_i\in
Z\}\subseteq X\cup
W'$, condition (i) implies that 
\begin{align}
  \label{eq:18}
  M_{X\cap Z,W\setminus Z}=M_{X\cap Z',W'\setminus Z'},\\
  \label{eq:19}
  M_{X\setminus Z,W\cap Z}=M_{X\setminus Z',W'\cap Z'}.
\intertext{Condition (ii) implies}
  \label{eq:20}
  M_{W\cap Z,W\setminus Z}=M_{W'\cap Z',W'\setminus Z'}.
\end{align}
Lemma~\ref{lem:nwl3} (in the special case)
follows easily, because 
\begin{equation}
  \label{eq:21}
  \kct_{X\cup W}(Z)=\rho_{G[X\cup W]}(Z)=\rk(M_{Z,(X\cup W)\setminus Z})=
\rk\left(
\begin{pmatrix}
  M_{X\cap Z,X\setminus Z}&M_{X\cap Z,W\setminus Z}\\
  M_{W\cap Z,X\setminus Z}&M_{W\cap Z,W\setminus Z}
\end{pmatrix}
\right)
\end{equation}
and
  \begin{equation}
    \label{eq:22}
    \kct_{X\cup W'}(Z')=
\rk\left(
\begin{pmatrix}
  M_{X\cap Z',X\setminus Z'}&M_{X\cap Z',W'\setminus Z'}\\
  M_{W'\cap Z',X\setminus Z'}&M_{W'\cap Z',W'\setminus Z'}
\end{pmatrix}
\right).
  \end{equation}
Since $X\cap Z'=X\cap Z$ and $X\setminus Z'=X\setminus Z$, equations
\eqref{eq:18}, \eqref{eq:19}, and \eqref{eq:20} imply that the matrices in
the rightmost terms in \eqref{eq:21}
and \eqref{eq:22} are equal.

\medskip
Let us now turn to the general case. The situation is more
difficult here because the sets $C_i$ and hence the matrices involved
in our argument above, in particular the matrices $M_{Z\expand,W\setminus Z\expand}$ may have unbounded size (in
terms of $k$ and $\ell$). The crucial observation is that we can bound
the size of the $C_i$ in terms of $k\le k_1$, exploiting the fact
that $\rho_G(C_i)=\kct(\{c_i\})< k$. To simplify the notation, we
assume that $B=\{c_{m+1},\ldots,c_n\}$ for some $n\ge m$, and for
$i=m+1,\ldots,n$, we let $C_i:=\{c_i\}$, so that actually
$A=\bigcup_{i=0}^nC_i$ and $\Act=\{c_0,\ldots,c_n\}$.

Suppose that for some $i\in[m]$ there are distinct vertices $v,v'\in C_i$ such that for
all $w\in V(G)\setminus C_i$ we have $vw\in E(G)\iff v'w\in E(G)$. Let
$G':=G\setminus\{v'\}$ and $C_i':=C_i\setminus\{v'\}$ and $C_j':=C_j$
for $j\neq i$. Then contracting $C_1',\ldots,C_m'$ in $G'$ has the
same effect as contracting $C_1,\ldots,C_m$
in $G$, that is,
\begin{align*}
\Act&=V(G)\contract_{C_1,\ldots,C_m}=V(G')\contract_{C_1',\ldots,C_m'},\\
\kct&=\rho_G\contract_{C_1,\ldots,C_m}=\rho_{G'}\contract_{C_1',\ldots,C_m'}.
\end{align*}
By repeating this construction we arrive at an induced subgraph
$G''\subseteq G$ and a partition $C_1'',\ldots,C_n''\subseteq V(G'')$,
where $C_i''\subseteq C_i$ for $0\le i\le m$ and $C_i''=C_i=\{c_i\}$
for $m+1\le i\le n$, such
that 
\begin{align*}
\Act&=V(G'')\contract_{C_1'',\ldots,C_m''},\\
\kct&=\rho_{G''}\contract_{C_1'',\ldots,C_m''},
\end{align*}
and for all $i\in[n]$ and distinct $v,v'\in C_i''$ there is a $w\in
V(G'')\setminus C_i''$ such that $vw\in E(G)\not\Leftrightarrow v'w\in
E(G)$.  Observe that the construction of $G''$ and
$C_1'',\ldots,C_m''$ from $G$ and $C_1,\ldots,C_m$ is canonical and
can be carried out in polynomial time. To simplify the notation, in
the following we assume that $G''=G$ and $C_i''=C_i$ for all
$i\in[n]$.

Let $i\in[n]$. As for all distinct $v,v'\in C_i$ there is a $w\in\bar C_i$ such that $vw\in E(G)\not\Leftrightarrow v'w\in
E(G)$, the rows of the matrix $M_{C_i,\bar C_i}$ are mutually distinct.
As $k>\kct(\{c_i\})=\rho_G(C_i)=\rk(M_{C_i,\bar C_i})$, the
matrix $M_{C_i,\bar C_i}$, being a matrix over $\mathbb F_2$, has at
most $2^{k-1}$ distinct rows. This implies that 
\[
|C_i|\le 2^{k-1}.
\]
Now we are ready to define the equivalence relation
$\equiv^\ell_X$. For this we let  $\mathbf
w=(w_1,\ldots,w_\ell),\mathbf w'=(w'_1,\ldots,w'_\ell)'\in\bar X^{\underline\ell}$. To
simplify the notation, for
every $i\in[\ell]$ we let $w_i\expand:=\{w_i\}\expand$, that is, if
$w_i=c_j$ then $w_i\expand=C_j$. Similarly, we let
$w'_i\expand:=\{w'_i\}\expand$. We let $\mathbf w\equiv_X^\ell\mathbf w'$ if
for every $i\in[\ell]$ there are linear orders $\le_i$ of $w_i\expand$
and $\le_i'$ of $w'_i\expand$ such that the following two conditions
are satisfied.
\begin{eroman}
  \item
    For all $i\in[\ell]$ the matrices $M_{X,w_i\expand}$ and
    $M_{X,w'_i\expand}$ are equal if the columns of the matrices are
    ordered according to the linear orders $\le_i$ and $\le_i'$, respectively.
  \item
    For all $i\in[\ell]$ the matrices $M_{w_1\expand\cup\ldots\cup
      w_{i-1}\expand\cup w_{i+1}\expand\cup \ldots \cup w_\ell\expand,w_i\expand}$
    and
    $M_{w'_1\expand\cup\ldots\cup
      w'_{i-1}\expand\cup w'_{i+1}\expand\cup\ldots\cup
      w'_\ell\expand,w'_i\expand}$ are equal if 
    \begin{enumerate}[label=(ii-\alph*)]
    \item the rows of the matrices are ordered lexicographically
      according to the natural order on the indices $j$ of the $w_j$, $w_j'$
      and, within the sets $w_j\expand,w'_j\expand$, according to the
      linear orders $\le_j,\le_j'$, respectively;
    \item columns of the matrices are
    ordered according to the linear orders $\le_i$, $\le_i'$, respectively.
    \end{enumerate}
\end{eroman}

\begin{proof}[Proof of Lemma~\ref{lem:nwl3}]
  We have
 \begin{align}
  \notag
  \kct_{X\cup W}(Z)&=\rho_{G[(X\cup
                     W)\expand]}(Z\expand)=\rk(M_{Z\expand,((X\cup W)\setminus Z)\expand})\\
   \label{eq:23}
   &=
\rk\left(
\begin{pmatrix}
  M_{(X\cap Z)\expand,(X\setminus Z)\expand}&M_{(X\cap Z)\expand,(W\setminus Z)\expand}\\
  M_{(W\cap Z)\expand,(X\setminus Z)\expand}&M_{(W\cap Z)\expand,(W\setminus Z)\expand}
\end{pmatrix}
\right).
\end{align}
Similarly, for $Z':=(X\cap Z)\cup\{w'_i\mid w_i\in Z\}$,
  \begin{equation}
    \label{eq:24}
    \kct_{X\cup W}(Z)=
\rk\left(
\begin{pmatrix}
  M_{(X\cap Z)\expand,(X\setminus Z)\expand}&M_{(X\cap Z)\expand,(W'\setminus Z')\expand}\\
  M_{(W'\cap Z')\expand,(X\setminus Z)\expand}&M_{(W'\cap Z')\expand,(W'\setminus Z')\expand}
\end{pmatrix}
\right),
  \end{equation}
  where we use the fact that $X\cap Z=X\cap Z'$ and $X\setminus
  Z=X\setminus Z'$.
  We may assume that in all these matrices the rows and columns indexed by entries
  of $W,W'$ are ordered lexicographically according to the indices of
  the $i$ and the orders $\le_i,\le_i'$ as in (ii-a) above.

  Observe that (i) implies
  \begin{align*}
    M_{(X\cap Z)\expand,(W\setminus
    Z)\expand}&=M_{(X\cap Z)\expand,(W'\setminus Z')\expand},\\
  M_{(W\cap Z)\expand,(X\setminus Z)\expand}&=M_{(W'\cap
    Z')\expand,(X\setminus Z)\expand}.
  \end{align*}
  Furthermore, (ii) implies that 
  \[
  M_{(W\cap
    Z)\expand,(W\setminus Z)\expand}=M_{(W'\cap
    Z')\expand,(W'\setminus Z')\expand}. 
  \]
  Thus the matrices in
the rightmost terms in \eqref{eq:23}
and \eqref{eq:24} are equal.
\end{proof}

The following lemma collects further useful properties of the
equivalence relation $\equiv_X^{\ell}$.

\begin{lem}\label{lem:nwl4}
  Let $\ell\ge 1$
  \begin{enumerate}
  \item Given $X$, the equivalence relation $\equiv_X^\ell$ is
    canonical.
  \item There is a $e_1=e_1(k_1,\ell)$ (independent of $\kct$) such
    that the index of $\equiv_X^\ell$ is at most $e_1(k_1,\ell)$.
  \item Given the graph $G$, the sets $C_0,\ldots,C_m$, and the set
    $X$, the equivalence relation $\equiv_X^\ell$ can be
      computed in polynomial time (for fixed $k_1$ and $\ell$).
  \end{enumerate}
\end{lem}

\begin{proof}
  (1) and (3) are obvious from the construction.

  (2) follows easily from the following two observations.
  \begin{itemize}
  \item For every $\mathbf w=(w_1,\ldots,w_\ell)\in\bar X^{\ell}$ the set
    $\bigcup_{i=1}^\ell w_i\expand$ has at most $\ell\cdot2^{k-1}\le
  \ell\cdot2^{k_1-1}$ elements.
\item The matrix $M_{X\expand,\bar X\expand}$ has rank $k_1$ and thus at
  most $2^{k_1}$ different columns.\qedhere
  \end{itemize}
\end{proof}

We are now ready to prove the main result of this section. Let 
\begin{equation}
  \label{eq:25}
  e_2(k_1):=\max\{e_1(k_1,\ell)\mid 1\le\ell\le 2^{k_1}\},
\end{equation}
where $e_1$ is chosen according to Lemma~\ref{lem:nwl4}.

\begin{lem}\label{lem:nwl}
  Let $k_1\ge (3k+2)k$ and $e_2=e_2(k_1)$. Then for every
  $X\subseteq \Act$ of order $\kct(X)=k_1$ there is a canonical family of
  $e\le e_2$ partitions $X^{(i)}_1,X^{(i)}_2$ of $X$, for $1\le i\le e$, such that $\kct(X^{(i)}_j)<k$ for
  $1\le i\le e$, $j=1,2$.

  Furthermore, given $X$ and oracle access to $\kct$, the family of
  partitions can be computed in polynomial time (for fixed $k,k_1$).
\end{lem}

\begin{proof}
  Let $\ell$ be the number of distinct columns of $M_{X\expand,\bar
    X\expand}$. Observe that $k\le\ell\le2^{k_1}$.  We call a set
  $W\subseteq\bar X$ \emph{complete} if all columns of the matrix
  $M_{X\expand,\bar X\expand}$ already appear in the matrix
  $M_{X\expand,W\expand}$. A tuple $\mathbf w\in\bar X^{\underline\ell}$
  is \emph{complete} if the set of its entries is complete. Observe
  that if $\mathbf w$ is complete and $\mathbf w'\equiv_X^\ell\mathbf w$, then
  $\mathbf w'$ is complete as well.

  Let $\mathbf w^{(1)},\ldots,\mathbf w^{(e)}$ be a system of representatives of the
  $\equiv_X^\ell$-equivalence classes consisting of complete tuples.
  Note that $e\le e_2(k_1)$. For $1\le i\le e$, let $W^{(i)}$ be the set of
  entries of $\mathbf w^{(i)}$ and $\Act^{(i)}:=X\cup W^{(i)}$ and
  $\kct^{(i)}:=\kct_{\Act^{(i)}}$. By Lemma~\ref{lem:nwl3}, up to
  renaming of the elements, the connectivity function $\kct^{(i)}$
  only depends on the equivalence class of $\mathbf w^{(i)}$ and not on
  the choice of the specific tuple. Thus, up to renaming, the family
  of connectivity function $\kct^{(i)}$ is canonical.

  \begin{claim}[1]
    Let $i\in[e]$.
    For each $Z\subseteq X$, we have $\kct^{(i)}(Z)=\kct(Z)$.

    \proof
    This follows from the completeness of $W^{(i)}$.
    \uend
  \end{claim}

  In particular, we have $\kct^{(i)}(X)=\kct(X)=k_1$.
  Now we apply the construction of Section~\ref{sec:nwl1} to $\kct^{(i)}$.
  We define $\lambda^{(i)}:2^{W^{(i)}}\to\NN$ by
  $\lambda^{(i)}(Y):=\kct^{(i)}_{\min}(Y,X)$ and let $\CM(\lambda^{(i)})$ be the
  matroid induced by $\lambda^{(i)}$. Note that the order of the
  entries of the tuple $\mathbf
  w^{(i)}$ gives us a linear order $\le^{(i)}$ on $W^{(i)}$. We let
  $Y^{(i)}$ be
  the lexicographically first subset of $W^{(i)}\setminus\{c_0\}$ of size $3k+1$
  that is independent in $\CM(\lambda^{(i)})$.
  We let $Z_0^{(i)}$ be the lexicographically first subset of
  $Y^{(i)}$ such that for the leftmost minimum
  $(Z_0^{(i)},Y^{(i)}\setminus Z_0^{(i)})$-separation $Z^{(i)}$ we have
  \begin{equation}
    \label{eq:26}
    \kct(Z^{(i)})<\min\{|Y^{(i)}\cap Z^{(i)}|,|Y^{(i)}\setminus Z^{(i)}|\},
  \end{equation}
  and we let $X^{(i)}_1:=X\cap Z^{(i)}$ and $X^{(i)}_2:=X\setminus
  Z^{(i)}$. By Lemma~\ref{lem:nwl2} we have
  $\kct^{(i)}(X^{(i)}_j)<\kct^{(i)}(X)$, and by Claim~1 this
  implies $\kct (X_j)<\kct(X)$.

  Clearly, the construction is canonical and can be carried out in
  polynomial time.
\end{proof}

\begin{cor}\label{cor:nwl}
  For every $k_1$ there is a
  $c_1=c_1(k,k_1)$ such that for every $X\subseteq \Act$ of order
  $\kct(X)=k_1$ there is a canonical partial treelike decomposition
  $(T_X,\gamma_X)$ with the following properties.
  \begin{eroman}
    \item $T_X$ is a directed tree.
    \item $\gamma_X(r)=X$ for the root $r$ of $T_X$.
    \item $\bigcup_{t\in L(T_X)}\gamma_X(t)=X$ (but the sets
      $\gamma_X(t)$ for the leaves $t\in L(T_X)$ are not necessarily disjoint).
    \item $\kct(\gamma_X(t))\le k_1$ for all $t\in V(T_X)$.
    \item $\kct(\gamma_X(t))<(3k+2)k$ for all leaves $t\in L(T_X)$.
    \item $T_X$ has at most $c_1(k,k_1)$ nodes.
 \end{eroman}
  Furthermore, given $X$ and oracle access to $\kct$, the
  decomposition $(T_X,\gamma_X)$ can be computed in polynomial time (for fixed $k,k_1$).
\end{cor}

\begin{proof}
  We define $c_1$ inductively by letting $c_1(k,k_1)=1$ for all $k_1<(3k+2)k$
  and 
  \[
  c_1(k,k_1)=4e_2(k_1)\cdot c_1(k,k_1-1)
  \]
  for every $k_1\ge (3k+2)k$.

  We construct the decomposition $T_X$ recursively as follows: we start with a
  root $r$ and let $\gamma_X(r):=X$. If $\kct(X)<(3k+2)k$, this is
  the whole decomposition. So suppose that $\kct(X)\ge (3k+2)k$. Then we
  choose a canonical family of $e\le e_1(k_1)$ partitions $X^{(i)}_1,X^{(i)}_2$
  of $X$ according to Lemma~\ref{lem:nwl}. For
  every $i\in[e]$ and $j=1,2$, let $(T^{(i)}_j,\gamma^{(i)}_j)$ be the
  recursively constructed decomposition for $X^{(i)}_j$ (note that
  $\kct(X^{(i)}_j)<\kct(X)=k_1$). 

  To construct $(T_X,\gamma_X)$
  attach children $t^{(1)},\ldots,t^{(e)}$ to $r$ and let
  $\gamma_X(t^{(i)}):=X$. For every $i$, we attach the trees
  $T^{(i)}_1,T^{(i)}_2$ to
  $t^{(i)}$ such that the root of $T^{(i)}_j$ becomes a child of $t^{(i)}$ in
  $T_X$. For all $t\in V(T^{(i)}_j)$, we let
  $\gamma_X(t):=\gamma^{(i)}_j(t)$.

It is easy to see that this construction has the desired properties.
\end{proof}

\section{Constructing Canonical Treelike Decompositions}
\label{sec:ctl}

For the following lemma, we make Assumptions~\ref{ass:1}, \ref{ass:2},
and \ref{ass:4}.

\begin{lem}\label{lem:nodedec}
  There are $a_2=a_2(k)$ and $g_1=g_1(k)$
  such that there is a treelike decomposition $(T,\gamma)$
  of $\kct$ with the following properties.
  \begin{eroman}
  \item $T$ is a directed tree.
  \item $T$ has at most $n^{g_1}$ nodes, where $n:=|A|=|V(G)|$.
  \item $(T,\gamma)$ has width at most $a_2$.
  \item $\beta(t)=\{c_0\}$ for the root $r$ of $T$.
  \item $|\gamma(t)|=1$ for all leaves $t\in L(T)$.
  \end{eroman}
  Furthermore, given $\kappa$ and $C_0,\ldots,C_m$ the construction of $(T,\gamma)$ is
  canonical and can be carried out by a polynomial time algorithm (for
  fixed $k$).
\end{lem}

\begin{proof}
  We build the tree $T$ recursively in a top-down fashion. We start
  the construction of $T$ with the creation of a root $r$. We let
  $\gamma(r):=\Act$. For every triple cover $Q$ of the tangle $\CT_0$
  of size $|Q|\le \theta(3k-2)$ we create two fresh nodes $s_Q,t_Q$
  and edges $(r,s_Q)$ and $(s_Q,t_Q)$. We let
  \begin{align}
    \label{eq:27}
    \gamma(s_Q)&:=\Act\setminus\{c_0\},\\
    \label{eq:28}
    \gamma(t_Q)&:=\Act\setminus(Q^\vee\cup\{c_0\}),
  \end{align}
  where
  \[
  Q^\vee:=(B\cap Q)\cup\{c_i\mid 0\le i\le m, C_i\cap
  Q\neq\emptyset\}.
  \]
  (recall Assumptions~\ref{ass:3}).  We call the nodes $s_Q$
  \emph{cover nodes}.

  Now suppose that $t$ is some leaf of the current tree with
  $\gamma(t)=X$ for some $X$ of size at least $2$. Then some ancestor
  is the child $t_Q$ of a cover node $s_Q$. We let $Q$ be the triple
  cover associated with this node and $Q^{\vee}$ its contraction. Thus
  now we may also make Assumption~\ref{ass:3} and apply the results of
  Section~\ref{sec:1tan}. Note
  that $X\subseteq \gamma(t_Q)=\Act\setminus (Q^{\vee}\cup\{c_0\})$.

  If $\kct(X)<(3k+2)k$, we apply Lemma~\ref{lem:1tan} and obtain a
  partition $X_1,\ldots,X_p$ of $X$. We attach fresh nodes
  $t_1,\ldots,t_p$ as children to $X$ and let $\gamma(t_i):=X_i$. In
  this case, we call $t$ a \emph{small node}. Depending on whether we apply
  Lemma~\ref{lem:1tan}(i) or (ii), we call $t$ a small node of
  \emph{type (i)} or \emph{(ii)}.

  If $\kct(X)\ge (3k+2)k$, we apply
  Corollary~\ref{cor:nwl} and obtain a partial treelike decomposition
  $(T_X,\gamma_X)$. We assume $T_X$ to be disjoint from the tree we
  have constructed so far. We extend our tree by adding the tree $T_X$
  and identifying its root with $t$. For all $u\in V(T_X)$, we let
  $\gamma(u):=\gamma_X(u)$. In this case, we call $t$ a \emph{big
    node}.

  The construction stops once $|\gamma(t)|=1$ for all leaves.
  This guarantees property (v). 

  Property (i) is immediate from the construction. 
  Property (iv)
  follows from 
  \eqref{eq:27}, because the children of the root $r$ are the cover
  nodes $s_Q$. 
  The ``bounded width'' property (iii) follows from Lemma~\ref{lem:1tan} and
  Corollary~\ref{lem:nwl}. This is clear for nodes whose degree is bounded in terms
  of $k$. In particular, this is the case for all small nodes of type (i) and all big
  nodes and all nodes that are introduced as part of a tree $T_X$ at
  a big node $t$ with $\gamma(t)=X$, because by
  Corollary~\ref{lem:nwl}(v) the tree $T_X$ only has a bounded number
  of nodes. For small nodes of type (ii), it follows from
  Lemma~\ref{lem:nwl}(ii) that the width is bounded.

  It remains to prove (ii), that is, bound the size of the tree. As
  there are only polynomially many (at most $n^{\theta(3k-2)+1}$) triple
  covers $Q$ of $\CT_0$ of size $|Q|\le \theta(3k-2)$, it suffices to prove a polynomial bound on the
  size of the trees $T^{(Q)}$ rooted in the nodes $t_Q$, the grandchildren of
  the root. So let us fix $Q$ and only consider the tree $T':=T^{(Q)}$ from
  now on. 

  Next, we discard all internal nodes of the trees $T_X$
  associated with big nodes and draw direct edges from the root of
  each $T_X$ to its leaves; this reduces the size of the tree only by a constant
  factor. Let $T''$ be the resulting tree. For each node $t\in V(T'')$
  we define its \emph{weight} to be 
  \[
  w(t):=|\gamma(t)|
  \]
  Let 
  \[
  a_3:=\max\{a_1(k,(3k+2)k-1),\;2(3k+2)k-2,\;(\theta(3k-2)+1)k\}.
  \]

  \begin{claim}
    For all nodes $t\in V(T'')$ we have $\kct(\gamma(t))\le a_3.$

    \proof
    Let $t\in V(T'')$. If $t=t_Q$ is the root of $T''$, then
    \[
    |\Act\setminus \gamma(t)|=|Q\cup\{c_0\}|\le\theta(3k-2)+1
    \]
    and thus
    \[
    \kct(\gamma(t))=\kct(\Act\setminus \gamma(t))\le (\theta(3k-2)+1)k.
    \]
    Suppose next that $t$ is a child of a small node $s$. Then
    $\kct(\gamma(s))\le (3k+2)k-1$. By Lemma~\ref{lem:1tan}, if
    $s$ is a small node of type (i) then $\kct(\gamma(t))\le
    a_1(k,(3k+1)k-1)$ and if
    $s$ is a small node of type (ii) then $\kct(\gamma(t))\le
    2(3k+1)k-2$.
    \uend
  \end{claim}

  Let 
  \begin{align*}
    b&:=\max\{b_1(k,a_3),c_1(k,a_3)\},\\
    q&:=(1-1/f_1(k,a_3)).
  \end{align*}
  Then our tree $T''$ has the following properties.
  \begin{ealph}
   \item All internal nodes have at least two children.
    \item All children of a big node are small nodes or leaves.
    \item A big node has at most $b$ children.
    \item A small node of type (i) has at most $b$ children.
    \item For all children $u$ of a small node $t$ of type (i) it
      holds that $w(u)\le q\cdot w(t)$.
    \item A small node of type (ii) has at most one child $u$, called
      its \emph{heavy child}, such
      that $w(u)> q\cdot w(t)$. For all other children $u'$ it holds
      that $w(u')\le q\cdot w(t)$. 
    \item If a small node of type (ii) has a heavy child, then this
      child is a small node as well.
    \item For a small node (of either type) with children
      $u_1,\ldots,u_m$ it holds that
      \[
      w(t)=\sum_{i=1}^mw(u_i).
      \]
    \item For all leaves $t$ it holds that $w(t)=1$.
  \end{ealph}
  We re-structure the tree again. We inductively define the
  \emph{small subtree} of a small node $t$ to consist of $t$ and the
  small subtrees of all small children of $t$. We let $T'''$ be the tree obtained
  from $T''$ by contracting all small subtrees to their roots. That
  is, we delete all nodes that are small children of their parents,
  and add edges from every small node that remains to all children of
 leaves of its small subtree. Observe that $|T'''|\ge|T''|/2$. Thus it
  suffices to bound the size of $T'''$.

  The tree $T'''$ has the following properties.
   \begin{ealph}[resume]
   \item \label{li:ctl-two} All internal nodes have at least two children.
    \item \label{li:ctl-c1} All children of a big node are small nodes or leaves.
    \item \label{li:ctl-big} A big node has at most $b$ children.
    \item\label{li:ctl-dec} For all children $u$ of a small node $t$ it
      holds that $w(u)\le q\cdot w(t)$.
    \item\label{li:ctl-sum} For a small node with children
      $u_1,\ldots,u_m$ it holds that
      \[
      w(t)=\sum_{i=1}^mw(u_i).
      \]
    \item\label{li:ctl-leaf} For all leaves $t$ it holds that $w(t)=1$.
  \end{ealph}
  
  To estimate the size of $T'''$, we use the following arithmetic fact
  (more or less, Jensen's inequality).

  \begin{claim}[resume]
    Let $m\in\NN$ and $c,q,x,x_1,\ldots,x_m\in\RR$ such that $c\ge 1$,
    $0\le q\le 1$, $x>0$, $0\le x_1,\ldots,x_m\le qx$, and
    $\sum_{i=1}^mx_i=x$. Then 
    \[
    \sum_{i=1}^mx_i^c\le (q^c+(1-q)^c)x^c.
    \]
    \proof
    Without loss of generality we may assume that $q\ge 1/2$, because
    otherwise the claim for $q$ follows from the claim for $q'=(1-q)$.
    
    For all $x>0$, $1/2\le q\le 1$, $m\ge 2$, we let $\mathbb
    D(x,q,m)$ be the set of all $(x_1,\ldots,x_m)\in\RR^m$ such that
    $\sum_{i=1}^mx_i=x$ and $0\le x_i\le qx$ for all $i\in[m]$. We
    shall prove that for all
    $(x_1,\ldots,x_m)\in \mathbb D(x,q,m)$ we have
    \begin{equation}
      \label{eq:29}
      \sum_{i=1}^mx_i^c\le (q^c+(1-q)^c)x^c.
    \end{equation}
    The proof is by induction on $m$.
    For the base step $m=2$, we just observe that the function $y\mapsto
    y^c+(x-y)^c$ is convex on the interval $[x-qx,qx]$ and thus
    assumes its maximum at the boundary.

    For the inductive step $m\to m+1$, where $m\ge 2$, let $x>0$ and 
    $x_1,\ldots,x_{m+1}\in\mathbb D(x,q,m+1)$. Then $(x_1,\ldots,
    x_m+x_{m+1})\in\mathbb D(x,q,m)$, and by the inductive hypothesis we
    have
    \[
    \sum_{i=1}^{m-1}x_i^c+(x_m+x_{m+1})^c\le (q^c+(1-q)^c)x^c.
    \]
    By applying the inductive hypothesis $m=2$ with $q=1$,
    $x=x_m+x_{m+1}$, we see that
    $x_m^c+x_{m+1}^c\le (x_m+x_{m+1})^c$. Thus
     \[
    \sum_{i=1}^{m+1}x_i^c\le \sum_{i=1}^{m-1}x_i^c+(x_m+x_{m+1})^c\le
    (q^c+(1-q)^c)x^c.
    \uende
    \]
  \end{claim}

  We choose $c\ge 1$ such that 
  \begin{equation}
    \label{eq:30}
    (q^c+(1-q)^c)\le\frac{1}{2(b+1)},
  \end{equation}
  For every node
  $t\in V(T''')$, let $s(t)$ be the number of nodes in the subtree of
  $T'''$ rooted in $t$. 

  \begin{claim}[resume]
    \begin{enumerate}
    \item For all big nodes $t\in V(T''')$ we have $s(t)\le 1+bw(t)^c$.
    \item For all small nodes and leaves $t\in V(T''')$ we have $s(t)\le w(t)^c$.
    \end{enumerate}

    \proof
    The proof is by induction on $T'''$. For the leaves $t$ we have
    $s(t)=w(t)=1$ and thus the claim is trivial.

    For the inductive step, let $t$ be a node with children
    $u_1,\ldots,u_m$. 
    Assume first that $t$ is a small node. Let $x:=w(t)$ and $x_i:=w(u_i)$ for
    all $i\in[m]$. Then $1\le x_i\le qx$ by \ref{li:ctl-dec} and
    $\sum_{i=1}^mx_i=x$ by \ref{li:ctl-sum}. By the induction
    hypothesis, we have $s(u_i)\le 1+bx_i^c\le(b+1)x_i^c$. Hence by Claim~1,
    \[
    s(t)=1+\sum_{i=1}^ms(u_i)\le 1+(b+1)\sum_{i=1}^mx_i^c\le
    1+(b+1)(q^c+(1-q)^c)x^c\le 1+\frac{w(t)^c}{2},
    \]
    where the last inequality holds by \eqref{eq:30}. As $w(t)\ge 2$
    by \ref{li:ctl-two} and \ref{li:ctl-leaf}, we have $1\le w(t)^c/2$
    and thus $s(t)\le w(t)^c$.

    Assume next that $t$ is a big node. By \ref{li:ctl-c1}, all the
    $u_i$ are small nodes or leaves, and thus by the induction
    hypothesis we have $s(u_i)\le w(u_i)^c\le w(t)^c$. By
    \ref{li:ctl-big} we have $m\le b$. Thus
    \[
    s(t)=1+\sum_{i=1}^ms(u_i)\le 1+bw(t)^c.\uende
    \]
  \end{claim}

  Claim~2 implies that $|T'''|\le 1+bn^c$, which gives the desired
  polynomial bound. 

  It is clear from the construction that the decomposition
  $(T,\gamma)$ can be computed in polynomial time in its input and the
  number of nodes of the tree $T$.
\end{proof}

\begin{theo}[Canonical Decomposition Theorem]\label{theo:candec}
  Let $k\in\NN$. Then there is an $a=a(k)\in\mathbb{N}$ and a polynomial time algorithm that, given
  a graph $G$ of rank width at most $k$, computes a canonical treelike
  decomposition of $\rho_G$ of width at most $a$.
\end{theo}

\begin{proof}
  Let $G$ be a graph of rank width $k$. We let $A:=V(G)$ and
  $\kappa:=\rho_G$. Let $\KT_{\max}$ be the set of all maximal
  $\kappa$-tangles. By the Duality Theorem~\ref{theo:duality}, we have
  $\KT_{\max}=\KT_{\max}^{\le k}$, that is, all tangles in
  $\KT_{\max}$ have order at most $k$. 

  Let $\CT_{\text{root}}\in\KT_{\max}$ be arbitrary. In a first step, we show how to construct
  a decomposition that is canonical given $\CT_{\text{root}}$. In a second step,
  we will show how to combine the decompositions for all $\CT_{\text{root}}\in\KT_{\max}$.

  Let $(T_1,\gamma_1,\tau_1)$ be the directed tree decomposition for
  $\KT_{\max}$ computed by the algorithm of
  Theorem~\ref{theo:dcandec}. Let $r_1$ be the root of $T_1$.

  Let $t_1\in V(T_1)$, and let $u_1,\ldots,u_m$ be the children of
  $t_1$ in $T_1$ (possibly
  $m=0$). Let $C_0^{(t_1)}:=\bar{\gamma_1(t_1)}$ and $C_i^{(t_1)}:=\gamma_1(u_i)$. We
  define $c_0^{(t_1)},\ldots,c_m^{(t_1)}$, $B^{(t_1)}$, $\Act^{(t_1)}$
  and $\kct^{(t_1)}$ as in Assumptions~\ref{ass:1}(5)--(7). Let 
  $\CT_0^{(t_1)}:=\tau_1^{-1}(t_1)$ be the tangle associated with
  $t_1$ and $k_0^{(t_1)}:=\ord(\CT_0^{(t_1)})$.

  Observe that Assumptions~\ref{ass:1}, \ref{ass:2}, and \ref{ass:4} are
  satisfied. This is trivial for Assumptions~\ref{ass:1}(1), (2), (5)--(7),
  \ref{ass:2}(3) and \ref{ass:4}. We have $\kappa(C_i)<k$, because
  $\bar C_0=\gamma_1(t_1)\in\CT_0^{(t_1)}$ and
    $C_i=\gamma_1(u_i)\in\tau_1^{-1}(u_i)$ by \ref{li:dtd2}, and the order of all these
    tangles is at most $k=\bw(\kappa)$. This implies 
    Assumption~\ref{ass:1}(3). Assumption~\ref{ass:1}(4)
  follows from \ref{li:dtd2}.
  Assumptions~\ref{ass:2}(1) and (2) follow from \ref{li:dtd1} and
  \ref{li:dtd2}. We let
  $(T_2^{(t_1)},\gamma_2^{(t_1)})$ be the decomposition computed by the
  algorithm of Lemma~\ref{lem:nodedec}. Let $r_2^{(t_1)}$ be the root of
  $T^{(2)}$. By
  Lemma~\ref{lem:nodedec}(iv) we have
  \begin{ealph}
    \item\label{li:dec1}
      $\beta_2^{(t_1)}(r_2^{(t_1)})=\{c_0^{(t_1)}\}$
  \end{ealph}
  and thus
  \[
  c_0^{(t_1)}\not\in\gamma_2^{(t_1)}(t_2)
  \]
  for all $t_2\in V(T_2^{(t_1)})\setminus\{r_2^{(t_1)}\}$.
  
  We may assume without loss of generality that the decomposition
  $(T_2^{(t_1)},\gamma_2^{(t_1)})$ has the
  following two properties:
  \begin{ealph}[resume]
  \item\label{li:dec2} $\beta_2^{(t_1)}(t_2)=\emptyset$ for all $t_2\in
    V(T_2^{(t_1)})\setminus \big(L(T_2^{(t_1)})\cup\{r_2^{(t_1)}\big)$;
  \item\label{li:dec3} $|\gamma_2^{(t_1)}(t_2)|=1$ for all $t_2\in
    L(T_2^{(t_1)})$.
  \end{ealph}
  If $(T_2^{(t_1)},\gamma_2^{(t_1)})$ does not have these properties,
  we normalise it (see Lemma~\ref{lem:normal}), but leave the bag of
  the root untouched.

  Now we join the decompositions $(T_1,\gamma_1)$ and
  $(T_2^{(t_1)},\gamma_2^{(t_1)})$ for $t_1\in V(T_1)$ into a new decomposition
  $(D,\gamma)$.
The directed graph $D$ is defined as follows.
\begin{itemize}
\item The node set $V(D)$ consist of all pairs $(t_1,t_2)$, where
  $t_1\in V(T_1)$ and $t_2\in V(T_2^{(t_1)})$.
\item For $(t_1,t_2),(u_1,u_2)\in V(D)$, there is an edge from
  $(t_1,t_2)$ to $(u_1,u_2)$ in $E(D)$ if 
  \begin{itemize}
  \item either $t_1=u_1$ and $(t_2,u_2)\in E(T_2^{(t_1)})$%
  \item or $(t_1,u_1)\in E(T_1)$ and $t_2$ is a leaf of $T_2^{(t_1)}$ with
    $\gamma_2^{(t_1)}(t_2)\expand=\gamma_1(u_1)$ and  $u_2=r_2^{(u_1)}$ is the root of $T_2^{(u_1)}$.
  \end{itemize}
\end{itemize}
Thus we take disjoint copies  of all trees
  $T_2^{(t_1)}$ for $t_1\in V(T_1)$ and connect them along the
  edges of $T_1$, connecting leaves of a tree $T_2^{(t_1)}$ to
  matching roots of trees $T_2^{(u_1)}$ for children $u_1$ of $t_1$ in $T_1$.

Observe that for every $t_1\in V(T_1)$ and for every leaf
$t_2\in V(T_2^{(t_1)})$, either $(t_1,t_2)$ is a leaf of $D$ (that is,
a node of has out-degree $0$)
or it has precisely one child (that is, out-neighbour), which is the root of $T_2^{(u_1)}$ for
some child $u_1$ of $t_1$ in $T_1$. This follows from \ref{li:dec3} and the fact
that $(T_1,\gamma_1)$ is a tree decomposition, and not just a tree
like decomposition, so $t_1$ cannot have two distinct children $u_1,u_1'$ with the same
cone $\gamma_1(u_1)=\gamma_1(u_1')$.

We define $\gamma:V(D)\to A$ by
\[
\gamma(t_1,t_2):=
\begin{cases}
  \gamma_1(t_1)&\text{if }t_2=r_2^{(t_1)},\\
  \gamma_2^{(t_1)}(t_2)\expand&\text{otherwise}.
\end{cases}
\]

\begin{claim}
  $(D,\gamma)$ is a treelike decomposition of $\kappa$.

  \proof
  As $T_1$ and the $T_2^{(t_1)}$ are all directed trees, $D$ is a
  directed acyclic graph (not necessarily a tree, though), because a
  cycle in $D$ would project to a cycle in $T_1$ or some
  $T_2^{(t_1)}$. This proves \ref{li:tl1}.

  To verify \ref{li:tl2}, let $\big((t_1,t_2),(u_1,u_2)\big)\in
  E(D)$. If $t_1=u_1$, then $(t_2,u_2)\in E(T_2^{(t_1)})$ and thus
  $\gamma_2^{(t_1)}(u_2)\subseteq \gamma_2^{(t_1)}(t_2)$ by
  \ref{li:tl2} for $(T_2^{(t_1)},\gamma_2^{(t_1)})$. If
  $t_2\neq r_2^{(t_1)}$ then 
  \[
  \gamma(u_1,u_2)=\gamma_2^{(t_1)}(u_2)\expand\subseteq
  \gamma_2^{(t_1)}(t_2)\expand=\gamma(t_1,t_2).
  \]
  If $t_2= r_2^{(t_1)}$ then
  \[
  \gamma(u_1,u_2)=\gamma_2^{(t_1)}(u_2)\expand\subseteq (\Act^{(t_1)}\setminus\{c_0^{(t_1)}\})\expand=\gamma_1(t_1)=\gamma_1(t_1,t_2),
  \]
  where the inclusion $\gamma_2^{(t_1)}(u_2)\expand\subseteq
  (\Act^{(t_1)}\setminus\{c_0^{(t_1)}\})\expand$ follows from \ref{li:dec1}.

  If $(t_1,u_1)\in E(T_1)$, then 
  \[
  \gamma(t_1,t_2)=\gamma_2^{(t_1)}(t_2)\expand
  =\gamma_1(u_1)=\gamma(u_1,u_2).
  \]
  
  To verify \ref{li:tl3}, let $(t_1,t_2)\in V(D)$ and $(u_1,u_2),(v_1,v_2)\in
  N_+^D(t_1,t_2)$; without loss of generality distinct. Then 
  $t_1=u_1=v_1$, because otherwise  $(t_1,t_2)$ would have out-degree
  one. Then $u_2,v_2$ are children of
  $t_2$ in $T_2^{(t_1)}$, and it follows from \ref{li:tl3}
  for $(T_2^{(t_1)},\gamma_2^{(t_1)})$ that either
  $\gamma_2^{(t_1)}(u_2)=\gamma_2^{(t_1)}(u_2)$, which implies
  $\gamma(u_1,u_2)=\gamma(v_1,v_2)$, or $\gamma_2^{(t_1)}(u_2)\cap\gamma_2^{(t_1)}(u_2)=\emptyset$, which implies
  $\gamma(u_1,u_2)\cap\gamma(v_1,v_2)=\emptyset$. 

Finally, \ref{li:tl4} holds because
  \[
  \gamma(r_1,r_2^{(r_1)})=\gamma_1(r_1)=A
  \]
  by \ref{li:tl4} for $(T_1,\gamma_1)$.
  \uend
\end{claim}

\begin{claim}[resume]
  The width of $(D,\gamma)$ is at most 
  the maximum width of the $(T_2^{(t_1)},\gamma_2^{(t_1)})$ for
  $t_1\in V(T_1)$.

  \proof
  Let $(t_1,t_2)\in V(D)$.

  If $t_2$ is a leaf of $T_2^{(t_1)}$ such
  that $\gamma_2^{(t_1)}(t_2)\expand=\gamma_1(u_1)$ for some child $u_1$ of $t_1$
  in $T_1$, then $(u_1,r_2^{(u_1)})$ is a child of $(t_1,t_2)$ in
  $D$, and we have
  $\gamma(u_1,r_2^{(u_1)})=\gamma_1(u_1)=\gamma(t_1,t_2)$. This
  implies $\beta(t_1,t_2)=\emptyset$. The
  width of $(D,\gamma)$ at $t$ is 
  \[
  \max\{\kappa(\emptyset),\kappa(\gamma_1(u_1)\},
  \]
  which is also the width of $(T_2^{(t_1)},\gamma_2^{(t_1)})$ at $t_2$.

  Otherwise, the children of $(t_1,t_2)$ in $D$ are $(t_1,u_2)$ for
  all children $u_2$ of $t_2$ in $T_2^{(t_1)}$. Assume first that
  $t_2$ is not the root of $T_2^{(t_1)}$. Then 
  $\beta(t_1,t_2)=\beta_2^{(t_1)}(t_2)=\emptyset$ by \ref{li:dec2}. For every
  subset $U$ of the children of $t_2$ in $T_2^{(t_1)}$ we have 
  \[
  \kappa\left(\bigcup_{u_2\in U}\gamma(t_1,u_2)\right)
  =\kappa\left(\bigcup_{u_2\in U}\gamma_2^{(t_1)}\expand\right)
  =\kct^{(t_1)}\left(\bigcup_{u_2\in U}\gamma_2^{(t_1)}\right).
  \]
  It follows that the width of $(D,\gamma)$ at $t$ is equal to the
  width of $(T_2^{(t_1)},\gamma_2^{(t_1)})$ at $t_2$.

  If $t_2=r_2^{(t_1)}$, we can argue similarly.
  \uend
\end{claim}

This completes the first (and main) step of our proof; the
construction of a
bounded-width treelike decomposition that is canonical given the
  ``root tangle'' $\CT_{\textup{root}}$. 

To obtain a fully canonical decomposition, we carry out the
construction above for every $\CT_{\text{root}}\in\KT_{\max}$. Say, we obtain
decompositions $(D_1,\gamma_1),\ldots,(D_m,\gamma_m)$. As there is
only a linear number of maximal tangles, we can compute all these
decompositions in polynomial time. Our final decomposition is simply
the disjoint union of these decompositions $(D_1,\gamma_1),\ldots,(D_m,\gamma_m)$.
\end{proof}

\section{Matrices of Bounded Partition Rank}
\label{sec:part:rank}

In this section we consider symmetric matrices $P\in\{0,1,?\}^{V\times
  V}$ with entries $0,1,?$ and row and column indices from a set
$V$. We usually denote the entries of such a matrix $P$ by $p_{vw}$,
  for $v,w\in V$, and we denote the row $(p_{vw}\mid w\in V)$ with index $v$ by $\mathbf
  p_v$. We need no special notation for the columns
  and just refer to them via their indices $w\in V$.

  If the $?$-entries of such a matrix $P$ form a block diagonal matrix,
  we call $P$ a \emph{$?$-block matrix}. That is,
  $P\in\{0,1,?\}^{V\times V}$ is a $?$-block matrix if it is symmetric
  and there are mutually disjoint subsets $I_1,\ldots,I_m\subseteq V$
  such that $p_{vw}=?$ if and only if there is a $j\in[m]$ such that
  $v,w\in I_j$. We call the sets $I_1,\ldots,I_m$ the
  \emph{$?$-indices} of $P$, and we say that row $\mathbf p_v$ \emph{has
    $?$-index $I_j$} if $v\in I_j$ (similarly for columns).  For
  disjoint subsets~$B,C \subseteq \{I_1,\ldots,I_m\}$, we
  let~$\Mat_{B,C}$ be the submatrix of $P$ obtained by deleting all rows
  corresponding to indices that are not in~$B$ and deleting all
  columns corresponding to indices that are not in~$C$. Note
  that~$\Mat_{B,C}$ is a~$\{0,1\}$-matrix. We denote
  by~$\Mat_{B,\overline{B}}$ the
  matrix~$\Mat_{B,\{I_1,\ldots,I_m\}\setminus B}$.

We say that the matrix~$\Mat$ has \emph{partition rank at most~$k$} if
for each partition of the family of $?$-indices into two parts~$B$
and~$\bar B$, the submatrix~$\Mat_{B,\bar B}$ has rank at most~$k$ over~$\mathbb{F}_2$.

We are interested in $?$-block matrices and their partition rank because
we can use them to describe the width of treelike decompositions of
cut-rank functions.  Let $(D,\gamma)$ be a normal treelike decomposition of the cut rank
  function $\rho_G$ of a graph $G$, and let $t\in V(D)$ be a node with
children~$u_1,\ldots,u_{\ell}$ such that the children have pairwise
disjoint cones. We define an
\emph{associated $?$-block matrix~$\Mat\in\{0,1,?\}^{V(G)\times V(G)}$} with
entries $p_{vw}$ defined as follows:
\begin{itemize}
\item if there is an $i\in[\ell]$ such that $v,w\in\gamma(u_i)$ then
  $p_{uw}=?$;
\item if $v,w\in\bar{\gamma(t)}$ then
  $p_{uw}=?$;
\item otherwise, if $vw\in E(G)$ then $p_{vw}=1$ and if $vw\not\in E(G)$ then $p_{vw}=0$.
\end{itemize}
Note that the $?$-indices of $P$ are the sets $\gamma(u_1),\ldots,\gamma(u_\ell),\bar{\gamma(t)}$.

\begin{lem}\label{lem:pr_vs_width}
  Let $(D,\gamma)$ be a normal treelike decomposition of the cut rank
  function $\rho_G$ of a graph $G$. Let $t\in V(D)$ be a node whose
  children have mutually disjoint cones, and let $P$ be the $?$-block
  matrix associated with $t$. Then the partition rank of $P$ is equal
  to the width of $(D,\gamma)$ at $t$.
\end{lem}

\begin{proof}
  Let $u_1,\ldots,u_\ell$ be the children of $t$. Then
  \[
  \CI:=\{\gamma(u_1),\ldots,\gamma(u_{\ell}), \bar{\gamma(t)}\}.
  \]
  is the set of $?$-indices of $P$. The width of the decomposition
  at node $t$ is
  \[
  \width(D,\gamma,t)=\max_{U\subseteq\{u_1,\ldots,u_\ell\}}\rho_G\Big(\bigcup_{u\in
    U}\gamma(u)\Big)=\max_{B\subseteq\{\gamma(u_1),\ldots,\gamma(u_\ell)\}}\rk(\Mat_{B,\CI\setminus
    B})=\max_{B\subseteq\CI}\rk(\Mat_{B,\CI\setminus B}),
  \]
  where the first equality holds because $(D,\gamma)$ is normal and
  thus $\beta(t)=\emptyset$ and the last equality holds because $\rk(\Mat_{B,\CI\setminus
    B})=\rk(\Mat_{\CI\setminus B,B})$ for all $B\subseteq I$.
\end{proof}

An \emph{extension} of a
$\{0,1,?\}$-vector is a $\{0,1\}$-vector obtained by
replacing each~`$?$'-entry by a~$0$ or a~$1$. That is,
$\mathbf x=(x_v\mid v\in V)\in\{0,1\}^V$ is an extension of
$\mathbf p=(p_v\mid v\in V)\in\{0,1,?\}^V$ if $p_v\in\{0,1\}$ implies
$x_v=p_v$, for all $v\in V$.  We say that two~$\{0,1,?\}$-vectors are
\emph{compatible} if they have a common extension.  An
\emph{isomorphism} from a matrix $\Mat\in\{0,1,?\}^{V\times V}$ to a
matrix $\Mat'\in\{0,1,?\}^{V'\times V'}$ is a bijective mapping
$\phi:V\to V'$ such that $p_{vw}=p'_{\phi(v)\phi(v)}$ for all
$v,w\in V$, where as usual we denote the entries of $P$ by $p_{vw}$
and the entries of $P'$ by $p'_{v'w'}$.

Let~$\Mat\in\{0,1,?\}^{V\times V}$ be a $?$-block matrix. An
\emph{extension set} for $\Mat$ is a set of vectors~$\Ext\subseteq
\{0,1\}^{V}$ such that every row in~$\mathbf p_v$ of $\Mat$ has an
extension in~$\Ext$. If $\Ext$ is an extension set for~$\Mat$, then
for every $v\in V$ we denote the set of all extensions of $\mathbf
p_v$ in $\Ext$ by $\Ext(v)$.
We call a construction that assigns an extension set to every $?$-block
matrix \emph{canonical} if for every two isomorphic $?$-block matrices
$\Mat\in\{0,1,?\}^{V\times V}$ and $\Mat'\in\{0,1,?\}^{V'\times V'}$
and every isomorphism $\psi$ from $\Mat$ to $\Mat'$ the following
two conditions are satisfied.
\begin{eroman}\label{page:canext}
  \item There is a bijection~$\chi$ from~$\Ext$ to~$\Ext'$ such that
    $\chi(\Ext(v)) = \Ext'(\psi(v))$ for all $v\in V$.
  \item For every~$\mathbf x=(x_v\mid v\in V)\in \Ext$ with $\chi(\mathbf x)=:\mathbf x'=(x'_{v'}\mid
v'\in V')\in\Ext'$ and every $v\in V$ we have $x_v=x'_{\psi(v)}$.
\end{eroman}

\begin{theo}\label{thm:compute:canonical:row:extensions}
  Let $k\in\NN$. Then there is an $e=e(k)\in\mathbb{N}$ and a polynomial time
  algorithm that, given a $?$-block matrix
  $\Mat\in\{0,1,?\}^{V\times V}$ of partition rank at most $k$,
  computes a canonical extension set~$\Ext\subseteq \{0,1\}^V$ for $P$
  of size $\lvert \Ext\rvert\leq e$.
\end{theo}

\begin{proof}
  Let $\Mat\in\{0,1,?\}^{V\times V}$ be a $?$-block matrix, and let $\CI=\{I_1,\ldots,I_m\}$ be the set of $?$-indices of $P$. 
  Without loss of generality we assume that~$\Mat$ neither has
  repeated rows nor repeated columns. Note that this implies that
  each $?$-index $I_j$ has size at most~$|I_j|\le 2^k$. Also note
  that the assumption implies that a row cannot be compatible with
  another row of the same index.

  Let~$h$ be the function defined recursively by $h(0):=1$ and $h(i+1):=2^k\cdot h(i)+4$ for all $i\ge 0$.

  \begin{claim}
    Let $R$ be a set of rows of $P$ that 
    have mutually distinct
    $?$-indices and are mutually incompatible. Then $|R|\le h(k)$.

    \proof We say that two rows $\mathbf p_v,\mathbf p_{v'}$ are
    \emph{compatible in column $w$} if either $p_{vw}=p_{v'w}$ or
$p_{vw}=?$ or $p_{v'w}=?$. For a set $W\subseteq V$, rows $\mathbf p_v,\mathbf p_{v'}$ are
    \emph{compatible in $W$} if they are compatible in all columns
    $w\in W$, and for a set $\CJ\subseteq \CI$,  rows $\mathbf p_v,\mathbf p_{v'}$ are \emph{compatible in $\CJ$}
    if they are compatible in $\bigcup_{j\in B}I_j$. Rows $\mathbf
    p_v,\mathbf p_{v'}$ are \emph{incompatible in $w$, $W$, or $\CJ$} if they are
    not compatible in $w$, $W$, $\CJ$, respectively.

    For a set $R$ of rows of $P$, by $\CI(R)$ we denote the set of
    $?$-indices of the rows in $R$. 

    By induction on $i\ge 0$, we shall prove the following.
    \begin{ealph}
    \item If $|R|>h(i)$ and $\CJ\subseteq\CI$ such that the rows in
      $R$ are mutually incompatible in $\CJ$, then there are 
      sets $B\subseteq\CI(R),C\subseteq \CJ$ such that $B\cap C=\emptyset$ and
      $\rk(P_{B,C})>i$.
    \end{ealph}
    As the partition rank of $P$ is at most $k$, assertion (A) implies
    that $|R|\le h(k)$, that is, the assertion of the claim.

    For the base step $i=0$, we observe that if $|R|>1$ then there are
    at least two rows $\mathbf p,\mathbf p'\in R$ that are
    incompatible in some $J\in\CJ$. If we let $B$ be the set
    consisting of the ?-indices of  $\mathbf p,\mathbf p'$ and $C:=\{J\}$, we obtain a submatrix
    $P_{B,C}$ of rank at least $1$.

    For the inductive step $i\to i+1$, we assume that $|R|>h(i+1)$. We
    first prove that there are four $?$-indices
    $J,J'\in\CI(R)$, $K,K'\in\CJ$ such that $\{J,J'\}\cap\{K,K'\}=\emptyset$ and
    $P_{\{J,J'\},\{K,K'\}}$ has rank at least $2$.
    Since $|R|>h(i+1)\ge 4$,
    we find distinct rows $\mathbf p_1,\ldots,\mathbf p_4\in R$ that
    all have at least one $1$-entry in a column $w\in K$ for some $K\in\CJ$. For
    $i=1,\ldots,4$, let $J_i\in\CI(R)$ be the~$?$-index of row
    $\mathbf p_i$. 
    \begin{cs}
      \case 1 There are distinct $i,j\in[4]$ and a~$?$-index $K\in\CJ$
      such that rows $\mathbf p_i$ and $\mathbf p_j$ both have a
      $1$-entry in a column $w\in K$. 

      Let $K'\in\CJ$ such that $\mathbf p_i$ and $\mathbf p_j$ are
      incompatible in $K'$. Then the matrix $P_{\{J_i,J_j\},\{K,K'\}}$
      has rank at least $2$.  

      \case 2 For all distinct $i,j\in[4]$
      there is no ?-index $K\in\CJ$ such that rows $\mathbf p_i$ and
      $\mathbf p_j$ both have a $1$-entry in a column $w\in K$.

      Let $K_1\in\CJ$ such that $\mathbf p_1$ has a $1$-entry in
      $K_1$. There is at most one $i\in\{2,3,4\}$ such that
      $K_1=J_i$. Without loss of generality, $J_2,J_3\neq K_1$. Then
      $\mathbf p_2,\mathbf p_3$ only have $0$-entries in $K_1$. Let
      $K_2\in\CJ$ such that $\mathbf p_2,\mathbf p_3$ are
      incompatible in $K_2$. Without loss of generality we may assume
      that $\mathbf p_2$ has a $1$-entry in $K_2$. If $K_2\neq J_1$,
      then the
      matrix $P_{\{J_1,J_2\},\{K_1,K_2\}}$ has rank at least $2$. 

      Suppose that $K_2=J_1$. Let $K_3\in\CJ$ such that $\mathbf p_3$
      has a $1$-entry in $K_3$. By the assumption of Case~2, we have
      $K_3\neq K_2=J_1$. Then the
      matrix $P_{\{J_1,J_3\},\{K_1,K_3\}}$ has rank at least $2$. 
    \end{cs}
    In the following, we fix $J,J'\in\CI(R)$, $K,K'\in\CI$ such that
    $\{J,J'\}\cap\{K,K'\}=\emptyset$ and $P_{\{J,J'\},\{K,K'\}}$ has
    rank at least $2$.

    At most four rows in $R$ have their $?$-index in $\{J,J',K,K'\}$,
    because the rows in $R$ have mutually distinct $?$-indices. We define
    an equivalence relation $\sim$ on the remaining rows in $R$ by letting
    $\mathbf p\sim\mathbf p'$ if and only if the rows $\mathbf p$ and
    $\mathbf p'$ coincide in all columns in $J\cup J'\cup K\cup
    K'$. As the rank of the matrix $P_{\CI\setminus\{J,J',K,K'\},
      \{J,J',K,K'\}}$ is at most $k$, there are at most $2^k$
    equivalence classes. Hence there is an equvalence class $R'$ of
    size at least 
    \[
    \frac{|R|-4}{2^k}>\frac{h(i+1)-4}{2^k}= h(i).
    \]
    Note that $\{J,J',K,K'\}\cap\CI(R')=\emptyset$ by the definition of
    $R'$. Furthermore, the rows in
    $R'$ are mutually incompatible in $\CJ':=\CJ\setminus\{J,J'K,K'\}$,
    because they are mutually incomptabible, but coincide in all
    columns in $J\cup J'\cup K\cup K'$.
    By the induction hypothesis, there are disjoint sets
    $B'\subseteq\CI(R')$ and $C'\subseteq \CJ'$
    such that $B'\cap C'=\emptyset$ and
    $\rk(P_{B',C'})>i$. Let $B:=B'\cup\{J,J'\}$ and
    $C:=C'\cup\{K,K'\}$. Then $B\cap C=\emptyset$. We claim that $\rk(P_{B,C})>
    i+1$. To see this, let $\mathbf p,\mathbf p'$ be two rows of the matrix
    $P_{B,C}$ with $?$-indices in $J\cup J'$ such that the projections of
    $\mathbf p,\mathbf p'$ to the
    columns in $K\cup K'$ are linearly
    independent. 
As all rows in $R'$ coincide on all
    columns in $K\cup K'$, either $\mathbf p$ or $\mathbf p'$ is not a
    linear combination of rows of the matrix $P_{B',C}$. This shows
    that $\rk(P_{B,C})>\rk(P_{B',C})\ge \rk(P_{B',C'})\ge
    i+1$.
This completes the proof of (A) and hence
    of Claim~1.
    \uend
  \end{claim}

We say a row~$\mathbf p_v$ of $P$ is \emph{lonely} if there are less
than $g(k):= (k+1)\cdot2^k$ rows that are compatible with~$\mathbf p_v$.

  \begin{claim}[resume]
    The number of lonely rows of $P$ is at
    most~$2^k g(k) h(k)$.

    \proof 
    Let $R$ be the set of all lonely rows of $P$, and let $r:=|R|$. We
    shall prove that $r\le  2^k g(k) h(k)$. 

    As $|I_j|\le 2^k$ for
    all $j$, there is a subset $R_0\subseteq R$ of size 
    \[
    r_0:=|R_0|\ge
    \frac{r}{2^k}.
    \]
    such that the rows $\mathbf p\in R_0$ have
    mutually distinct $?$-indices. As the rows $\mathbf p\in
    R_0\subseteq R$ are lonely, each such row is
    compatible with less than $g(k)$ rows. Thus there
    is a subset $R_1\subseteq R_0$ of size 
    \[
    r_1:=|R_1|\ge\frac{r_0}{g(k)}\ge\frac{r}{2^kg(k)}
    \]
    such that the rows in $R_1$ are mutually incompatible. 

    It follows from Claim~1 that $r_1\le h(k)$. Thus $r\le
    2^kg(k)h(k)$.
    \uend
\end{claim}

We call a~$\{0,1\}$-vector a \emph{supported extension} if it is an extension
of at least $k+2$ rows of $P$.

\begin{claim}[resume]
The number of supported extensions is at most~$2^k$. 

\proof
Let $\ell:=2^{k}+1$.
Suppose for contradiction that $\mathbf x_1,\ldots, \mathbf x_\ell$ is a sequence of
pairwise distinct supported extensions. Let $X$ be the
$[2^{k}+1]\times V$-matrix whose rows are the vectors $\mathbf
x_i$. The rank of this matrix is at least $k+1$, because a matrix of
rank at most $k$ over $\mathbb F_2$ has at most $2^k$ distinct
rows. Let $w_1,\ldots,w_{k+1}\in V$ be column indices such that
the columns of $X$ indexed by the $w_j$ are linearly independent, and
let $X'$ be submatrix of $X$ obtained by deleting all columns except
those with indices $w_1,\ldots,w_{k+1}$. Then the rank of $X'$ is $k+1$.

For
all $j\in[k+1]$, let $K_j\in\CI$ be the $?$-index of column $w_j$ of
$P$. Let 
\[
C:=\{K_1,\ldots,K_{k+1}\}.
\]
Let $i\in[\ell]$. The vector $\mathbf x_i$ is an extension of distinct
rows $\mathbf p^1_i,\ldots,\mathbf p_i^{k+2}$ of $P$. These rows must
have mutually distinct $?$-indices, because distinct rows with the same
$?$-index cannot have a common extension. As $|C|=k+1$, there is a $j\in[k+2]$
such that the $?$-index of $\mathbf p_i^j$ is not in $C$. We let
$\mathbf p_i:=\mathbf p_i^j$. Then $\mathbf p_i$ and its extension
$\mathbf x_i$ coincide on all columns in $C$. Let  $J_i$ the $?$-index of $\mathbf
p_i$. Then $J_i\not\in C$.

Now let $B:=\{J_1,\ldots,J_\ell\}$. Then $B\cap
C=\emptyset$. Then matrix $P_{B,C}$ contains the rank-$(k+1)$ matrix $X'$ as a
submatrix, obtained by deleting all rows except the rows $\mathbf p_i$
and all columns except those with indices $w_1,\ldots,w_{k+1}$. This
contradicts the partition rank of $P$ being at most $k$.
\uend
\end{claim}

Let~$\Ext\subseteq\{0,1\}$ be the union of the set of all vectors
$\mathbf x$ that are
supported extensions and the set of all vectors $\mathbf y$ that are
extensions of lonely rows of $P$. Every row $\mathbf p$ of $P$ has at
most $2^{2^k}$ extensions, because its $?$-index has size at most
$2^k$. Thus, by Claims~2 and 3,
\[
\lvert \Ext\rvert\le 2^k+2^{2^k+k}\cdot g(k)\cdot h(k).
\]
We claim that every row $\mathbf p$ of~$\Mat$ has an extension
in~$\Ext$. This is obvious if~$\mathbf p$ is lonely. Suppose~$\mathbf
p$ is not lonely. Let $R$ be the set of all rows of $P$ that are
compatible with $\mathbf p$. Then $|R|\ge g(k)=(k+1)\cdot2^k$. There is a
subset $R'\subseteq R$ of size 
\[
|R'|=\frac{|R|}{2^k}\ge k+1
\]
such that all rows in $R'$ coincide on the columns in the $?$-index $I$ of
$\mathbf p$. We extend $\mathbf p$ to a vector $\mathbf x\in\{0,1\}^V$
that coincides with the vectors in $|R'|$. Then $\mathbf x$ is a
common extension of $\mathbf p$ and all vectors in $R'$. Thus
$\mathbf x$ is a supported extension and hence in $\Ext$.

To finish the proof, it suffices now to note that the construction of~$\Ext$ is canonical.
\end{proof}

\section{Computing the Automorphism Groups}
\label{sec:iso}

We use various standard algorithms for permutation groups. Recall that
a permutation group~$\Gamma$ that permutes elements in some set~$V$
can be succinctly represented by a generating set. For a
set~$\{g_1,\ldots,g_t\}$ of permutations on~$V$ the group generated by
the set is denoted by~$\langle g_1,\ldots,g_t \rangle$ or
by~$\langle\{g_1,\ldots,g_t\} \rangle$. For sets~$V'$ and~$V$,
slightly abusing terminology, we call a set of bijections~$\Lambda$
from~$V'$ to~$V$ a \emph{$(V,V')$-coset}, or just a \emph{coset} if $V$ and
$V'$ are clear from the context, if there
is a bijection~$\sigma$ in~$\Lambda$ and a permutation group~$\Gamma$
on~$V'$ such that~$\sigma \Gamma = \Lambda$.\footnote{To relate this
  to the standard group theoretic notion of coset, note that if $V=V'$, then a
$(V,V')$-coset is a left coset of a subgroup of the symmetric group on
$V$.}
We also regard the empty
set as a coset.

Our typical example
of a $(V,V')$-coset is the set $\ISO(G,G')$ of all isomorphism from a graph $G$
with vertex set $V$ to a graph $G'$ with vertex set $V'$. Indeed, if
$\ISO(G,G')$ is nonempty then for every isomorphism $\phi$ from $G$ to
$G'$ we have $\ISO(G,G')=\phi\Aut(G')$, where $\Aut(G')$ denote the
automorphism group of $G'$.

It is important to note that if $\sigma\Gamma$ is a coset,
then for all $\sigma'\in\sigma\Gamma$ it holds that
\begin{equation}
  \label{eq:32}
  \sigma^{-1}\sigma'\in\Gamma
\end{equation}
and
\begin{equation}
  \label{eq:33}
  \sigma\Gamma=\sigma'\Gamma.
\end{equation}
Indeed, writing $\sigma'=\sigma g'$ for some $g'\in\Gamma$ we have
$\sigma^{-1}\sigma'=g'\in\Gamma$ and 
$
\sigma g=\sigma'(g')^{-1}g\in\sigma'\Gamma
$ and
$
\sigma'g=\sigma g'g\in\sigma\Gamma$
for all $g\in\Gamma$.

If we want to
indicate that $\Lambda$ is a subcoset of a coset $\Lambda'$ (rather
than just a subset), we write $\Lambda\le\Lambda'$ (instead of
$\Lambda\subseteq\Lambda'$).
We will always assume
that nonempty cosets are succinctly represented by one explicit
bijection~$\sigma$ in~$\Lambda$ and a generating set for the
permutation group~$\Gamma$. There is a standard polynomial time algorithm that computes a succinct representation for a given coset.
Permutations from the group~$\Gamma$ will
be applied from left to right. That is, for~$g_1,g_2 \in \Gamma$, the
element~$g_1 g_2$ is the permutation that first applies~$g_1$ and then
applies~$g_2$. Note however that for maps~$\varphi$ we denote
by~$\varphi(x)$ the image of~$x$ under~$\varphi$. It will be apparent
from context which of the two notations is used. For more details on
the algorithmic theory of permutation groups we refer
to~\cite{seress2003permutation}.

Let~$G$
and~$G'$ be graphs with vertex sets $V,V'$, respectively, and let  $(D,\gamma)$ and~$(D',\gamma')$ be directed
decopmpositions of~$V$ and $V'$, respectively. 
An isomorphism~$\varphi $ from~$G$ to~$G'$ is said to
\emph{respect~$(D,\gamma)$ and~$(D',\gamma')$} if there is an isomorphism~$\widehat{\varphi}$ from~$D$ to~$D'$ such
that~$\gamma'(\widehat{\varphi}(t)) = {\varphi}(\gamma(t))$ for
all~$t\in V(D)$. We sometimes say that $\widehat{\phi}$ is an \emph{isomorphism
  from $(D,\gamma)$ to $(D',\gamma')$ extending $\phi$}.
We denote the coset of all isomorphism from $G$ to $G'$ that respect
$(D,\gamma)$ and~$(D',\gamma')$ by $\ISO(G_{D,\gamma},
G'_{D',\gamma'})$. Note that 
\[
\ISO(G_{D,\gamma},
G'_{D',\gamma'})\le\ISO(G,G').
\]

Before we proceed to our isomorphism test for graphs of bounded rank
width, we prove two simple lemmas about cosets. A \emph{least upper
  bound} for two $(V,V')$-cosets $\Lambda_1,\Lambda_2$ (in the lattice
of all $(V,V')$-cosets) is a $(V,V')$-coset $\Lambda$ such that
$\Lambda_1,\Lambda_2\le\Lambda$ and $\Lambda\le\Lambda'$ for all
cosets $\Lambda'$ with $\Lambda_1,\Lambda_2\le\Lambda$. Clearly, if a
least upper bound exists then it is unique. The next lemma shows that
indeed least upper bounds exist (this is easy to see) and can be computed in polynomial
time.

\begin{lem}\label{lem:patching:cosets}
  There is a polynomial time algorithm that, given two $(V,V')$-cosets
  $\Lambda_1,\Lambda_2$, computes the least upper bound of $\Lambda_1$
  and $\Lambda_2$.
\end{lem}

Of course we assume that both the input cosets $\Lambda_1,\Lambda_2$
and the output cosets are represented succinctly via generators.

\begin{proof}
Without loss of generality we assume that $\Lambda_1$ and
$\Lambda_2$ are nonempty.
Suppose that $\Lambda_i=\sigma_i\Gamma_i$, where
$\sigma_i:V\to V'$ and $\Gamma_i\le\Sym(V')$ is represented by
generators $g_{i1},\ldots,g_{it_i}$.
We let $\sigma:=\sigma_1$ and 
\begin{equation}
  \label{eq:31}
  \Gamma:=\Big\langle\big\{g_{1j}\bigmid
  j\in[t_1]\big\}
\cup\big\{\sigma_1^{-1}\sigma_2g_{2j}\bigmid j\in[t_2]\big\}
\cup\big\{\sigma_1^{-1}\sigma_2g_{2j}\sigma_2^{-1}\sigma_1\bigmid j\in[t_2]\big\}
\Big\rangle
\end{equation}
Let $\Lambda:=\sigma\Gamma$. Then $\Lambda_1\le\Lambda$, because
$\sigma_1=\sigma$ and $\Gamma_1=\angles{\{g_{1j}\mid
  j\in[t_1]\}}\le\Gamma$. To see that $\Lambda_2\le\Lambda$, let
$\sigma_2 g_2\in\Lambda_2$, where
$g_2=g_{2j_1}g_{2j_2}\cdots g_{2j_\ell}\in\Gamma_2$. Then
\[
\sigma_2g_2=\underbrace{\sigma_1}_{=\sigma}
\underbrace{(\sigma_1^{-1}\sigma_2 g_{2j_1}\sigma_2^{-1}\sigma_1)
(\sigma_1^{-1}\sigma_2 g_{2j_2}\sigma_2^{-1}\sigma_1)
\cdots
(\sigma_1^{-1}\sigma_2 g_{2j_{\ell-1}}\sigma_2^{-1}\sigma_1)
(\sigma_1^{-1}\sigma_2 g_{2j_{\ell-1}})}_{\in\Gamma}\in\sigma\Gamma=\Lambda.
\]
Let $\Lambda'=\sigma'\Gamma'$ be a coset with
$\Lambda_1,\Lambda_2\le\Lambda'$. We need to prove that
$\Lambda\le\Lambda'$.

As $\sigma_i\in\Lambda_i\le\Lambda'$, by \eqref{eq:33} we have
$\Lambda'=\sigma_i\Gamma'$, and thus by \eqref{eq:32} we have
$\sigma_i^{-1}\sigma_{3-i}\in\Gamma'$. 
We observe next that $\Gamma_i\le\Gamma'$. Indeed, let
$g\in\Gamma_i$. Then $\sigma_ig\in\Lambda_i\le\Lambda=\sigma_i\Gamma'$ and thus
$g=\sigma_i^{-1}\sigma_ig\in\Gamma'$.

Hence all generators of $\Gamma$ are in $\Gamma'$. It follows that
$\Gamma\le\Gamma'$ and thus $\Lambda=\sigma_1\Gamma\le\sigma_1\Gamma'=\Lambda'$.

Clearly, $\sigma$ and the system of generators for $\Gamma$ in
\eqref{eq:31} can be computed in polynomial time given the $\sigma_i$ and the
generators $g_{ij}$.
\end{proof}

\begin{lem}\label{lem:stabilizers:of:cosets}
Let~$\Lambda$ be a $(V,V')$-coset. Let $W\subseteq V$ and
$\phi:W\to V'$ an injective mapping. Then the set
\[
\Lambda_{[\phi]}:=\{\psi\in\Lambda\mid\psi(w)=\phi(w)\text{  for all }w\in
W\}
\]
is a subcoset of $\Lambda$. Furthermore, there is a polynomial time algorithm that,
given $\Lambda$, $W$, an $\phi$, computes $\Lambda_{[\phi]}$.
\end{lem}

\begin{proof}
  Without loss of generality we assume that
  $\Lambda=\sigma\Gamma\neq\emptyset$.  If there is a $\psi\in\Lambda$
  such that $\psi(w)=\phi(w)\text{ for all }w\in W$, then
  $\Lambda_{[\phi]}=\psi\Gamma_{\phi(W)}$, where $\Gamma_{\phi(W)}\le\Gamma$ is the pointwise
  stabiliser of $\phi(W)$ in $\Gamma$. Otherwise,
  $\Lambda_{[\phi]}=\emptyset$. Thus $\Lambda_{[\phi]}$ is a coset.

  For the algorithmic claim, it suffices to show that for
  given~$w\in V$ and~$w'\in V'$ it is possible to compute the
  subcoset~$\Lambda_{w,w'}$ of those elements in~$\Lambda=\sigma \Gamma$ that
  map~$w$ to~$w'$. For this we find an element~$g\in \Gamma$ that
  maps~$\sigma(w)$ to~$v'$ using the standard orbit algorithm
  (see~\cite{seress2003permutation}). If no such element exists then
  $\Lambda_{w,w'}=\emptyset$.
  Otherwise~$\Lambda_{w,w'}= \sigma g\Gamma_{w'}$
  where~$\Gamma_{w'}$ is the stabiliser of $w'$ in $\Gamma$. The
  stabiliser of an element can also be computed in polynomial time
  (see~\cite{seress2003permutation}).
\end{proof}

The main technical result of this section is the following theorem. Combined with
the Canonical Decomposition Theorem (Theorem~\ref{theo:candec}), it
yields a polynomial time isomorphism test for graphs of bounded rank
width (Theorem~\ref{theo:main}).

\begin{theo}\label{theo:coset}
  For every $k\in\NN$ there is a polynomial time algorithm that, given
  graphs $G$, $G'$ and normal treelike decompositions $(D,\gamma)$,
  $(D',\gamma')$ of $\rho_G$, $\rho_{G'}$, respectively, of width at
  most $k$, computes  a coset $\Lambda$ such that 
  \[
  \ISO(G_{D,\gamma}, G'_{D',\gamma'}) \leq \Lambda\leq \ISO(G,G').
  \]
\end{theo}

Note that $\Lambda$ may be empty if
$\ISO(G_{D,\gamma}, G'_{D',\gamma'})$ is empty, that is, if there is
no isomorphism from $G$ to $G'$ that respects the treelike
decompositions. It is worth noting that it is isomorphism-hard to
compute $\ISO(G_{D,\gamma}, G'_{D',\gamma'})$ even if $G,G'$ are
graphs with no edges and $(D,\gamma)$, $(D',\gamma')$ treelike
decompositions of width $0$. Thus computing
a coset that is sandwiched between
$\ISO(G_{D,\gamma}, G'_{D',\gamma'})$ and $\ISO(G,G')$ is a crucial
trick.

In the following, let us fix $k$ and $G,G'$ and $(D,\gamma)$,
$(D',\gamma')$.  Let $V:=V(G)$ and $V':=V(G')$. 

Let $t\in V(D)$. We shall define a graph $G_t$ that represents the
induced subgraph $G[\gamma(t)]$ as well as an ``abstraction'' of the edges from $\gamma(t)$
to $\bar{\gamma(t)}=V\setminus\gamma(t)$. Let~$\Type_t\subseteq\{0,1\}^{\gamma(t)}$ be the 
set of rows that appear in the
matrix~$M_{\bar{\gamma(t)}, \gamma(t)}$. 
Since the width of $(D,\gamma(t))$ is at most~$k$, the rank of the matrix
$M_{\bar{\gamma(t)}, \gamma(t)}$ is at most $k$, and thus
the set~$\Type_t$ has
size at most~$2^k$.  
We may view the elements $\mathbf w=(w_v\mid v\in\gamma(t))$ in~$W_t$ as
``types'', or equivalence classes of vertices $w\in
\bar{\gamma(t)}$, where two vertices $w,w'$ have the same type, or
are equivalent, if they have the same adjacencies with the vertices in
$\gamma(t)$. The entries of the vector $\mathbf w$ are these
adjacencies; $w_v=1$ means that all vectors of this type are adjacent
to $v$ and $w_v=1$ means that they are not adjacent.

Now we are ready to define the graph $G_t$. The
vertex set is 
\[
V(G_t):=\gamma(t)\cup\Type_t,
\]
and the edge set is
\begin{align*}
  E(G_t):=\,&\big\{vv'\bigmid v,v'\in\gamma(t)\text{ such that }vv'\in
              E(G)\big\}\\
  \cup&\big\{v\mathbf w\bigmid v\in\gamma(t)\text{ and }\mathbf
        w=(w_{v'}\mid v'\in\gamma(t))\in W_t\text{ such that }w_v=1\big\}.
\end{align*}
Thus $\Type_t$ is an indpendent set in $G_t$, and $G_t[\gamma(t)] =
G[\gamma(t)]$. We colour the graph~$G_t$ so that the vertices
in~$\Type_t$ are coloured red and all other vertices are coloured
blue.  We let $D_t$ be the induced subgraph of $D$ whose vertex set
consist of all vertices that are reachable from $t$ in $D$, and we let
$\gamma_t$ be the restriction of $\gamma$ to $V(D_t)$. Then
$(D_t,\gamma_t)$ is a normal treelike decomposition of
$V(G_t)\setminus W_t$. (We may
also view it as a partial treelike decomposition of $V(G_t)$ or even
$V(G)$; this does not matter, as we are not interested in the width of
this decomposition.)
Note that if~$r$ is the unique root of~$D$, which exists by
\ref{li:ntl4}, then we have~$G_r = G$
and~$(D_r,\gamma_r) = (D,\gamma)$. 

We define sets
$\Type'_t\subseteq\{0,1\}^{\gamma'(t)}$, graphs $G'_t$, and
decompositions $(D'_t,\gamma'_t)$ analogously for  all nodes $t\in V(D')$.

Recall that our goal is to compute a coset $\Lambda$ such that
  \[
  \ISO(G_{D,\gamma}, G'_{D',\gamma'}) \leq \Lambda\leq \ISO(G,G')
  \]
We do this by a dynamic programming algorithm that processes the
nodes of $D$ in a bottom-up manner (starting from the leaves). 
The next lemma describes the inductive step.

\begin{lem}
There is a polynomial time algorithm that, given
$G,G',(D,\gamma),(D',\gamma')$ as above and in addition
\begin{itemize}
\item nodes $t\in V(D)$ and $t'\in V(D')$;
\item for all $u\in N_+^D(t)$ and $u'\in N_+^{D'}(t')$ a coset $\Lambda(u,u')$
 satisfying~
 \begin{equation}
   \label{eq:34}
    \ISO((G_u)_{{D_u,\gamma_u}}, 
 ({G'}_{u'})_{{D'_{u'},\gamma'_{u'}}}) \le \Lambda(u,u') \le \ISO
 (G_u,G'_{u'}),
 \end{equation}
\end{itemize}
computes a coset $\Lambda$ such that
\[
\ISO((G_t)_{{D_t,\gamma_t}}, ({G'}_{t'})_{{D'_{t'},\gamma'_{t'}}})
\le\Lambda \le \ISO (G_t,G'_{t'}).
\]
\end{lem}

\begin{proof}
  Let $t\in V(D)$ and $t'\in V(D')$, and for all $u\in N_+^D(t)$ and
  $u'\in N_+^{D'}(t')$, let $\Lambda(u,u')$ be a coset satisfying
  \eqref{eq:34}.

  We split the lengthy proof into several steps.

  \medskip\noindent
  \textit{Step~1: The easy cases.}\\
  Suppose first that $t$ is a leaf of $D$. As the decomposition
  $(D,\gamma)$ is normal, by \ref{li:ntl2} we have $|\gamma(t)|=1$. In
  this case, we can compute $\ISO(G_t,G'_{t'})$ by brute force and let
  $\Lambda:=\ISO(G_t,G'_{t'})$.

  \medskip
  In the following, we assume that $t$ has at least one child. By \ref{li:ntl3} either
  all children of $t$ have the same cone or the children have mutually
  disjoint cones.
  
  Suppose that all children of $t$ have the same cone.  As
  $\beta(t)=\emptyset$, this implies that for all $u\in N_+^D(t)$ we
  have $\gamma(u)=\gamma(t)$. This implies $W_{u}=W_t$ and
  $G_u=G_t$. (But note that $(D_u,\gamma_u)\neq (D_t,\gamma_t)$!)

  If there is a child $u'$ of $t'$ with $\gamma'(u')\neq\gamma'(t')$,
  then
  $\ISO((G_t)_{{D_t,\gamma_t}},
  ({G'}_{t'})_{{D'_{t'},\gamma'_{t'}}})=\emptyset$,
  and we can let $\Lambda:=\emptyset$. So we assume that for
  all $u'\in N_+^{D'}(t')$ we have $\gamma'(u')=\gamma'(t')$.  Again,
  this implies $W'_{u'}=W'_{t'}$ and $G'_{u'}=G'_{t'}$.

  Hence for all $u\in N_+^D(t)$, $u'\in N_+^{D'}(t')$ we have
  $
  \ISO(G_u,G'_{u'})=\ISO(G_t,G'_{t'})
  $
  and thus
  \begin{equation}
    \label{eq:35}
    \ISO((G_u)_{{D_u,\gamma_u}}, 
 ({G'}_{u'})_{{D'_{u'},\gamma'_{u'}}}) \le\Lambda(u,u')\le\ISO(G_t,G'_{t'})
  \end{equation}
  We let $\Lambda$ be the least upper bound of the cosets
  $\Lambda(u,u')$ for $u\in N_+^D(t)$, $u'\in N_+^{D'}(t')$,
  which we can compute in polynomial time by
  Lemma~\ref{lem:patching:cosets}. Then by \eqref{eq:35} , for all $u\in N_+^D(t)$, $u'\in N_+^{D'}(t')$ we have
  \[
  \ISO((G_u)_{{D_u,\gamma_u}}, 
 ({G'}_{u'})_{{D'_{u'},\gamma'_{u'}}}) \le\Lambda\le\ISO(G_t,G'_{t'})
 \]
 It remains to prove that $\ISO((G_t)_{{D_t,\gamma_t}},
 ({G'}_{t'})_{{D'_{t'},\gamma'_{t'}}})\le\Lambda$. Let $\psi\in \ISO((G_t)_{{D_t,\gamma_t}},
 ({G'}_{t'})_{{D'_{t'},\gamma'_{t'}}})$, and let $\hat\psi$ be an
 isomorphism form $(D_t,\gamma_t)$ to $(D'_{t'},\gamma'_{t'})$
 extending $\psi$. Let $u\in N_+^D$ and $u':=\hat\psi(u)$. Then the
 restriction of $\hat\psi$ to $D_u$ is an isomorphism from
 $(D_u,\gamma_u)$ to $(D'_{u'},\gamma'_{u'})$ that extends the
 isomorphism $\psi$ from $G_u=G_t$ to $G'_{u'}=G'_{t'}$. Hence $\psi\in \ISO((G_u)_{{D_u,\gamma_u}}, 
 ({G'}_{u'})_{{D'_{u'},\gamma'_{u'}}}) \le\Lambda$.

 \medskip\noindent
  \textit{Step~2: Fixing the outside.}\\
 In the following, we assume that the children of $t$ have
 mutually disjoint cones. If the children of $t'$ have the same cone,
 then
 $\ISO((G_t)_{{D_t,\gamma_t}},
 ({G'}_{t'})_{{D'_{t'},\gamma'_{t'}}})=\emptyset$,
 and we can let $\Lambda:=\emptyset$. So we assume that the
 children of $t'$ have mutually disjoint cones as well.

The colouring of the graphs $G_t$ and $G'_{t'}$ forces each
isomorphism $\psi\in\ISO(G_t,G'_{t'})$ to map $W_t$ bijectively to
$W'_{t'}$. 
For each bijection $\phi:W_t\to W'_{t'}$, we shall compute
a coset $\Lambda_\phi$ such that
\[
\ISO((G_t)_{{D_t,\gamma_t}}, ({G'}_{t'})_{{D'_{t'},\gamma'_{t'}}})_{[\phi]}
\le\Lambda_\phi \le \ISO (G_t,G'_{t'});
\]
recall that 
$\ISO((G_t)_{{D_t,\gamma_t}},
({G'}_{t'})_{{D'_{t'},\gamma'_{t'}}})_{[\phi]}$ denotes the subcoset
of $\ISO((G_t)_{{D_t,\gamma_t}},
({G'}_{t'})_{{D'_{t'},\gamma'_{t'}}})$ consisting of all isomorphisms
$\psi$ whose restriction to $W_t$ is $\phi$. 
Then we let $\Lambda$ be the least upper bound of the cosets
$\Lambda_\phi$ (see Lemma~\ref{lem:patching:cosets}). 
As $|W_t|,|W'_{t'}|\le 2^k$, there are at most $2^k!$ bijections
$\phi:W_t\to W'_{t'}$, and if we can compute each $\Lambda_\phi$ in
polynomial time, we can also compute $\Lambda$ in polynomial time (for
fixed $k$).

From now on, we fix a bijection $\phi:W_t\to W'_{t'}$.

Let $u\in N_+^D(t)$ and $u'\in N_+^{D'}(t')$. Recall that
$\gamma(u)\subseteq\gamma(t)$.  The \emph{$u$-projection} of a
vector $\mathbf w=(w_v\mid v\in\gamma(t))\in\{0,1\}^{\gamma(t)}$ is the
vecor $\mathbf w|_u=(w_v\mid
v\in\gamma(u))\in\{0,1\}^{\gamma(u)}$. We observe that for every vector
$\mathbf w\in W_t$ we have $\mathbf w|_u\in W_u$. To see this,
recall that $\mathbf w$ appears as a row in the matrix
$M_{\bar{\gamma(t)},\gamma(t)}$. Let $w\in\bar\gamma(t)$ be the index
of this row. Then $w\in\bar{\gamma(u)}\supseteq\bar{\gamma(t)}$, and
hence there is also a row with index $w$ in the matrix
$M_{\bar{\gamma(u)},\gamma(u)}$, and this row is precisely $\mathbf w|_u$.
Similarly, for $u'\in N_+^{D'}(t')$, every $\mathbf w'\in W'_{t'}$ has a $u'$-projection $\mathbf
w'|_{u'}$ in $W'_{u'}$.

Now let $\psi\in\Lambda(u,u')$. Then $\psi$ induces a bijection
from $W_u$ to $W'_{u'}$.
We say that $\psi$ \emph{agrees} with $\phi$ if for every $\mathbf
w=(w_v\mid v\in\gamma(t))\in W_t$ it holds that $\psi(\mathbf
w|_u)=\phi(\mathbf w)|_{u'}$. Observe that with $\mathbf w'=(w'_{v'}\mid v'\in
\gamma'(t')):=\phi(\mathbf w)$, this implies $w_v=w'_{\psi(v)}$ for
all $v\in\gamma(u)$. Indeed, for $v\in\gamma(u)$ we have
\[
w_v=1\iff v\mathbf w|_u\in E(G_u)\iff \psi(v)\mathbf w'|_{u'}\in
E(G'_{u'})\iff w'_{\psi(v)}=1,
\]
because $\psi$ is an isomorphism from $G_u$ to $G'_{u'}$.
We denote the subcoset of all $\psi\in\Lambda(u,u')$ that agree with
$\phi$ by $\Lambda(u,u')_{[\phi]}$. This set can be computed in polynomial time by Lemma~\ref{lem:stabilizers:of:cosets}.

 \medskip\noindent
  \textit{Step~3: Fixing all extensions.}\\
Let~$P\in\{0,1,?\}^{V\times V}$ and $P'\in\{0,1,?\}^{V'\times V'}$ be the $?$-block matrices associated with $t$ and~$t'$,
respectively. By Lemma~\ref{lem:pr_vs_width}, both $P$ and $P'$ have
partition rank at most $k$.
By~Theorem~\ref{thm:compute:canonical:row:extensions}, we can compute
canonical extension sets~$\Ext\subseteq\{0,1\}^V$, $\Ext'\subseteq\{0,1\}^{V'}$ for $P,P'$ of size at most
$e(k)$. Recall that for a vertex $v\in\gamma(t)$, by $\Ext(v)$ we
denote the set of all extensions of the row $\mathbf p_v$ of $P$ in
$\Ext$, and similarly for $v'\in V'$ by $\Ext'(v')$ the extensions of
$\mathbf p'_{v'}$.

If $\lvert \Ext\rvert\neq\lvert \Ext'\rvert$ then $\ISO (G_t,G'_{t'})=\emptyset$, and we let
$\Lambda:=\emptyset$. Suppose otherwise. 
Let~$\chi\colon \Ext \rightarrow \Ext'$ be an arbitrary bijection. 
We say that an isomorphism~$\psi$ from~$G_t$ to~$G'_{t'}$ \emph{agrees}
with~$\chi$ if for all
$v\in\gamma(t)$ we have $\chi(\Ext(v))=\Ext'(\psi(v))$ and for every $\mathbf x=(x_v\mid v\in V)\in\Ext$ with
$\chi(\mathbf x)=(x'_{v'}\mid v'\in V')$ and every
$v\in \gamma(t)$ it holds that $x_v=x'_{\psi(v)}$.

We argue that every $\psi\in\ISO((G_t)_{{D_t,\gamma_t}},
({G'}_{t'})_{{D'_{t'},\gamma'_{t'}}})$ agrees with some bijection
$\chi\colon \Ext \rightarrow \Ext'$. To see this, note first that
$\psi$ induces an isomorphism from $P$ to $P'$. Since~$\Ext$ and~$\Ext'$
are canonical, there is a bijection $\chi\colon\Ext\to\Ext'$ such
that~$\chi(\Ext(v) ) = \Ext(\psi(v))$ for all~$v\in V$ and for every $\mathbf x=(x_v\mid v\in V)\in \Ext$ with $\chi(\mathbf x)=:\mathbf x'=(x'_{v'}\mid
v'\in V')\in\Ext'$ and every $v\in V$ we have $x_v=x'_{\psi(v)}$.

In the following, we fix a bijection $\chi\colon \Ext \rightarrow
\Ext'$, and we let $\ISO((G_t)_{{D_t,\gamma_t}},
({G'}_{t'})_{{D'_{t'},\gamma'_{t'}}})_{[\phi,\chi]}$ be the coset
consisting of all $\psi\in \ISO((G_t)_{{D_t,\gamma_t}},
({G'}_{t'})_{{D'_{t'},\gamma'_{t'}}})$ that agree with $\phi$ and
$\chi$. By the usual argument based on
Lemma~\ref{lem:patching:cosets}, it suffices to compute a coset
$\Lambda_{\phi,\chi}$ such that
\[
\ISO((G_t)_{{D_t,\gamma_t}},
({G'}_{t'})_{{D'_{t'},\gamma'_{t'}}})_{[\phi,\chi]}\le\Lambda_{\phi,\chi}\le\ISO(G_t,G'_{t'}).
\]
Let $u\in N_+^D(t)$ and $u'\in N_+^{D'}(t')$. We say that a
$\psi\in \Lambda(u,u')$ \emph{agrees} with $\chi$ if for all
$v\in\gamma(u)$ we have $\chi(\Ext(v))=\Ext'(\psi(v))$ and for every $\mathbf x=(x_v\mid v\in V)\in\Ext$ with
$\chi(\mathbf x)=(x'_{v'}\mid v'\in V')$ and every
$v\in \gamma(u)$ it holds that $x_v=x'_{\psi(v)}$. By
$\Lambda(u,u')_{[\phi,\chi]}$ we denote the
subcoset of $\Lambda(u,u')$ consisting of all $\psi$ that agree with
both $\phi$ and $\chi$. We can compute $\Lambda(u,u')_{[\phi,\chi]}$
by Lemma~\ref{lem:stabilizers:of:cosets}.

Let us write
$\Lambda(u,u')_{[\phi,\chi]}=\sigma_{u,u'}\Gamma'_{u,u'}$, where
$\Gamma'_{u,u'}\le\Aut(G'_{u'})$. Observe that
\begin{align}
  \label{eq:37}
 & g(\mathbf w')=\mathbf w'&\text{for all $g\in \Gamma'_{u,u'}$ and
    $\mathbf w'\in W'_{u'}$.}
\end{align}
This follows from the fact that all
$\psi\in\Lambda(u,u')_{[\phi,\chi]}$ agree with $\phi$ on
$W_u$. Indeed, applying this to $\psi_1=\sigma_{u,u'}$ and
$\psi_2=\sigma_{u,u'}g^{-1}$, for $\mathbf w'\in W'_{u'}$ we have
\[
\sigma_{u,u'}^{-1}(\mathbf w')=\psi_1^{-1}(\mathbf
w')=\phi^{-1}(\mathbf w')=\psi_2^{-1}(\mathbf
w')=\sigma_{u,u'}^{-1}(g(\mathbf w'))
\]
and thus $\mathbf w'=g(\mathbf w')$.
Furthermore, 
\begin{align}
    \label{eq:38}
&x'_v=x'_{g(v)}&\text{for all $\mathbf x'=(x'_v\mid v\in V')\in\Ext'$ and all
    $v\in\gamma'(u')$.}
\end{align}
This follows from the fact that all
$\psi\in\Lambda(u,u')_{[\phi,\chi]}$ agree with $\chi$. Let $\psi_1=\sigma_{u,u'}$ and
$\psi_2=\sigma_{u,u'}g$. Then 
\[
x'_v=x_{\psi_1^{-1}(v)}=x'_{\psi_2(\psi_1^{-1}(v))}=x'_{g(v)}.
\]
For each
$g\in\Gamma'_{u,u'}$ we define a permutation $\hat g$ of
$V(G'_{t'})$ by
\[
\hat g(v):=
\begin{cases}
  g(v)&\text{if }v\in\gamma'(u'),\\
  v&\text{if }v\in\gamma'(t')\setminus\gamma'(u'),\\
  v&\text{if }v\in W'_{t'}.
\end{cases}
\]
We let $\hat \Gamma'_{u,u'}$ be the set of all
$\hat g$ for $g\in \Gamma'_{u,u'}$. Then $\hat \Gamma'_{u,u'}$ is a
subgroup of $\Sym(V(G'_{t'}))$. In fact, the groups $\Gamma'_{u,u'}$
and $\hat \Gamma'_{u,u'}$ are identical
as abstract groups, they only differ in their permutation action.

\begin{claim}\label{claim:dp1}
  $\hat\Gamma'_{u,u'}\le\Aut(G'_{t'})$.

  \proof
  Let $\hat g\in \Gamma'_{u,u'}$ and $v_1,v_2\in V(G'_t)$ and
  $v_i':=\hat g(v_i')$. We shall prove 
  \begin{equation}
    \label{eq:36}
    v_1v_2\in
    E(G'_{t'})\iff v_1'v_2'\in E(G'_{t'}).
  \end{equation}
  Clearly, this equivalence holds if $v_1,v_2\in\gamma'(u')$,
  because $g$ is an automorphism of $G'_{u'}$, and also if
  $v_1,v_2\not\in\gamma'(u')$, because then we have $v_1=v_1'$ and
  $v_2=v_2'$.

  So let us assume that $v_1\in\gamma'(u')$ and
  $v_2\not\in\gamma'(u')$. Then $v_2=\hat g(v_2)$. Suppose first that $v_2=\mathbf w'=(w'_v\mid
  v\in \gamma'(t'))\in W'_{t'}$. Then the $u'$-projection
  $\mathbf w'|_{u'}=(w'_v\mid v\in\gamma'(u'))$ is in
  $W'_{u'}$. By \eqref{eq:37}, we have
  \begin{align*}
    v_1v_2\in
    E(G'_{t'})&\iff w'_{v_1}=1\iff v_1\mathbf w'|_{u'}\in
                              E(G'_{u'})\\
    &\iff
      g(v_1)\mathbf w'|_{u'}\in E(G'_{u'})
      \iff w'_{g(v_1)}=1
      \iff g(v_1)g(v_2)\in E(G'_{t'}).
  \end{align*}
  It remains to consider the case
  $v_2\in\gamma'(t')\setminus\gamma'(u')$. Then there is a $u''\in
  N_+^{D'}(t')\setminus\{u'\}$ such that $v_2\in\gamma'(u'')$. Let
  $\mathbf x'=(x'_{v'}\mid v'\in V')\in\Ext'(v_2)$. 
  As $v_1$ is not contained in $\gamma'(u'')$, the $?$-index of $v_2$,
  and $\mathbf x'$ is an extension of $\mathbf p'_{v_2}=(p'_{v_2v}\mid
  v\in V')$, we have
  \[
  v_1v_2\in E(G'_{t'})\iff p'_{v_2v_1}=1\iff x'_{v_1}=1,
  \]
  and similarly $g(v_1)v_2\in E(G'_{t'})\iff x'_{g(v_1)}=1$. By \eqref{eq:38}
  and $v_2=g(v_2)$ we thus have \eqref{eq:36}.
  \uend
\end{claim}

 \medskip\noindent
  \textit{Step~4: Admissible bijections.}\\
A bijection $\alpha \colon N_+^D(t)\to N_+^{D'}(t')$  is
\emph{admissible} if $\Lambda(u,\alpha(u))_{[\phi,\chi]}$ is nonempty for all~$u\in
N_+^D(t)$. 

Let $\psi\in\Iso(G_t,G'_{t'})$ and $u\in N_+^D(t)$, $u'\in
N_+^{D'}(t')$ such that $\psi(\gamma(u))=\gamma'(u')$. We define
the \emph{induced mapping} $\psi|_u:V(G_u)\to V(G'_{u'})$ by
\[
\psi|_u(v):=
\begin{cases}
  \psi(v)&\text{if }v\in\gamma(u),\\
  \psi(\mathbf w)|_{u'}&\text{if }v=\mathbf w|_u
  \text{ for some }\mathbf w\in W_t.
\end{cases}
\]

\begin{claim}[resume]\label{claim:dp2a}
  $\psi|_u\in\Iso(G_u,G_u')$. Furthermore, if $\psi$ agrees with
  $\phi$ and $\chi$, then so does $\psi|_u$.

  \proof We first need to prove that $\psi|_u$ is well-defined. Let
  $\mathbf w_1,\mathbf w_2\in W_t$ such that
  $\mathbf w_1|_u=\mathbf w_2|_u$. Suppose that
    $\mathbf w_i=(w_{iv}\mid v\in \gamma(t))$, and let
    $\mathbf w_i'=(w'_{iv'}\mid v'\in\gamma'(t')):=\psi(\mathbf
    w_i)$. Then for all $v\in\gamma(t)$ it holds that
  \[
  w_{iv}=1\iff v\mathbf w_i\in E(G_t)\iff \psi(v)\mathbf w'_i\in
  E(G'_{t'})\iff w'_{i\psi(v)}=1.
  \]
  Thus 
  \[
  \mathbf w_i|_u=(w_{iv}\mid v\in\gamma(u))
  =(w'_{i\psi(v)}\mid v\in\gamma(u))=(w'_{i v'}\mid
  v'\in\gamma'(u'))=\mathbf w_i'|_{u'},
  \]
  because $\psi(\gamma(u))=\gamma'(u')$. As $\mathbf
  w_1|_u=\mathbf w_2|_u$, it follows that $\mathbf
  w_1'|_{u'}=\mathbf w_2'|_{u'}$. A similar argument
  shows that $\psi|_u$ is a bijection.

  It now follows directly from the fact that $\psi$ is an isomorphism
  from $G_t$ to $G'_{t'}$ that $\psi|_u$ is an isomorphism, observing
  that for $v\in\gamma(u)$ and $\mathbf w\in W_t$ it holds that
  $v\mathbf w\in E(G_t)\iff v\mathbf w|_u\in E(G_u)$.

  The second claim follows directly from the definitions of a mapping in
  $\Iso(G_u,G'_{u'})$ agreeing with $\phi$ or $\chi$.
  \uend
\end{claim}

\begin{claim}[resume]\label{claim:dp2}
  Let
  $\psi\in\ISO((G_t)_{{D_t,\gamma_t}},
  ({G'}_{t'})_{{D'_{t'},\gamma'_{t'}}})_{[\phi,\chi]}$.
  Then there is an admissible bijection $\alpha_\psi$ such that for every
  $u\in N_+^D(t)$ it holds that
  $\psi(\gamma(u))=\gamma'(\alpha_\psi(u))$ and the induced mapping
  $\psi|_u$ is in $\ISO((G_u)_{{D_u,\gamma_u}},
  ({G'}_{\alpha_\psi(u)})_{{D'_{\alpha_\psi(u)},\gamma'_{\alpha_\psi(u)}}})_{[\phi,\chi]}$.

  \proof The mapping $\psi$ has an extension $\hat\psi$ that is an
  isomorphism from $(D_t,\gamma_t)$ to $(D'_{t'},\gamma'_{t'})$. Let
  $\alpha_\psi$ be the restriction of $\hat\psi$ to $N_+^D(t)$. Then
  $\alpha_{\psi}$ is a bijection from $N_+^D(t)$ to $N_+^{D'}(t')$.
  
  Let $u\in N_+^D(t)$. It follows from Claim~\ref{claim:dp2a} that the
  induced mapping $\psi|_u$ is in
  $\Iso(G_u,G'_{\alpha_\psi(u)})_{[\phi,\chi]}$.  The restriction of
  $\hat\psi$ to $D_u$ is an isomorphism from $(D_u,\gamma_u)$ to
  $(D'_{\alpha_\psi(u)}, \gamma_{\alpha_\psi(u)})$, which shows that
  $\psi|_u$ respects the decompositions. Hence $\psi|_u\in\ISO((G_u)_{{D_u,\gamma_u}},
  ({G'}_{\alpha_\psi(u)})_{{D'_{\alpha_\psi(u)},\gamma'_{\alpha_\psi(u)}}})_{[\phi,\chi]}$.
  \uend
\end{claim}

Thus if no admissible bijection exists, then
$\ISO((G_t)_{{D_t,\gamma_t}},
({G'}_{t'})_{{D'_{t'},\gamma'_{t'}}})_{[\phi,\chi]}=\emptyset$,
and we let $\Lambda_{\phi,\chi}:=\emptyset$. In the following, we
assume that there is an admissible bijection. We can find one by
computing a perfect matching in the bipartite graph in which~$u$ is
adjacent to~$u'$ if~$\Lambda(u,u')_{[\phi,\chi]}$ is nonempty.

Let $\alpha$ be an admissible bijection. For all $u\in N_+^D(t)$, let $\psi_{u,\alpha(u)}\in
\Lambda(u,\alpha(u))_{[\phi,\chi]}$.
We combine the
$\psi_{u,\alpha(u)}$ to a mapping $\psi_\alpha$ defined by
\begin{equation}
  \label{eq:40}
  \psi_\alpha(v):=
\begin{cases}
  \phi(v)&\text{if $v\in W_t$},\\
  \psi_{u,{\alpha}(u)}(v)&\text{if $v\in\gamma(u)$ for some $u\in N_+^D(t)$.}
\end{cases}
\end{equation}
Observe that for all $u\in N_+^D(t)$ we have $\psi_\alpha|_u=\psi_{u,\alpha(u)}$.

\begin{claim}[resume]\label{claim:dp3}
$\psi_\alpha$ is an isomorphism from~$G_t$ to~$G'_{t'}$. 

\proof
As $\gamma(t)$ is the disjoint union of $W_t$ and the $\gamma(u)$ and
$\gamma'(t')$ is the disjoint union of $W'_{t'}$  the $\gamma'(u')$, the mapping
$\psi_\alpha$ is well-defined and bijective.

Let $v_1,v_2\in V(G'_t)$ and $v_i':=\psi_\alpha(v_i)$. We shall
prove 
\begin{equation}
  \label{eq:39}
  v_1v_2\in
  E(G_{t})\iff v_1'v_2'\in E(G'_{t'}).
\end{equation}
\begin{cs}
  \case1
  $v_1,v_2\in\gamma(u)$ for some $u\in N_+^D(t)$.

  Then
  \begin{align*}
    v_1v_2\in E(G_t)&\iff
                      v_1v_2\in E(G_u)\\
                    &\iff \psi_{u,{\alpha}(u)}(v_1) \psi_{u,{\alpha}(u)}(v_2)\in
                      E(G'_{{\alpha}(u)})\iff v_1'v_2'\in E(G'_{t'}),
  \end{align*}
  because $\psi_{u,{\alpha}(u)}\in \Lambda(u,\alpha(u))_{[\phi,\chi]}\le\ISO(G_u,G'_{\alpha(u)})$.

  \case2
  $v_1\in\gamma(u)$ for some $u\in N_+^D(t)$ and $v_2\in
  W_t$.

  Say, $v_2=\mathbf w=(w_v\mid v\in \gamma(t))$ and $v_2'=\phi(\mathbf
  w)=\mathbf w'=(w'_{v'}\mid v'\in \gamma'(t'))$. Then the
  $u$-projection $\mathbf w|_u$ is in $W_u$ and the
  $u'$-projection $\mathbf w'|_{\alpha(u)}$ is in
  $W'_{\alpha(u)}$. As $\psi_{u,{\alpha}(u)}\in\ISO(G_u,G'_{\alpha(u)})$, we have
    \begin{align*}
    v_1\mathbf w\in
    E(G_{t})&\iff w_{v_1}=1\iff v_1\mathbf w|_u\in
                              E(G_{u})\\
    &\iff
      v_1'\mathbf w'|_{\alpha(u)}\in E(G'_{\alpha(u)})
      \iff w'_{v_1'}=1
      \iff v_1'\mathbf w'\in E(G'_{t'}).
  \end{align*}

  \case2
  $v_2\in\gamma(u)$ for some $u\in N_+^D(t)$ and $v_1\in
  W_t$.

  Symmetric to Case 2.

  \case3
  $v_1,v_2\in W_t$.

  Then $v_1v_2\not\in E(G_t)$ and $v_1'v_2'\not\in E(G'_{t'})$.

  \case4 $v_1\in\gamma(u_1)$ and $v_2\in\gamma(u_2)$ for distinct
  $u_1,u_2\in N_+^D(t)$.

  As $\psi_{u_2,{\alpha}(u_2)}$ agrees with $\chi$, we have
  $\chi(\Ext(v_2)) = \Ext(v_2')$. Let
  $\mathbf x=(x_v\mid v\in V)\in\Ext(v_2)$, and let
  $\mathbf x'=(x'_{v'}\mid v'\in V'):=\chi(\mathbf x)$. Then
  $\mathbf x'\in\Ext(v_2')$. Furthermore, $x_{v_1}=x'_{v_1'}$. As
  $v_1\not\in\gamma(u_2)$, the $?$-index of the row $\mathbf p_{v_2}$ of
  $P$, we have $p_{v_2v_1}\in\{0,1\}$, and as $\mathbf x$ is an extension
  of $\mathbf p_{v_2}$, this means that
  $x_{v_1}=p_{v_1v_2}$. Similarly, $x'_{v_1'}=p'_{v_1'v_2'}$. Thus
  $p_{v_1v_2}=p'_{v_1'v_2'}\in\{0,1\}$, and this implies \eqref{eq:39}.
  \uend
\end{cs}
\end{claim}

\medskip\noindent\textit{Step~5: Closure of the set of admissible
  bijections.}\\
In this step we prove that without loss of generality we may assume
that the sets $\Lambda(u,u')_{[\phi,\chi]}\neq\emptyset$ have the
following closure property.
\begin{ealph}
\item \label{closure:prop} For every pair of
  sequences~$u_1,u_2,\ldots, u_t\in N_+^D(t) $
  and~$u'_1,u'_2,\ldots, u'_t\in N_+^{D'}(t')$, if for all~$i\in [t]$
  we have~$\Lambda(u_i,u'_i)_{[\phi,\chi]}\neq \emptyset$ and for
  all~$j \in [t-1]$ we
  have~$\Lambda(u_{j+1},u'_j)_{[\phi,\chi]} \neq \emptyset$
  then~$\Lambda(u_0,u'_t)_{[\phi,\chi]} = \emptyset$.
\end{ealph}
We achieve this as follows. For all~$u= u_1,u' = u'_t$ for
which there are sequences~$u_1\ldots, u_t\in N_+^D(t) $
and~$u'_1,\ldots, u'_t\in N_+^{D'}(t')$ with said properties, we
pick for all~$i\in [t]$ an
element~$\nu_i \in \Lambda(u_i,u'_i)_{[\phi,\chi]}$ and for
all~$j\in [t-1]$ an
element~$\mu_j \in \Lambda(u_{j+1},u'_j)_{[\phi,\chi]}$. We then add
the map~$\nu_1 \mu_1^{-1}\nu_2\mu_2^{-2}\ldots \mu_{t-1}^{-1}\nu_t$ to
the set~$\Lambda(u_0,u'_t)_{[\phi,\chi]}$.  It follows from the fact
that isomorphisms compose that this map is an isomorphism
in~$\Iso(G_{u_0},G'_{u'_t})_{[\phi,\chi]}$.
 
It is also easy to verify that if we perform this procedure once for
all pairs~$u,u'$ (for which there are sequences with the properties
described above) then Property~\ref{closure:prop} is fulfilled. We
thus assume from now on that Property~\ref{closure:prop} holds.  Note
that Property~\ref{closure:prop} implies that for admissible
bijections~$\alpha_1,\alpha_2,\alpha_3$ the
map~$\alpha_1\alpha_2^{-1}\alpha_3$ is also admissible.

 \medskip\noindent
  \textit{Step~6: Construction of the generating set.}\\
We fix an admissibble bijection $\alpha^0$ and for all $u\in N_+^D(t)$, $u'\in
N_+^{D'}(t')$ such that $\Lambda(u,u')_{[\phi,\chi]}\neq\emptyset$, we
fix a mapping $\psi^0_{u,u'}\in
\Lambda(u,u')_{[\phi,\chi]}$. The choice of these mappings is arbitrary and need not
be canonical. For every admissible bijection $\alpha$ we define
$\psi^0_\alpha$ as in \eqref{eq:40} with the map $\psi^0_{u,u'}$ instead
of $\psi_{u,u'}$. Furthermore, we let $\psi^0:=\psi^0_{\alpha^0}$.

We say that an admissible bijection~$\alpha$ \emph{differs from~$\alpha^0$
in at most three positions} if there are at most three~$u\in N_+^D(t)$
for which~$\alpha(u) \neq \alpha^0(u)$. We denote the set of
bijections that differ from~${\alpha^0}$ in at most three positions
by~${\alpha^0}\pm 3$. Note that the set~${\alpha^0} \pm 3$ can be
computed in polynomial time.

Recall the definition of $\hat\Gamma_{u,u'}'$ from Step~3, and let
$S_{u,u'}$ be a generating set for $\hat\Gamma_{u,u'}'$ (the succinct
representation of $\Lambda(u,u')_{[\phi,\chi]}$ contains such a
generating set). We let
\[
S:=\big\{(\psi^0_{\alpha})^{-1}\psi^0_{\alpha'} \bigmid \alpha,\alpha'\in ({\alpha^0} \pm 3)\big\}\cup \bigcup_{u,u'} S_{u,u'},
\]
and we let
\[
\Lambda_{\phi,\chi}:=\psi^0\langle S\rangle.
\]
Note that $\psi^0\in \Lambda_{\phi,\chi}$ and
$\psi^0_{\alpha}=\psi^0(\psi^0)^{-1}\psi^0_\alpha\in
\Lambda_{\phi,\chi}$ for all $\alpha\in\alpha^0\pm3$.

\begin{claim}[resume]\label{claim:dp4}
$\ISO((G_t)_{{D_t,\gamma_t}},
({G'}_{t'})_{{D'_{t'},\gamma'_{t'}}})_{[\phi,\chi]}\le\Lambda_{\phi,\chi}\le\ISO(G_t,G'_{t'}).$

\proof
It follows from Claim~\ref{claim:dp3} that
$\psi^0(\psi^0_\alpha)^{-1}\psi^0_{\alpha'}\in\ISO(G_t,G'_{t'})$
for all admissible bijections $\alpha,\alpha'$, and it follows from
Claim~\ref{claim:dp1} that $\psi^0g\in\ISO(G_t,G'_{t'})$ for all $g\in
S_{u,u'}\subseteq\hat\Gamma'_{u,u'}$. This proves $\Lambda_{\phi,\chi}\le\ISO(G_t,G'_{t'}).$

To prove
$\ISO((G_t)_{{D_t,\gamma_t}},
({G'}_{t'})_{{D'_{t'},\gamma'_{t'}}})_{[\phi,\chi]}\le\Lambda_{\phi,\chi}$,
we consider the subset $\mathrm K$ of $\ISO(G_t,G'_{t'})_{[\phi,\chi]}$
consisting of all $\psi$ satisyfing the following two conditions:
\begin{eroman}
  \item
    There is an admissible bijection $\alpha_\psi$ such that
    $\psi(\gamma(u))=\gamma'(\alpha_\psi(u))$ for all $u\in
      N_+^D(t)$.

    \item $\psi|_u\in\Lambda(u,\alpha_{\psi}(u))_{[\phi,\chi]}$ for all $u\in
      N_+^D(t)$.
\end{eroman}
It follows
from Claim~\ref{claim:dp2} that every
$\psi\in \ISO((G_t)_{{D_t,\gamma_t}},
({G'}_{t'})_{{D'_{t'},\gamma'_{t'}}})_{[\phi,\chi]}$ is contained in~$K$. Furthermore,
for every admissible bijection $\alpha$ we have $\psi^0_\alpha\in
\mathrm K$, because
$\psi^0_{\alpha}|_u=\psi^0_{u,\alpha(u)}\in\Lambda(u,\alpha_{\psi}(u))_{[\phi,\chi]}$
for all $u\in N_+^D(t)$. In particular, $\psi^0=
\psi^0_{\alpha_0}\in\mathrm K$.

We shall prove that 
\[
\mathrm K\subseteq \Lambda_{\phi,\chi}.
\]
For $\psi\in\mathrm K$, we prove 
\[
\psi\in\Lambda_{\phi,\chi}.
\]
The map~$\psi^0_{\alpha_\psi}$ agrees with~$\psi$ up to elements in~$\hat\Gamma_{u,\alpha_{\psi}(u)}'$. More precisely,  
$\psi|_u\in\Lambda(u,\alpha_{\psi}(u))_{[\phi,\chi]}$
can be written as $\psi|_u=\psi^0_{u,\alpha_{\psi}(u)}g_{u}$ for some
  $g_{u}\in \Gamma_{u,\alpha_{\psi}(u)}'$. As the cones $\gamma(u')$ for
  $u'\in N_+^D(t')$ are mutually disjoint, the permutations $\hat g_u\in\hat\Gamma_{u,\alpha^0(u)}'$
  (see Step~3) commute. Thus we have
  \[
  \psi=\psi^0_{\alpha_{\psi}}g_{u_1}\cdots g_{u_m}
  \]
  for an arbitrary enumeration $u_1,\ldots,u_m$ of $N_+^D(t)$. This
  proves that \[\psi^0_{\alpha_\psi}\in\Lambda_{\phi,\chi} \iff \psi \in\Lambda_{\phi,\chi}.\]

  It thus suffices to show that~$\psi^0_{\alpha_\psi}\in
  \Lambda_{\phi,\chi}$. More generally, we now show for an arbitrary
  admissible
  bijection~${\alpha^*}$
  that~$\psi^0_{\alpha^*}\in\Lambda_{\phi,\chi}$.
  We show this by induction on the number $d({\alpha^*})$
  of $u\in N_+^D(t)$ such that ${\alpha^*}(u)\neq\alpha^0(u)$.

In the base case $d({\alpha^*})=0$ we have ${\alpha^*}=\alpha^0$ and thus~$\psi^0_{\alpha^*} = \psi^0\in\Lambda_{\phi,\chi}$.

  For the inductive step, suppose that ${\alpha^*}\neq\alpha^0$. Then
  there are at least two $u\in  N_+^D(t)$ such that
  ${\alpha^*}(u)\neq\alpha(u)$.
  \begin{cs}
    \case1 There are distinct $u_1,u_2\in N_+^D(t)$ such that
    ${\alpha^*}(u_1)=\alpha^0(u_2)$ and
    ${\alpha^*}(u_2)=\alpha^0(u_1)$.

    Let $\alpha$ be the bijection from $N_+^D(t)$ to $N_+^{D'}(t')$
    with $\alpha(u_i)={\alpha^*}(u_i)$ for $i=1,2$ and
    $\alpha(u)=\alpha^0(u)$ for all $u\neq u_1,u_2$. Then $\alpha$ is
   admissible, because
   $\alpha(u_i) = {\alpha^*}(u_i)$ for $i=1,2$ and~$\alpha(u) = \alpha^0(u)$ for all
      $u\neq u_1,u_2$. Hence $\alpha\in\alpha^0\pm 3$.

   Let $\alpha'={\alpha^*}\alpha^{-1}\alpha^0$, which is admissible by Property~\ref{closure:prop}. Then
   $d(\alpha')=d({\alpha^*})-2$. Let
   $\psi':=\psi^0_{{\alpha^*}}(\psi^0_\alpha)^{-1}\psi^0$. Then $\psi'\in K$ with
   $\alpha_{\psi'}=\alpha'$. By the inductive hypothesis,
   $\psi'\in\Lambda_{\phi,\chi}$. Furthermore,
   $\psi^0\in\Lambda_{\phi,\chi}$ and
   $\psi^0_\alpha\in\Lambda_{\phi,\chi}$, the latter because
   $\alpha\in\alpha^0\pm3$. Thus
   $\psi^0_{{\alpha^*}}=\psi'(\psi^0)^{-1}\psi^0_\alpha\in\Lambda_{\phi,\psi}$.

   \case2
   There are distinct $u_1,u_2,u_3\in N_+^D(t)$ such that
    ${\alpha^*}(u_1)=\alpha^0(u_2)$ and
    ${\alpha^*}(u_2)=\alpha^0(u_3)$.

    Let $\alpha_1,\alpha_2$ be the bijections from $N_+^D(t)$ to $N_+^{D'}(t')$
    defined by
    \begin{align*}
      \alpha_1(u_1)&:=\alpha^0(u_2),\\
      \alpha_1(u_2)&:=\alpha^0(u_3),\\
      \alpha_1(u_3)&:=\alpha^0(u_1),\\
      \alpha_1(u)&:=\alpha^0(u)&\text{for all }u\in N_+^D(t)\setminus\{u_1,u_2,u_3\},\\
     \alpha_2(u_1)&:=\alpha^0(u_3),\\
      \alpha_2(u_2)&:=\alpha^0(u_2),\\
      \alpha_2(u_3)&:=\alpha^0(u_1),\\
      \alpha_1(u)&:=\alpha^0(u)&\text{for all }u\in N_+^D(t)\setminus\{u_1,u_2,u_3\}.
     \end{align*}
     To prove that $\alpha_1$ is admissible, we need to prove that
     $\Lambda(u,\alpha_1(u))_{[\phi,\chi]}\neq\emptyset$ for all $u\in
     N_+^D(t)$. For~$u\neq u_3$ this is obvious since for such a~$u$ we have~$\alpha_1(u) = \alpha^0(u)$ or~$\alpha_1(u)= {\alpha^*}(u)$ and both~$\alpha^0$ and~${\alpha^*}$ are admissible. For~$u = u_3$ we observe that the two sequences~$u_3,u_2,u_1$ and~$\alpha^0(u_3),\alpha^0(u_2),\alpha^0(u_1)$ satisfy the assumptions on the sequences  in~Property\ref{closure:prop} and we therefore conclude~$\Lambda(u_3,\alpha^0(u_1))_{[\phi,\chi]}\neq\emptyset$.
     
     The proof that $\alpha_2$ is admissible is similar. 
     
     Since both $\alpha_1$ and $\alpha_2$ differ from $\alpha^0$ in at
     most $3$ places, we have $\alpha_1,\alpha_2\in\alpha^0\pm3$.

     Let $\alpha'={\alpha^*}\alpha_1^{-1}\alpha_2$, which is admissible by Property~\ref{closure:prop}. Then
     $\alpha'(u_1)=\alpha^0(u_3)$, $\alpha'(u_2)=\alpha^0(u_2)$, and
     $\alpha'(u)={\alpha^*}(u)$ for all $u\notin\{u_1,u_2\}$. Thus 
     $d(\alpha')=d({\alpha^*})-1$. Let
   $\psi':=\psi^0_{{\alpha^*}}(\psi^0_{\alpha_1})^{-1}\psi^0_{\alpha_2}$. 
   Then %
   $\alpha_{\psi'}=\alpha'$.  Furthermore,
   $\psi^0_{\alpha_1},\psi^0_{\alpha_2}\in\Lambda_{\phi,\chi}$. 
   Thus it suffices now  show that~$\psi'\in K$, since this implies $\psi'\in\Lambda_{\phi,\chi}$ by the inductive hypothesis, which implies
  $\psi^0_{{\alpha^*}}=\psi'(\psi^0_{\alpha_2})^{-1}\psi^0_{\alpha_1} \in\Lambda_{\phi,\psi}$. 
   
   To see that~$\psi'\in K$ we need to show that~$\psi'|_u \in\Lambda(u,{\alpha^*}(u))_{[\phi,\chi]}$ for all~$u\in N_+^D(t)$.
   
   Note that by the definition of~$\alpha_1$ and~$\alpha_2$, for all~$u'\in  N_+^{D'}(t')\setminus \{{\alpha^*}(u_1),{\alpha^*}(u_2)\}$ we have~$\alpha_1^{-1}(u') = \alpha_2^{-1}(u')$.
   
   Thus if~$u\in  N_+^{D}(t)\setminus \{u_1,u_2\}$ then
   \[\psi'|_u = \psi^0_{{\alpha^*}}|_{u}(\psi^0_{\alpha_1}|_{\alpha_1^{-1}({\alpha^*}(u))})^{-1}\psi^0_{\alpha_2}|_{\alpha_1^{-1}({\alpha^*}(u))} = \psi^0_{{\alpha^*}}|_{u} \in \Lambda(u,{\alpha^*}(u))_{[\phi,\chi]}.\]
   
   On the other hand if~$u \in \{u_1,u_2\}$ then~$\psi^0_{{\alpha^*}}|_{u} = \psi^0_{\alpha_1}|_{u}$ and thus 
   \[\psi'|_{u}   = \psi^0_{{\alpha^*}}|_{u}(\psi^0_{\alpha_1}|_{u})^{-1}\psi^0_{\alpha_2}|_{u} = \psi^0_{\alpha_2}|_{u} \in \Lambda(u,\alpha_2(u))_{[\phi,\chi]}.\]

   \end{cs}
   This shows that~$\psi^0_{{\alpha^*}}\in \Lambda_{\phi,\chi}$ for all admissible~${\alpha^*}$ and thus in particular~$\psi^0_{\alpha_\psi}\in \Lambda_{\phi,\chi}$, finishing the proof of the claim.\uend
\end{claim}
   
Since~$\Lambda_{\phi,\chi}$ can be computed in polynomial time for all
choices of~$\phi$ and~$\chi$ the theorem follows.
\end{proof}

\begin{proof}[Proof of Theorem~\ref{theo:coset}]
Using dynamic programming and the previous lemma, we can compute for all~$t \in D$ and~$t'\in D$ a coset~$\Lambda(t,t')$ satisfying~$\ISO((G_t)_{{D_t,\gamma_t}}, ({G'}_{t'})_{{D'_{t'},\gamma'_{t'}}}) \subseteq \Lambda(t,t') \subseteq \ISO (G_t,G'_{t'})$.
We let $\Lambda:=\Lambda_{r,r'}$ for the roots $r,r'$ of $D,D'$, respectively.
\end{proof}

\begin{theo}\label{theo:main}
For every $k\in\NN$ there is a polynomial time algorithm that, given
graphs $G$ and $G'$ of rank width at most $k$, computes the set $\ISO
(G,G')$ of all isomorphisms from~$G$ to~$G'$.
\end{theo}

\begin{proof}
By Theorem~\ref{theo:candec} we can compute for~$G$ and~$G'$ canonical treelike decompositions~$(D,\gamma)$ and~$(D',\gamma')$ of width at most~$a(k)$. 
By Lemma~\ref{lem:normal} we can assume that these decompositions are normal.
Since these decompositions are canonical,~$\ISO((G)_{{D,\gamma}}, ({G'})_{{D,\gamma'}})= \ISO (G,G')$.
The theorem now follows directly fromTheorem~\ref{theo:coset}.
\end{proof}

\section{Conclusions}
For every fixed $k$ we obtain a polynomial time isomorphism test for
graph classes of bounded rank width, unfortunately with a horrible
running time: we only have a non-elementary upper bound (in terms of
$k$) for the degree
of the polynomial bounding the running time. Thus before even asking whether
the isomorphism problem problem is fixed-parameter tractable if
parameterized by rank-width, we ask for an algorithm with a running
time $n^{O(k)}$. The bottleneck is the bound
we obtain for the size of a triple cover of a tangle (see
Lemma~\ref{lem:triple-cover}); our algorithm has to enumerate all
triple covers of all maximal tangles. But maybe there is a way to
avoid this. 

Our algorithm uses the group theoretic machinery, but the group theory
involved is fairly elementary. It seems conceivable that it can be
avoided altogether and there is a combinatorial algorithm deciding
isomorphism of rank width at most $k$. Specifically, we ask whether
for any $k$ there is an $\ell$ such that the $\ell$-dimensional
Weisfeiler-Lehman algorithm decides isomorphism of graphs of rank
width at most $k$.

Most of the arguments that we use in the construction of canonical
bounded width decompositions apply to arbitrary connectivity functions
and not just the cut rank function. (Only from
Section~\ref{sec:nwl} onwards we use specific properties of the cut
rank function.) It is an interesting question whether there is a
polynomial time isomorphism test for arbitrary connectivity functions
of bounded branch width. Even if this is not the case, it would be
interesting to understand for which connectivity functions beyond the
cut rank function such an isomorphism test exists.

In the end, the main question is whether our results help to solve the
isomorphism problem for general graphs. The immediate answer is
`no'. However, we do believe that structural techniques such as those
developed here (and also in \cite{gro12+a,gromar15}), in combination
with group theoretic techniques, may help to design graph isomorphism test with
an improved worst-case running time.

% \bibliographystyle{plain} 
% \bibliography{rw}

\begin{thebibliography}{10}

\bibitem{babgrimou82}
L.~Babai, D.Yu. Grigoryev, and D.M. Mount.
\newblock Isomorphism of graphs with bounded eigenvalue multiplicity.
\newblock In {\em Proceedings of the 14th ACM Symposium on Theory of
  Computing}, pages 310--324, 1982.

\bibitem{bod90}
H.L. Bodlaender.
\newblock Polynomial algorithms for graph isomorphism and chromatic index on
  partial {$k$}-trees.
\newblock {\em Journal of Algorithms}, 11:631--643, 1990.

\bibitem{boothcolbourn}
Kellogg~S. Booth and C.~J. Colbourn.
\newblock Problems polynomially equivalent to graph isomorphism.
\newblock Technical Report CS-77-04, Comp. Sci. Dep., Univ. Waterloo, 1979.

\bibitem{Corneil1981163}
D.~G. Corneil, H.~Lerchs, and L.~{Stewart Burlingham}.
\newblock Complement reducible graphs.
\newblock {\em Discrete Applied Mathematics}, 3(3):163--174, 1981.

\bibitem{coumakrot01}
B.~Courcelle, J.A. Makowsky, and U.~Rotics.
\newblock On the fixed-parameter complexity of graph enumeration problems
  definable in monadic second-order logic.
\newblock {\em Discrete Applied Mathematics}, 108(1--2):23--52, 2001.

\bibitem{couola00}
B.~Courcelle and S.~Olariu.
\newblock Upper bounds to the clique-width of graphs.
\newblock {\em Discrete Applied Mathematics}, 101:77--114, 2000.

\bibitem{CurLinMcC13}
Andrew Curtis, Min Lin, Ross McConnell, Yahav Nussbaum, Francisco Soulignac,
  Jeremy Spinrad, and Jayme Szwarcfiter.
\newblock Isomorphism of graph classes related to the circular-ones property.
\newblock {\em Discrete Mathematics and Theoretical Computer Science},
  15(1):157--182, 2013.

\bibitem{espegurwan01}
W.~Espelage, F.~Gurski, and E.~Wanke.
\newblock How to solve {NP}-hard graph problems on clique-width bounded graphs
  in polnomial time.
\newblock In A.~Brandst{\"a}dt and V.~Le, editors, {\em Proceedings of the 27th
  Workshop on Graph-Theoretic Concepts in Computer Science}, volume 2204 of
  {\em Lecture Notes in Computer Science}, pages 117--128. Springer-Verlag,
  2001.

\bibitem{filmay80}
I.~S. Filotti and J.~N. Mayer.
\newblock A polynomial-time algorithm for determining the isomorphism of graphs
  of fixed genus.
\newblock In {\em Proceedings of the 12th ACM Symposium on Theory of
  Computing}, pages 236--243, 1980.

\bibitem{fismakrav08}
E.~Fischer, J.A. Makowsky, and E.V. Ravve.
\newblock Counting truth assignments of formulas of bounded tree-width or
  clique-width.
\newblock {\em Discrete Applied Mathematics}, 156(4):511 -- 529, 2008.

\bibitem{geegerwhi09}
J.~Geelen, B.~Gerards, and G.~Whittle.
\newblock Tangles, tree-decompositions and grids in matroids.
\newblock {\em Journal of Combinatorial Theory, Series B}, 99(4):657--667,
  2009.

\bibitem{gro12+a}
M.~Grohe.
\newblock Descriptive complexity, canonisation, and definable graph structure
  theory.
\newblock Manuscript available at
  \url{http://www.lii.rwth-aachen.de/de/mitarbeiter/13-mitarbeiter/professoren/39-book-descriptive-complexity.html}.

\bibitem{gro08a}
M.~Grohe.
\newblock Definable tree decompositions.
\newblock In {\em Proceedings of the 23rd IEEE Symposium on Logic in Computer
  Science}, pages 406--417, 2008.

\bibitem{gromar15}
M.~Grohe and D.~Marx.
\newblock Structure theorem and isomorphism test for graphs with excluded
  topological subgraphs.
\newblock {\em {SIAM} Journal on Computing}, 44(1):114--159, 2015.

\bibitem{grosch15}
M.~Grohe and P.~Schweitzer.
\newblock Computing with tangles.
\newblock {\em ArXiv}, arXiv:1503.00190 [cs.DM], 2015.
\newblock Conference version to appear in STOC'15 Proceedings.

\bibitem{hoptar72}
J.E. Hopcroft and R.~Tarjan.
\newblock Isomorphism of planar graphs (working paper).
\newblock In R.~E. Miller and J.~W. Thatcher, editors, {\em Complexity of
  Computer Computations}. Plenum Press, 1972.

\bibitem{iwaflefuj01}
S.~Iwata, L.~Fleischer, and S.~Fujishige.
\newblock A combinatorial strongly polynomial algorithm for minimizing
  submodular functions.
\newblock {\em Journal of the ACM}, 48(4):761--777, 2001.

\bibitem{kobrot03}
D.~Kobler and U.~Rotics.
\newblock Edge dominating set and colorings on graphs with fixed clique-width.
\newblock {\em Discrete Applied Mathematics}, 126(2--3):197 -- 221, 2003.

\bibitem{KoblerKV13}
Johannes K{\"{o}}bler, Sebastian Kuhnert, and Oleg Verbitsky.
\newblock Helly circular-arc graph isomorphism is in logspace.
\newblock In Krishnendu Chatterjee and Jiri Sgall, editors, {\em Mathematical
  Foundations of Computer Science 2013 - 38th International Symposium}, volume
  8087 of {\em Lecture Notes in Computer Science}, pages 631--642. Springer,
  2013.

\bibitem{KratschS12}
Stefan Kratsch and Pascal Schweitzer.
\newblock Graph isomorphism for graph classes characterized by two forbidden
  induced subgraphs.
\newblock In {\em Graph-Theoretic Concepts in Computer Science - 38th
  International Workshop}, volume 7551 of {\em Lecture Notes in Computer
  Science}, pages 34--45. Springer, 2012.

\bibitem{LimouzyMR07}
Vincent Limouzy, Fabien de~Montgolfier, and Micha{\"{e}}l Rao.
\newblock {NLC-2} graph recognition and isomorphism.
\newblock In Andreas Brandst{\"{a}}dt, Dieter Kratsch, and Haiko M{\"{u}}ller,
  editors, {\em Graph-Theoretic Concepts in Computer Science, 33rd
  International Workshop, {WG} 2007, Dornburg, Germany, June 21-23, 2007.
  Revised Papers}, volume 4769 of {\em Lecture Notes in Computer Science},
  pages 86--98. Springer, 2007.

\bibitem{lokpilpil+14}
D.~Lokshtanov, M.~Pilipczuk, M.~Pilipczuk, and S.~Saurabh.
\newblock Fixed-parameter tractable canonization and isomorphism test for
  graphs of bounded treewidth.
\newblock In {\em Proceedings of the 55th Annual IEEE Symposium on Foundations
  of Computer Science}, pages 186--195, 2014.

\bibitem{LuekerB79}
G.S. Lueker and K.S. Booth.
\newblock A linear time algorithm for deciding interval graph isomorphism.
\newblock {\em Journal of the ACM}, 26(2):183--195, 1979.

\bibitem{luk82}
E.M. Luks.
\newblock Isomorphism of graphs of bounded valance can be tested in polynomial
  time.
\newblock {\em Journal of Computer and System Sciences}, 25:42--65, 1982.

\bibitem{mil80}
G.~L. Miller.
\newblock Isomorphism testing for graphs of bounded genus.
\newblock In {\em Proceedings of the 12th ACM Symposium on Theory of
  Computing}, pages 225--235, 1980.

\bibitem{oumsey06}
S.-I. Oum and P.D. Seymour.
\newblock Approximating clique-width and branch-width.
\newblock {\em Journal of Combinatorial Theory, Series B}, 96:514--528, 2006.

\bibitem{oxl11}
J.~Oxley.
\newblock {\em Matroid Theory}.
\newblock Cambridge University Press, 2nd edition, 2011.

\bibitem{pon88}
I.~N. Ponomarenko.
\newblock The isomorphism problem for classes of graphs that are invariant with
  respect to contraction.
\newblock {\em Zap. Nauchn. Sem. Leningrad. Otdel. Mat. Inst. Steklov. (LOMI)},
  174(Teor. Slozhn. Vychisl. 3):147--177, 182, 1988.
\newblock In Russian.

\bibitem{gm10}
N.~Robertson and P.D. Seymour.
\newblock Graph minors~{X}. {O}bstructions to tree-decomposition.
\newblock {\em Journal of Combinatorial Theory, Series B}, 52:153--190, 1991.

\bibitem{oumsey07}
S.-I-Oum and P.~Seymour.
\newblock Testing branch-width.
\newblock {\em Journal of Combinatorial Theory, Series B}, 97:385--393, 2007.

\bibitem{schrijver00}
A.~Schrijver.
\newblock A combinatorial algorithm minimizing submodular functions in strongly
  polynomial time.
\newblock {\em Journal of Combinatorial Theory, Series B}, 80(2):346--355,
  2000.

\bibitem{Schweitzer15}
Pascal Schweitzer.
\newblock Towards an isomorphism dichotomy for hereditary graph classes.
\newblock In {\em 32nd International Symposium on Theoretical Aspects of
  Computer Science}, volume~30 of {\em LIPIcs}, pages 689--702, 2015.

\bibitem{seress2003permutation}
{\'A}kos Seress.
\newblock {\em Permutation Group Algorithms}.
\newblock Cambridge Tracts in Mathematics. Cambridge University Press, 2003.

\end{thebibliography}

\end{document}